\documentclass[10pt]{article}

\usepackage[utf8]{inputenc}

\usepackage{amsmath}
\usepackage{amsfonts}
\usepackage{amssymb}
\usepackage{amsthm}

\usepackage{graphicx}
\usepackage{float}
\usepackage{wrapfig}
\usepackage{subcaption}
\usepackage[font=small]{caption}

\usepackage{lipsum}
\usepackage{xcolor}
\usepackage{enumitem}
\usepackage{booktabs}
\usepackage{siunitx}
\usepackage{url}

\usepackage{algorithm}
\usepackage{algpseudocode}

\usepackage{tikz}
\usetikzlibrary{positioning,fit,calc,arrows.meta}

\usepackage[
    a4paper,
    left=1.50cm,
    right=1.50cm,
    top=1.50cm,
    bottom=1.50cm
]{geometry}

\usepackage[numbers,sort&compress]{natbib}
\usepackage[hidelinks]{hyperref}
\usepackage{cleveref}

\newcounter{aim}

\crefname{aim}{Aim}{Aims}
\Crefname{aim}{Aim}{Aims}

\newlength{\aimleftindent}
\newlength{\aimlabelwidth}
\newlength{\aimvspace}

\newcommand{\aim}[2][]{%
  \par\addvspace{\aimvspace}%
  \refstepcounter{aim}%
  \if\relax\detokenize{#1}\relax
  \else
    \label{#1}%
  \fi
  \begingroup
    \noindent
    \hangindent=\dimexpr\aimleftindent+\aimlabelwidth\relax
    \hangafter=1
    \hspace*{\aimleftindent}%
    \textbf{\color{black!85}Aim~\theaim:}\enspace
    #2%
    \par
  \endgroup
  \addvspace{\aimvspace}%
}

\newcounter{challenge}
\renewcommand{\thechallenge}{\Roman{challenge}}

\crefname{challenge}{Challenge}{Challenges}
\Crefname{challenge}{Challenge}{Challenges}

\newlength{\challengelabelwidth}

\newcommand{\challenge}[2][]{%
  \par\addvspace{\aimvspace}%
  \refstepcounter{challenge}%
  \if\relax\detokenize{#1}\relax
  \else
    \label{#1}%
  \fi
  \begingroup
    \noindent
    \hangindent=\dimexpr\aimleftindent+\challengelabelwidth\relax
    \hangafter=1
    \hspace*{\aimleftindent}%
    \textbf{\color{black!85}Challenge~\thechallenge:}\enspace
    #2%
    \par
  \endgroup
  \addvspace{\aimvspace}%
}

\newcommand{\R}{\mathbb{R}}
\newcommand{\Rd}{\mathbb{R}^d}
\newcommand{\N}{\mathbb{N}}

\newcommand{\E}{\mathbb{E}}
\newcommand{\rmd}{\mathrm{d}}

\newcommand{\Mcal}{\mathcal{M}}
\newcommand{\Pcal}{\mathcal{P}}
\newcommand{\Wcal}{\mathcal{W}}
\newcommand{\cP}{\Pcal}
\newcommand{\Ecal}{\mathcal{E}}
\newcommand{\Acal}{\mathcal{A}}
\newcommand{\Cpl}{{\rm Cpl}}
\newcommand{\MCov}{{\rm MCov}}

\renewcommand{\phi}{\varphi}

\newcommand{\Law}{\operatorname{Law}}
\newcommand{\dom}{\operatorname{dom}}
\newcommand{\co}{\operatorname{co}}
\newcommand{\supp}{\operatorname{supp}}
\newcommand{\pr}{\operatorname{pr}}

\newcommand{\blfootnote}[1]{%
  \begingroup
  \renewcommand{\thefootnote}{}%
  \footnote{#1}%
  \addtocounter{footnote}{-1}%
  \endgroup
}

\newtheorem{theorem}{Theorem}
\numberwithin{theorem}{section}

\newtheorem{lemma}[theorem]{Lemma}
\newtheorem{proposition}[theorem]{Proposition}
\newtheorem{corollary}[theorem]{Corollary}

\theoremstyle{definition}
\newtheorem{definition}[theorem]{Definition}

\crefname{theorem}{Theorem}{Theorems}
\Crefname{theorem}{Theorem}{Theorems}

\crefname{lemma}{Lemma}{Lemmas}
\Crefname{lemma}{Lemma}{Lemmas}

\crefname{proposition}{Proposition}{Propositions}
\Crefname{proposition}{Proposition}{Propositions}

\crefname{corollary}{Corollary}{Corollaries}
\Crefname{corollary}{Corollary}{Corollaries}

\crefname{definition}{Definition}{Definitions}
\Crefname{definition}{Definition}{Definitions}

\crefname{example}{Example}{Examples}
\Crefname{example}{Example}{Examples}

\crefname{remark}{Remark}{Remarks}
\Crefname{remark}{Remark}{Remarks}

\AddToHook{env/lemma/begin}{%
  \crefalias{theorem}{lemma}%
}

\AddToHook{env/proposition/begin}{%
  \crefalias{theorem}{proposition}%
}

\AddToHook{env/corollary/begin}{%
  \crefalias{theorem}{corollary}%
}

\AddToHook{env/definition/begin}{%
  \crefalias{theorem}{definition}%
}

\AddToHook{env/example/begin}{%
  \crefalias{theorem}{example}%
}

\AddToHook{env/remark/begin}{%
  \crefalias{theorem}{remark}%
}

\author{Manuel Hasenbichler}

\title{Stochastic Knothe--Rosenblatt:
Light-speed Calibration of Stochastic Local Volatility Models}

\begin{document}
	
	\begin{center}
{\Large Stochastic Knothe--Rosenblatt: \\  \vspace{1em} Light-speed Calibration of Stochastic Local Volatility Models}

		\vspace{1em}
		{\large M.\ Beiglböck, M.\ Hasenbichler, G.\ Pammer
        }
		\blfootnote{\textit{Funding.}  MB
		gratefully acknowledges financial support from FWF through P34743,
		P35197, and FW506064 and from OeNB through AB1898311. All authors are grateful to Valentin
		Tissot-Daguette for many insightful discussions and recommendations.
        
    {In the appendix we report a proof of convergence of SKR to SLV which was established with very substantial support by ChatGPT 6.0.}}
	\end{center}

	\begin{center}
		\rule{150mm}{0.2mm}
	\end{center}		

	\begin{abstract}
	
	European option smiles determine the risk-neutral marginal laws of an asset, but not their intertemporal coupling, which is decisive for many applications. 
    The Bass martingale construction selects, among all calibrated martingales, the one closest to Bachelier dynamics; it permits fast calibration at discrete maturities and recovers the Dupire local-volatility (LV) model as the maturity grid is refined.
    
    This article develops a modular calibration overlay for existing stochastic and path-dependent volatility models.
    We sequentially construct a martingale that matches all prescribed marginals exactly while remaining as close as possible, in an adapted Knothe--Rosenblatt sense, to the reference dynamics. 
    As with the Bass LV model, each calibration step is amenable to an efficient Martingale Sinkhorn algorithm.
    We develop the theoretical foundations, numerical implementation and consider convergence to the SLV model.  We also benchmark the method for Heston and Bergomi finite-factor path-dependent volatility dynamics and develop the multi-asset extension.

	\medskip

	\noindent\textit{Keywords:} Bass martingale;  martingale Sinkhorn algorithm; calibration; stochastic local
	volatility; Dupire local volatility; martingale optimal transport;
	weak optimal transport; adapted Wasserstein distance; Knothe--Rosenblatt
	coupling.

	\noindent\textit{MSC2020 subject classifications:} Primary 60G42, 49Q22,
	91G20; secondary 60G44, 91G60.
	\end{abstract}

	\begin{center}
		\rule{150mm}{0.2mm}
	\end{center}		

	\vspace{5mm}

\section{Introduction}

\subsection{Outline}

Optimal transport provides variational tools for model calibration when the risk-neutral marginal laws $\mu_0, \ldots, \mu_n$ of the asset price are available at finitely many maturities $0= T_0 < T_1< \ldots < T_n$ ~\cite{TaTo13,BeJu16,GuLoWa19,GuLo21,NuWiZh23,Gu24,BePaSc22,Zh24,LiNuShTi26}. 
A prominent example is the Bass LV model~\cite{CoHe21}. It solves the Bachelier \emph{minimal distortion calibration} problem 
\begin{equation}\label{eq:bass.MDC}\tag{Bass-MDC}
    \inf_{\rmd X_t = \sigma_t^X\rmd B_t, X_{t_i} \sim\mu_i}
    \E \Big[
    \int_0^{T_n} |\sigma_t^X - \sigma^{\text{Bach}}|^2\,\rmd t
    \Big],
\end{equation}
and is closest to the Bachelier model, subject to the marginal constraints.
Crucially, (Bass-MDC)
can be \emph{separated} into $n$ independent martingale transport problems, which allows for efficient numerics based on the \emph{Martingale Sinkhorn algorithm}~\cite{CoHe21}.
For many applications, however, one would like to calibrate relative to richer reference models.

{
\paragraph{Challenge (MDC).} Given a reference model $R$, possibly driven by stochastic factors,
construct a  model $X$

1. \ that exactly reproduces the prescribed marginals $\mu_0,\ldots,\mu_n$, and

2. \ matches the dynamics of the reference model $R$ as closely as possible.
}
\bigskip

To obtain a numerically feasible approach to this challenge, a tempting simplification is to mimic $R$ on each interval $[T_i,T_{i+1}]$ and then concatenate the corresponding solutions. However, this effectively enforces stepwise Markovianity. 
In the regular diffusion setting, any continuous strong Markov martingale limit with the prescribed marginals is the Dupire LV model~\cite{Du94}, just as for the much simpler Bass model; cf.~\cite{Lo08,Lo09}.
To preserve the reference dynamics while retaining tractable computation, we pursue the following alternative.

\paragraph{Stochastic Knothe--Rosenblatt approach to MDC.}

We propose a \emph{Stochastic Knothe--Rosenblatt} (SKR) construction, which builds the calibrated model forward in time. 
Its modular design allows quant desks to retain their existing stochastic volatility models as references and add SKR as an exact calibration overlay.
At each maturity, we retain the joint law of the calibrated model and the reference for the next calibration step. 
This preserves the dependence structure while reducing the optimization to single-period problems. 
These are solved using a fast and robust Martingale Sinkhorn algorithm that is simple to implement~\cite{HaPaTh26,HaJoLoObPa25}. 
 
\medskip

We develop the theoretical foundations
and benchmark the method for Heston, Bergomi, and finite-factor path-dependent volatility dynamics. 
Under the regularity assumptions stated in Appendix~\ref{app:skr:sec:convergence}, the SKR model approaches the SLV model associated with the same reference model $R$ as the maturity grid becomes finer. In particular, the SKR construction provides an efficient and stable numerical scheme for SLV models, in the same way that the Bass construction does for the Dupire LV model.

\subsection{Existence and uniqueness of SKR} 

Let $\mathsf E$ be a Polish space and let $R=(Y,L)$ be a Markov process on $\R\times\mathsf E$. Throughout we assume that for every starting point $r\in \R\times \mathsf E$, the first coordinate $Y$ is a square integrable  martingale, with $\Law(Y_t\mid R_s=r)$  continuous for $s<t$.  Let $\delta_{x_0}=\mu_0\leq_c \ldots \leq_c \mu_n$ and assume that all marginals have finite second moment.

Suppose that calibration has reached $T_i$. The state passed to the next step is $\kappa_i:=\operatorname{Law}(X_{T_i},L_{T_i})$; its first marginal is the prescribed law $\mu_i$. Starting from $\kappa_i$, we calibrate to the next marginal $\mu_{i+1}$ by choosing the dynamics of $X$ as close as possible to the dynamics of $Y$. We formulate this as minimizing the squared difference of the volatilities of $X$ and $Y$:
\begin{equation}\label{eq:bridge}\tag{Bridge}
    V_i(\kappa_i,\mu_{i+1}) := \min_{X\in \Acal(\kappa_i, \mu_{i+1})} \E[[X-Y]_{T_{i+1}}- [X-Y]_{T_{i}}],
\end{equation}
where $\Acal(\kappa_i,\mu_{i+1})$ consists of square-integrable
$(\mathcal F_t)$-martingales $X$ satisfying $X_{T_i}=Y_{T_i}$ and
$X_{T_{i+1}}\sim\mu_{i+1}$. We work with the usual augmentation
$(\mathcal F_t)$ of the natural filtration generated by $R=(Y,L)$,
started with law $\kappa_i$ at $T_i$, and assume that $R$ is Markov
with respect to this filtration.

Thus, for each $T_i$ the reference process is started in $\kappa_i$. The bridge starts from the same random state as the reference process and minimizes the quadratic variation accumulated by $X-Y$ over $[T_i,T_{i+1}]$.\footnote{
When $X$ and $Y$ can be represented as It\^o processes
$    \rmd X_t=\sigma_t^X\,\rmd W_t $,  $\rmd Y_t=\sigma_t^Y\,\rmd W_t$, the optimization objective
can also be written in the form  $
    \mathbb E\Big[ [X-Y]_{T_{i+1}}-[X-Y]_{T_i} \Big] = \mathbb E\Big[ \int_{T_i}^{T_{i+1}} \lvert \sigma_t^X-\sigma_t^Y\rvert^2\,dt \Big],
$ as in \eqref{eq:bass.MDC}. However, the formulation in \eqref{eq:bridge} is technically more convenient, in particular in view of applications to non-continuous models.}

The joint evolution of the optimizer $X$ and the reference model then determines $\kappa_{i+1} = \mathrm{Law}(X_{T_{i+1}},L_{T_{i+1}})$, which serves as the initial condition for the next step. 
By constructing $\kappa_i$ sequentially from the preceding state, our approach follows the \emph{Knothe--Rosenblatt} paradigm from classical transport: building complex couplings through tractable conditional transports~\cite{CaGaSa10,BePaPo23}.

\begin{figure}[H]
\centering

\begin{tikzpicture}[
        scale=0.94,
        transform shape,
        every node/.style={font=\normalsize},
        reference/.style={
            line width=0.75pt,
            -{Stealth[length=2.5mm]}
        },
        transport/.style={
            line width=0.4pt
        },
        transporthead/.style={
            line width=0.4pt,
            -{Triangle[open,length=2.2mm,width=2.2mm]}
        }
    ]
    
    \everymath{\displaystyle}
    
    \coordinate (c0) at (0,0);
    \coordinate (c1) at (3.6,0);
    \coordinate (c2) at (7.2,0);
    \coordinate (cn) at (11.5,0);
    
    \foreach \c in {c0,c1,c2,cn}{
        \draw[rounded corners=3pt]
            ($(\c)+(-0.70,-0.9)$)
            rectangle
            ($(\c)+(0.70,1.06)$);
    }
    
    \node (z0) at ($(c0)+(0,0.75)$) {$L_{T_0}$};
    \node (s0) at ($(c0)+(0,0.05)$) {$X_{T_0}$};
    \node at ($(c0)+(0,-0.3)$) {\rotatebox{90}{$\sim$}};
    \node at ($(c0)+(0,-0.64)$) {$\mu_0$};
    
    \node (z1) at ($(c1)+(0,0.75)$) {$L_{T_1}$};
    \node (s1) at ($(c1)+(0,0.05)$) {$X_{T_1}$};
    \node at ($(c1)+(0,-0.3)$) {\rotatebox{90}{$\sim$}};
    \node at ($(c1)+(0,-0.64)$) {$\mu_1$};
    
    \node (z2) at ($(c2)+(0,0.75)$) {$L_{T_2}$};
    \node (s2) at ($(c2)+(0,0.05)$) {$X_{T_2}$};
    \node at ($(c2)+(0,-0.3)$) {\rotatebox{90}{$\sim$}};
    \node at ($(c2)+(0,-0.64)$) {$\mu_2$};
    
    \node (zn) at ($(cn)+(0,0.75)$) {$L_{T_n}$};
    \node (sn) at ($(cn)+(0,0.05)$) {$X_{T_n}$};
    \node at ($(cn)+(0,-0.3)$) {\rotatebox{90}{$\sim$}};
    \node at ($(cn)+(0,-0.64)$) {$\mu_n$};
    
    \coordinate (r0out) at ($(z0.east)+(-0.75mm,0.10)$);
    \coordinate (r1in) at ($(z1.west)+(0.75mm,0.10)$);
    
    \coordinate (z0out) at ($(z0.east)+(-0.75mm,0)$);
    \coordinate (s0out) at ($(s0.east)+(-0.75mm,0)$);
    \coordinate (s1in) at ($(s1.west)+(0.75mm,0)$);
    \coordinate (merge01) at ($(s1in)+(-3.75mm,0)$);
    
    \draw[reference]
        (r0out)
        to[bend left=13]
        node[above,pos=0.50] {$\mathrm dR$}
        (r1in);
    
    \draw[transport]
        (z0out) -- (merge01);
    
    \draw[transport]
        (s0out) -- (merge01);
    
    \draw[transporthead]
        (merge01) -- (s1in);
    
    \node at ($(s0out)!0.52!(s1in)+(0,-0.44)$)
        {\eqref{eq:bridge}};
    
    \coordinate (r1out) at ($(z1.east)+(-0.75mm,0.10)$);
    \coordinate (r2in) at ($(z2.west)+(0.75mm,0.10)$);
    
    \coordinate (z1out) at ($(z1.east)+(-0.75mm,0)$);
    \coordinate (s1out) at ($(s1.east)+(-0.75mm,0)$);
    \coordinate (s2in) at ($(s2.west)+(0.75mm,0)$);
    \coordinate (merge12) at ($(s2in)+(-3.75mm,0)$);
    
    \draw[reference]
        (r1out)
        to[bend left=13]
        node[above,pos=0.50] {$\mathrm dR$}
        (r2in);
    
    \draw[transport]
        (z1out) -- (merge12);
    
    \draw[transport]
        (s1out) -- (merge12);
    
    \draw[transporthead]
        (merge12) -- (s2in);
    
    \node at ($(s1out)!0.52!(s2in)+(0,-0.44)$)
        {\eqref{eq:bridge}};
    
    \coordinate (r2out) at ($(z2.east)+(-0.75mm,0.10)$);
    \coordinate (z2out) at ($(z2.east)+(-0.75mm,0)$);
    \coordinate (s2out) at ($(s2.east)+(-0.75mm,0)$);
    
    \coordinate (rnin) at ($(cn)+(-0.35,0.85)$);
    \coordinate (snin) at ($(cn)+(-0.32,0.05)$);
    
    \coordinate (merge2n) at ($(snin)+(-3.7mm,0)$);
    
    \path
        (r2out)
        to[bend left=8.5]
        coordinate[pos=0.50] (rdot)
        (rnin);
    
    \coordinate (mid) at ($(r2out)!0.50!(rnin)$);
    \coordinate (vbot) at ($(mid)+(0,-2)$);
    \coordinate (vtop) at ($(mid)+(0,2)$);
    
    \coordinate (zdot)
        at (intersection of z2out--merge2n and vbot--vtop);
    
    \coordinate (sdot)
        at (intersection of s2out--merge2n and vbot--vtop);
    
    \draw[reference]
        (r2out)
        to[bend left=8.5]
        (rnin);
    
    \draw[transport]
        (z2out) -- (merge2n);
    
    \draw[transport]
        (s2out) -- (merge2n);
    
    \draw[transporthead]
        (merge2n) -- (snin);

    \fill[white]
        ($(mid)+(-4mm,-1.10)$)
        rectangle
        ($(mid)+(4.05mm,0.55)$);
    
    \foreach \dx in {-1.5mm,0mm,1.5mm}{
        \fill ($([xshift=\dx]rdot)$) circle[radius=0.55pt];
    }

    \fill ($(zdot)!1.5mm!(z2out)$)   circle[radius=0.55pt];
    \fill (zdot)                     circle[radius=0.55pt];
    \fill ($(zdot)!1.5mm!(merge2n)$) circle[radius=0.55pt];
    
    \foreach \dx in {-1.5mm,0mm,1.5mm}{
        \fill ($([xshift=\dx]sdot)$) circle[radius=0.55pt];
    }

\end{tikzpicture}
\medskip

\caption{Minimal Distortion Calibration via Stochastic Knothe--Rosenblatt. At each maturity, the propagated state has law $\kappa_i=\operatorname{Law}(X_{T_i},L_{T_i})$, whose first marginal is $\mu_i$.}
\label{fig:SKR}
\end{figure}
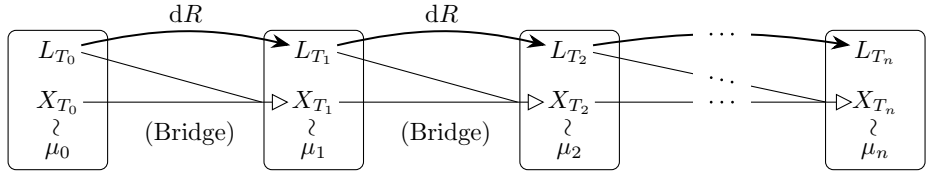

To solve~\eqref{eq:bridge}, a key observation is that this \emph{dynamic} problem admits an equivalent \emph{static} formulation as a \emph{weak martingale optimal transport} (WMOT) problem~\cite{BeJoMaPa23,CaMaSy25}.
Unlike in classical martingale transport, where the cost is assigned to a pair of points $(x,y)$, WMOT costs may depend on $x$ and the terminal law of a bridge started at $x$.

Under the above assumptions and the regularity conditions specified in
\Cref{sec:bridge.wmot} for each calibration interval, we obtain

\begin{theorem}[Existence and uniqueness of the SKR construction]
\label{thm:skr}
The one-period optimizers exist and are unique in law. The stochastic Knothe--Rosenblatt model exists and is unique in law. 
\end{theorem}

\subsection{Bass-structure}
\label{sec:bass.sinkhorn}

Fix one calibration interval and rescale it to $[0,1]$. Let $\mu$, $\nu$ be probabilities on $\R$ with finite second moments, $\mu\leq_c\nu$, where $(\mu,\nu)$ is irreducible, i.e.\ the call potentials $c_\mu$, $c_\nu$ satisfy that $I:=\{c_{\mu}<c_{\nu}\}$ is an interval and $\mu(I)=1$. Let $\kappa$ be a probability on $\R\times \mathsf E$ with first marginal $\mu$ and for $r \in \R \times \mathsf E$, let
\[
    \gamma_{r}:=\operatorname{Law}(Y_1\mid R_0 = r)
\]
be the conditional terminal law of the reference martingale. We assume that $\gamma_r$ is  continuous for every $r$ and that the map $r\mapsto \gamma_r$ is $\Wcal_2$ continuous. 

Starting the reference process $(Y,L)$ in $\kappa$,  the unit interval bridge problem is to find the martingale $X$ solving
\begin{align}\label{eq:Ubridge}
\inf_{X: X_0= Y_0, X_1\sim \nu} \E[X-Y]_1.
\end{align}
We will give a structural characterization in the spirit of Bass martingales.
To this end we introduce some notation.

Given a function $g:\R \to [-\infty,\infty]$ and a probability $\gamma$ on $\R$ we denote the `reflected' convolution by
$g \star \gamma (a):= \int g(a+y)\, d\gamma(y).$
If $g$ is increasing and $\gamma$ is atom-free, then $g\star \gamma$ is increasing and continuous on the interval
where it is finite. We use $(g\star \gamma)^{-1}$ to denote an inverse of $g\star \gamma$.
\begin{definition}[$R$-Bass martingale]
\label{def:R.Bass}
A martingale $X$
is called an
$R$-Bass martingale if for some  increasing
$g:\R\to[-\infty,\infty]$ and $A:=a(R_0):= (g\star \gamma_{R_0})^{-1}(X_0)$
\begin{equation}\label{eq:R.Bass}
    X_t
    =
    \E\bigl[g(A+Y_1) \,\big|\, \sigma(A,R_s:0\leq s\leq t)\bigr]
    =
    \E[g(A+Y_1)\mid A,R_t],
    \qquad 0\leq t\leq1.
\end{equation}
The second equality follows from the Markov property of the reference
process and the fact that $A=a(R_0)$ is measurable at the initial time.
\end{definition}
We briefly comment on \Cref{def:R.Bass}.
First, note that the martingale condition $X_0=\E[X_1|X_0, L_0]=
\E[X_1|R_0]$ amounts to 
\[
    X_0 = \E[g(a(R_0) +Y_1)|R_0] 
    \quad \Longleftrightarrow \quad 
    x_0 = (g\star \gamma_{r_0})(a(r_0)), \ {\kappa\text{-a.s.}}
\]
This explains the choice $a(r_0)= (g\star \gamma_{r_0})^{-1}(x_0)$.
In the definition of $R$-Bass martingales we tacitly assume that $(g\star \gamma_{r_0})^{-1}(x_0)$ exists $\kappa$-a.s. While the inverse used here is not necessarily unique, {it is straightforward to see that $g(A+Y_1)$ does not depend on the specific choice of the inverse.}

We observe also that shifting the function $g$ to the right by a constant $c$ amounts to adding $c$ to the respective $A$. This degree of freedom will be used to normalize the shifts.

\begin{theorem}[Bass characterization]
\label{thm:R.Bass.characterization}
Under the standing assumptions, an admissible martingale $X$ from $\kappa$ to
$\nu$ solves the unit interval bridge problem \eqref{eq:Ubridge} 
{if and only if it is an $R$-Bass martingale.}
\end{theorem}

\definecolor{refblue}{RGB}{31,86,150}
\definecolor{trajred}{RGB}{125,30,30}
\definecolor{mapgreen}{RGB}{32,110,55}
\definecolor{refblue}{RGB}{31,86,150}
\definecolor{trajred}{RGB}{125,30,30}
\definecolor{mapgreen}{RGB}{32,110,55}

\begin{figure}[H]
\centering
\begin{tikzpicture}[
        scale=0.68,
        axis/.style   ={line width=0.7pt, color=refblue, -{Stealth[length=2.4mm]}},
        traj/.style   ={line width=0.5pt, color=trajred},
        dens/.style   ={line width=0.7pt, color=refblue, fill=refblue!8},
        blob/.style   ={line width=0.7pt, color=refblue, fill=refblue!12},
        slice/.style  ={line width=0.5pt, color=black!40, fill=black!6},
        fibre/.style  ={line width=0.6pt, color=refblue},
        gfibre/.style ={line width=0.5pt, color=black!40},
        map/.style    ={line width=0.8pt, color=mapgreen, -{Stealth[length=2.8mm]}},
        axlab/.style  ={font=\small, color=refblue},
        meas/.style   ={font=\small, color=refblue},
        dynlab/.style ={font=\small, color=trajred},
        maplab/.style ={font=\small, color=mapgreen}
    ]

    \begin{scope}[shift={(0,4.800)}]

    \draw[axis] (0.000,-1.180) -- (0.000,2.150);
    \draw[axis] (6.450,-1.180) -- (6.450,2.150);
    \draw[axis] (0.532,-0.049) -- (-1.199,-1.440);
    \node[axlab] at (0,2.150) [above] {$A+Y_0$};
    \node[axlab] at (6.450,2.150) [above left=-1pt] {$A+Y_1$};
    \node[axlab] at (-1.199,-1.440) [below left=-2pt and -1pt] {$L_0$};

    \draw[gfibre] (-0.874,-1.179) -- (-0.874,0.025);
    \draw[slice] plot coordinates {
        (-0.874,-1.049) (-0.843,-1.040) (-0.838,-1.032) (-0.832,-1.024) (-0.826,-1.015) (-0.818,-1.007) 
        (-0.810,-0.999) (-0.802,-0.991) (-0.792,-0.982) (-0.782,-0.974) (-0.771,-0.966) (-0.759,-0.957) 
        (-0.747,-0.949) (-0.734,-0.941) (-0.721,-0.932) (-0.707,-0.924) (-0.693,-0.916) (-0.679,-0.907) 
        (-0.665,-0.899) (-0.651,-0.891) (-0.637,-0.882) (-0.624,-0.874) (-0.612,-0.866) (-0.600,-0.857) 
        (-0.589,-0.849) (-0.580,-0.841) (-0.571,-0.832) (-0.564,-0.824) (-0.559,-0.816) (-0.555,-0.807) 
        (-0.553,-0.799) (-0.552,-0.791) (-0.553,-0.783) (-0.556,-0.774) (-0.561,-0.766) (-0.566,-0.758) 
        (-0.574,-0.749) (-0.582,-0.741) (-0.592,-0.733) (-0.602,-0.724) (-0.614,-0.716) (-0.626,-0.708) 
        (-0.638,-0.699) (-0.651,-0.691) (-0.664,-0.683) (-0.676,-0.674) (-0.688,-0.666) (-0.700,-0.658) 
        (-0.711,-0.649) (-0.721,-0.641) (-0.731,-0.633) (-0.739,-0.624) (-0.746,-0.616) (-0.752,-0.608) 
        (-0.756,-0.599) (-0.760,-0.591) (-0.761,-0.583) (-0.762,-0.575) (-0.761,-0.566) (-0.759,-0.558) 
        (-0.755,-0.550) (-0.750,-0.541) (-0.744,-0.533) (-0.737,-0.525) (-0.728,-0.516) (-0.719,-0.508) 
        (-0.709,-0.500) (-0.698,-0.491) (-0.686,-0.483) (-0.674,-0.475) (-0.662,-0.466) (-0.650,-0.458) 
        (-0.638,-0.450) (-0.625,-0.441) (-0.614,-0.433) (-0.603,-0.425) (-0.592,-0.416) (-0.582,-0.408) 
        (-0.574,-0.400) (-0.566,-0.391) (-0.560,-0.383) (-0.554,-0.375) (-0.551,-0.366) (-0.548,-0.358) 
        (-0.548,-0.350) (-0.548,-0.342) (-0.550,-0.333) (-0.554,-0.325) (-0.559,-0.317) (-0.566,-0.308) 
        (-0.573,-0.300) (-0.582,-0.292) (-0.592,-0.283) (-0.603,-0.275) (-0.614,-0.267) (-0.626,-0.258) 
        (-0.639,-0.250) (-0.652,-0.242) (-0.666,-0.233) (-0.679,-0.225) (-0.693,-0.217) (-0.706,-0.208) 
        (-0.719,-0.200) (-0.732,-0.192) (-0.744,-0.183) (-0.756,-0.175) (-0.767,-0.167) (-0.777,-0.158) 
        (-0.787,-0.150) (-0.797,-0.142) (-0.805,-0.134) (-0.813,-0.125) (-0.821,-0.117) (-0.827,-0.109) 
        (-0.833,-0.100) (-0.839,-0.092) (-0.843,-0.084) (-0.874,-0.075)
    } -- cycle;
    \draw[gfibre] (-0.569,-0.933) -- (-0.569,0.170);
    \draw[slice] plot coordinates {
        (-0.569,-0.911) (-0.533,-0.903) (-0.531,-0.895) (-0.529,-0.886) (-0.526,-0.878) (-0.524,-0.870) 
        (-0.522,-0.861) (-0.520,-0.853) (-0.518,-0.845) (-0.515,-0.837) (-0.513,-0.828) (-0.511,-0.820) 
        (-0.508,-0.812) (-0.506,-0.803) (-0.504,-0.795) (-0.501,-0.787) (-0.499,-0.778) (-0.496,-0.770) 
        (-0.494,-0.762) (-0.492,-0.753) (-0.489,-0.745) (-0.487,-0.737) (-0.485,-0.728) (-0.483,-0.720) 
        (-0.480,-0.712) (-0.478,-0.703) (-0.476,-0.695) (-0.474,-0.687) (-0.473,-0.678) (-0.471,-0.670) 
        (-0.469,-0.662) (-0.467,-0.653) (-0.466,-0.645) (-0.464,-0.637) (-0.463,-0.628) (-0.462,-0.620) 
        (-0.460,-0.612) (-0.459,-0.604) (-0.458,-0.595) (-0.457,-0.587) (-0.455,-0.579) (-0.454,-0.570) 
        (-0.453,-0.562) (-0.451,-0.554) (-0.450,-0.545) (-0.448,-0.537) (-0.446,-0.529) (-0.444,-0.520) 
        (-0.441,-0.512) (-0.439,-0.504) (-0.435,-0.495) (-0.432,-0.487) (-0.427,-0.479) (-0.423,-0.470) 
        (-0.417,-0.462) (-0.411,-0.454) (-0.404,-0.445) (-0.396,-0.437) (-0.388,-0.429) (-0.379,-0.420) 
        (-0.369,-0.412) (-0.358,-0.404) (-0.346,-0.396) (-0.334,-0.387) (-0.321,-0.379) (-0.307,-0.371) 
        (-0.293,-0.362) (-0.278,-0.354) (-0.262,-0.346) (-0.247,-0.337) (-0.232,-0.329) (-0.216,-0.321) 
        (-0.201,-0.312) (-0.186,-0.304) (-0.172,-0.296) (-0.159,-0.287) (-0.146,-0.279) (-0.135,-0.271) 
        (-0.125,-0.262) (-0.117,-0.254) (-0.110,-0.246) (-0.105,-0.237) (-0.101,-0.229) (-0.100,-0.221) 
        (-0.100,-0.212) (-0.103,-0.204) (-0.107,-0.196) (-0.114,-0.188) (-0.122,-0.179) (-0.132,-0.171) 
        (-0.143,-0.163) (-0.156,-0.154) (-0.170,-0.146) (-0.186,-0.138) (-0.203,-0.129) (-0.220,-0.121) 
        (-0.238,-0.113) (-0.256,-0.104) (-0.275,-0.096) (-0.294,-0.088) (-0.313,-0.079) (-0.331,-0.071) 
        (-0.350,-0.063) (-0.367,-0.054) (-0.385,-0.046) (-0.401,-0.038) (-0.417,-0.029) (-0.431,-0.021) 
        (-0.445,-0.013) (-0.458,-0.004) (-0.470,0.004) (-0.481,0.012) (-0.492,0.021) (-0.501,0.029) 
        (-0.509,0.037) (-0.517,0.045) (-0.524,0.054) (-0.530,0.062) (-0.569,0.070)
    } -- cycle;

    \draw[traj] plot coordinates {
        (0.280,0.349) (0.378,0.359) (0.475,0.501) (0.573,0.629) (0.670,0.585) (0.768,0.562) (0.866,0.517) 
        (0.963,0.580) (1.061,0.581) (1.159,0.662) (1.256,0.485) (1.354,0.647) (1.451,0.645) (1.549,0.719) 
        (1.647,0.712) (1.744,0.681) (1.842,0.734) (1.940,0.822) (2.037,0.809) (2.135,0.800) (2.232,0.875) 
        (2.330,0.795) (2.428,0.651) (2.525,0.697) (2.623,0.637) (2.720,0.453) (2.818,0.379) (2.916,0.339) 
        (3.013,0.227) (3.111,0.085) (3.209,0.096) (3.306,0.192) (3.404,0.175) (3.501,0.108) (3.599,0.153) 
        (3.697,0.231) (3.794,0.208) (3.892,0.268) (3.990,0.379) (4.087,0.365) (4.185,0.291) (4.282,0.332) 
        (4.380,0.363) (4.478,0.479) (4.575,0.358) (4.673,0.299) (4.770,0.222) (4.868,0.057) (4.966,0.076) 
        (5.063,0.135) (5.161,0.068) (5.259,0.213) (5.356,0.301) (5.454,0.370) (5.551,0.417) (5.649,0.518) 
        (5.747,0.393) (5.844,0.460) (5.942,0.527) (6.040,0.358) (6.137,0.399) (6.235,0.381) (6.332,0.465) 
        (6.430,0.428)
    };
    \draw[traj] plot coordinates {
        (0.280,0.349) (0.378,0.348) (0.475,0.382) (0.573,0.296) (0.670,0.356) (0.768,0.346) (0.866,0.396) 
        (0.963,0.345) (1.061,0.453) (1.159,0.514) (1.256,0.497) (1.354,0.560) (1.451,0.686) (1.549,0.865) 
        (1.647,0.709) (1.744,0.797) (1.842,0.844) (1.940,0.836) (2.037,0.736) (2.135,0.862) (2.232,0.737) 
        (2.330,0.794) (2.428,0.924) (2.525,0.765) (2.623,0.736) (2.720,0.606) (2.818,0.631) (2.916,0.782) 
        (3.013,0.984) (3.111,0.808) (3.209,0.751) (3.306,0.822) (3.404,0.980) (3.501,1.022) (3.599,0.949) 
        (3.697,0.979) (3.794,0.978) (3.892,0.958) (3.990,0.886) (4.087,0.925) (4.185,0.956) (4.282,0.947) 
        (4.380,0.926) (4.478,0.799) (4.575,0.751) (4.673,0.872) (4.770,0.853) (4.868,0.711) (4.966,0.844) 
        (5.063,0.897) (5.161,1.108) (5.259,1.114) (5.356,1.069) (5.454,0.926) (5.551,1.058) (5.649,1.314) 
        (5.747,1.233) (5.844,1.169) (5.942,1.229) (6.040,1.147) (6.137,1.121) (6.235,1.087) (6.332,1.106) 
        (6.430,1.216)
    };
    \draw[traj] plot coordinates {
        (0.280,0.349) (0.378,0.346) (0.475,0.432) (0.573,0.384) (0.670,0.412) (0.768,0.193) (0.866,0.044) 
        (0.963,0.118) (1.061,0.053) (1.159,0.106) (1.256,0.154) (1.354,0.280) (1.451,0.355) (1.549,0.451) 
        (1.647,0.434) (1.744,0.359) (1.842,0.281) (1.940,0.227) (2.037,0.109) (2.135,0.050) (2.232,0.035) 
        (2.330,0.054) (2.428,0.015) (2.525,-0.181) (2.623,-0.194) (2.720,-0.177) (2.818,-0.075) 
        (2.916,-0.023) (3.013,-0.092) (3.111,-0.170) (3.209,0.025) (3.306,0.095) (3.404,0.271) 
        (3.501,0.477) (3.599,0.391) (3.697,0.424) (3.794,0.464) (3.892,0.514) (3.990,0.562) (4.087,0.577) 
        (4.185,0.589) (4.282,0.434) (4.380,0.412) (4.478,0.332) (4.575,0.339) (4.673,0.287) (4.770,0.343) 
        (4.868,0.419) (4.966,0.445) (5.063,0.470) (5.161,0.474) (5.259,0.424) (5.356,0.402) (5.454,0.348) 
        (5.551,0.381) (5.649,0.377) (5.747,0.429) (5.844,0.292) (5.942,0.374) (6.040,0.293) (6.137,0.215) 
        (6.235,0.189) (6.332,0.128) (6.430,0.151)
    };

    \draw[fibre] (0.280,-0.251) -- (0.280,0.819);
    \draw[blob] plot coordinates {
        (0.280,-0.021) (0.318,-0.013) (0.323,-0.005) (0.328,0.004) (0.333,0.012) (0.340,0.020) 
        (0.346,0.029) (0.353,0.037) (0.361,0.045) (0.369,0.054) (0.378,0.062) (0.387,0.070) (0.397,0.079) 
        (0.408,0.087) (0.419,0.095) (0.430,0.104) (0.442,0.112) (0.455,0.120) (0.468,0.128) (0.482,0.137) 
        (0.496,0.145) (0.510,0.153) (0.525,0.162) (0.540,0.170) (0.555,0.178) (0.570,0.187) (0.585,0.195) 
        (0.600,0.203) (0.615,0.212) (0.630,0.220) (0.644,0.228) (0.658,0.237) (0.672,0.245) (0.684,0.253) 
        (0.696,0.262) (0.708,0.270) (0.718,0.278) (0.728,0.287) (0.736,0.295) (0.744,0.303) (0.750,0.312) 
        (0.755,0.320) (0.759,0.328) (0.761,0.336) (0.763,0.345) (0.763,0.353) (0.761,0.361) (0.759,0.370) 
        (0.755,0.378) (0.750,0.386) (0.744,0.395) (0.736,0.403) (0.728,0.411) (0.718,0.420) (0.708,0.428) 
        (0.696,0.436) (0.684,0.445) (0.672,0.453) (0.658,0.461) (0.644,0.470) (0.630,0.478) (0.615,0.486) 
        (0.600,0.495) (0.585,0.503) (0.570,0.511) (0.555,0.520) (0.540,0.528) (0.525,0.536) (0.510,0.545) 
        (0.496,0.553) (0.482,0.561) (0.468,0.569) (0.455,0.578) (0.442,0.586) (0.430,0.594) (0.419,0.603) 
        (0.408,0.611) (0.397,0.619) (0.387,0.628) (0.378,0.636) (0.369,0.644) (0.361,0.653) (0.353,0.661) 
        (0.346,0.669) (0.340,0.678) (0.333,0.686) (0.328,0.694) (0.323,0.703) (0.318,0.711) (0.280,0.719)
    } -- cycle;
    \node[axlab] at (0.280,-0.251) [below right=-1pt and -1pt] {$\ell$};
    \node[meas] at (-1.074,-0.579) [left] {};
    \node[meas, anchor=south west, inner sep=1pt] at (0.300,1.178) {$\Law(A+Y_0\mid L_0=\ell)$};

    \draw[dens] 
        (6.450,-0.500)
        -- (6.745,-0.500)
        -- plot[smooth] coordinates {
        (6.755,-0.492) (6.765,-0.484) (6.775,-0.475) (6.786,-0.467) (6.797,-0.459) 
        (6.807,-0.451) (6.818,-0.443) (6.829,-0.435) (6.840,-0.426) (6.851,-0.418) (6.863,-0.410) 
        (6.874,-0.402) (6.885,-0.394) (6.897,-0.385) (6.908,-0.377) (6.920,-0.369) (6.931,-0.361) 
        (6.943,-0.353) (6.954,-0.344) (6.966,-0.336) (6.977,-0.328) (6.988,-0.320) (7.000,-0.312) 
        (7.011,-0.304) (7.022,-0.295) (7.033,-0.287) (7.044,-0.279) (7.054,-0.271) (7.065,-0.263) 
        (7.075,-0.254) (7.085,-0.246) (7.095,-0.238) (7.105,-0.230) (7.114,-0.222) (7.124,-0.214) 
        (7.133,-0.205) (7.141,-0.197) (7.150,-0.189) (7.158,-0.181) (7.166,-0.173) (7.173,-0.164) 
        (7.181,-0.156) (7.188,-0.148) (7.194,-0.140) (7.200,-0.132) (7.206,-0.123) (7.212,-0.115) 
        (7.217,-0.107) (7.222,-0.099) (7.226,-0.091) (7.230,-0.083) (7.234,-0.074) (7.237,-0.066) 
        (7.240,-0.058) (7.243,-0.050) (7.245,-0.042) (7.247,-0.033) (7.248,-0.025) (7.249,-0.017) 
        (7.250,-0.009) (7.250,-0.001) (7.250,0.007) (7.250,0.016) (7.249,0.024) (7.248,0.032) 
        (7.246,0.040) (7.245,0.048) (7.242,0.057) (7.240,0.065) (7.238,0.073) (7.235,0.081) (7.231,0.089) 
        (7.228,0.098) (7.224,0.106) (7.221,0.114) (7.217,0.122) (7.212,0.130) (7.208,0.138) (7.203,0.147) 
        (7.199,0.155) (7.194,0.163) (7.189,0.171) (7.184,0.179) (7.179,0.188) (7.174,0.196) (7.169,0.204) 
        (7.164,0.212) (7.159,0.220) (7.153,0.228) (7.148,0.237) (7.143,0.245) (7.138,0.253) (7.133,0.261) 
        (7.129,0.269) (7.124,0.278) (7.119,0.286) (7.115,0.294) (7.110,0.302) (7.106,0.310) (7.102,0.319) 
        (7.098,0.327) (7.095,0.335) (7.091,0.343) (7.088,0.351) (7.085,0.359) (7.082,0.368) (7.080,0.376) 
        (7.077,0.384) (7.075,0.392) (7.073,0.400) (7.072,0.409) (7.070,0.417) (7.069,0.425) (7.068,0.433) 
        (7.068,0.441) (7.067,0.449) (7.067,0.458) (7.067,0.466) (7.068,0.474) (7.068,0.482) (7.069,0.490) 
        (7.070,0.499) (7.071,0.507) (7.073,0.515) (7.074,0.523) (7.076,0.531) (7.078,0.540) (7.080,0.548) 
        (7.082,0.556) (7.085,0.564) (7.087,0.572) (7.090,0.580) (7.093,0.589) (7.096,0.597) (7.099,0.605) 
        (7.102,0.613) (7.105,0.621) (7.108,0.630) (7.111,0.638) (7.114,0.646) (7.117,0.654) (7.120,0.662) 
        (7.123,0.671) (7.126,0.679) (7.129,0.687) (7.132,0.695) (7.135,0.703) (7.137,0.711) (7.140,0.720) 
        (7.142,0.728) (7.144,0.736) (7.146,0.744) (7.148,0.752) (7.150,0.761) (7.152,0.769) (7.153,0.777) 
        (7.154,0.785) (7.155,0.793) (7.155,0.801) (7.156,0.810) (7.156,0.818) (7.156,0.826) (7.155,0.834) 
        (7.155,0.842) (7.154,0.851) (7.153,0.859) (7.151,0.867) (7.149,0.875) (7.147,0.883) (7.145,0.892) 
        (7.142,0.900) (7.139,0.908) (7.136,0.916) (7.132,0.924) (7.129,0.932) (7.124,0.941) (7.120,0.949) 
        (7.115,0.957) (7.110,0.965) (7.105,0.973) (7.099,0.982) (7.094,0.990) (7.088,0.998) (7.081,1.006) 
        (7.075,1.014) (7.068,1.022) (7.061,1.031) (7.053,1.039) (7.046,1.047) (7.038,1.055) (7.030,1.063) 
        (7.022,1.072) (7.014,1.080) (7.005,1.088) (6.997,1.096) (6.988,1.104) (6.979,1.113) (6.970,1.121) 
        (6.961,1.129) (6.952,1.137) (6.943,1.145) (6.933,1.153) (6.924,1.162) (6.914,1.170) (6.905,1.178) 
        (6.895,1.186) (6.885,1.194) (6.876,1.203) (6.866,1.211) (6.856,1.219) (6.847,1.227) (6.837,1.235) 
        (6.828,1.243) (6.818,1.252) (6.808,1.260) (6.799,1.268) (6.790,1.276) (6.780,1.284) (6.771,1.293) 
        (6.762,1.301) (6.753,1.309) (6.744,1.317) (6.735,1.325) (6.727,1.334) (6.718,1.342) (6.709,1.350) 
        (6.701,1.358) (6.693,1.366) (6.685,1.374) (6.677,1.383) (6.669,1.391) (6.662,1.399) (6.654,1.407) 
        (6.647,1.415) (6.640,1.424) (6.633,1.432) (6.626,1.440) (6.619,1.448) (6.613,1.456) (6.606,1.464) 
        (6.600,1.473) (6.594,1.481) (6.589,1.489) (6.583,1.497) (6.577,1.505) (6.572,1.514) (6.567,1.522) 
        (6.562,1.530) (6.557,1.538) (6.552,1.546) (6.548,1.555) (6.543,1.563) (6.539,1.571) (6.535,1.579) 
        (6.531,1.587) (6.527,1.595) (6.523,1.604) (6.520,1.612)
    }
    -- (6.51,1.620)
    -- (6.450,1.620)
    -- cycle;
    \node[meas] at (7.600,-0.001) {$\rho^{a}$};
    \node[dynlab] at (3.612,-0.794) {$Y$-dynamics};
    \end{scope}

    \begin{scope}

    \draw[axis] (0.000,-1.180) -- (0.000,2.150);
    \draw[axis] (6.450,-1.180) -- (6.450,2.150);
    \draw[axis] (0.532,0.001) -- (-1.199,-1.390);
    \node[axlab] at (0,2.150) [above] {$X_0$};
    \node[axlab] at (6.450,2.150) [above left=-1pt] {$X_1$};
    \node[axlab] at (-1.199,-1.390) [below left=-2pt and -1pt] {$L_0$};

    \draw[gfibre] (-0.874,-1.129) -- (-0.874,0.219);
    \draw[slice] plot coordinates {
        (-0.874,-1.118) (-0.852,-1.110) (-0.848,-1.101) (-0.844,-1.093) (-0.839,-1.084) (-0.834,-1.075) 
        (-0.828,-1.067) (-0.822,-1.058) (-0.815,-1.049) (-0.807,-1.040) (-0.799,-1.031) (-0.790,-1.022) 
        (-0.781,-1.013) (-0.771,-1.004) (-0.761,-0.994) (-0.750,-0.985) (-0.739,-0.976) (-0.728,-0.966) 
        (-0.717,-0.957) (-0.706,-0.947) (-0.695,-0.937) (-0.684,-0.928) (-0.674,-0.918) (-0.664,-0.908) 
        (-0.655,-0.898) (-0.646,-0.888) (-0.639,-0.878) (-0.632,-0.868) (-0.627,-0.858) (-0.622,-0.848) 
        (-0.619,-0.838) (-0.618,-0.827) (-0.617,-0.817) (-0.618,-0.807) (-0.620,-0.796) (-0.623,-0.786) 
        (-0.628,-0.775) (-0.634,-0.765) (-0.640,-0.754) (-0.648,-0.743) (-0.656,-0.733) (-0.665,-0.722) 
        (-0.675,-0.711) (-0.684,-0.700) (-0.694,-0.690) (-0.705,-0.679) (-0.715,-0.668) (-0.725,-0.657) 
        (-0.734,-0.646) (-0.743,-0.635) (-0.752,-0.624) (-0.760,-0.613) (-0.767,-0.602) (-0.773,-0.590) 
        (-0.779,-0.579) (-0.783,-0.568) (-0.787,-0.557) (-0.789,-0.546) (-0.791,-0.535) (-0.791,-0.524) 
        (-0.790,-0.512) (-0.789,-0.501) (-0.786,-0.490) (-0.782,-0.479) (-0.778,-0.467) (-0.773,-0.456) 
        (-0.766,-0.445) (-0.759,-0.434) (-0.752,-0.423) (-0.744,-0.411) (-0.735,-0.400) (-0.726,-0.389) 
        (-0.717,-0.378) (-0.707,-0.367) (-0.698,-0.356) (-0.688,-0.345) (-0.679,-0.334) (-0.670,-0.323) 
        (-0.662,-0.312) (-0.654,-0.301) (-0.647,-0.290) (-0.640,-0.279) (-0.635,-0.268) (-0.630,-0.257) 
        (-0.626,-0.246) (-0.624,-0.235) (-0.622,-0.225) (-0.622,-0.214) (-0.622,-0.203) (-0.624,-0.193) 
        (-0.627,-0.182) (-0.631,-0.171) (-0.636,-0.161) (-0.642,-0.150) (-0.649,-0.140) (-0.656,-0.130) 
        (-0.664,-0.119) (-0.673,-0.109) (-0.682,-0.099) (-0.692,-0.089) (-0.702,-0.079) (-0.712,-0.069) 
        (-0.723,-0.059) (-0.733,-0.049) (-0.743,-0.039) (-0.753,-0.029) (-0.763,-0.020) (-0.772,-0.010) 
        (-0.781,-0.000) (-0.790,0.009) (-0.798,0.019) (-0.806,0.028) (-0.813,0.037) (-0.820,0.047) 
        (-0.826,0.056) (-0.831,0.065) (-0.837,0.074) (-0.841,0.083) (-0.846,0.092) (-0.849,0.101) 
        (-0.853,0.110) (-0.874,0.119)
    } -- cycle;
    \draw[gfibre] (-0.569,-0.884) -- (-0.569,0.345);
    \draw[slice] plot coordinates {
        (-0.569,-1.013) (-0.542,-1.006) (-0.540,-0.999) (-0.539,-0.991) (-0.537,-0.984) (-0.535,-0.976) 
        (-0.534,-0.969) (-0.532,-0.961) (-0.530,-0.953) (-0.529,-0.946) (-0.527,-0.938) (-0.525,-0.930) 
        (-0.523,-0.922) (-0.521,-0.914) (-0.519,-0.906) (-0.517,-0.898) (-0.516,-0.889) (-0.514,-0.881) 
        (-0.512,-0.873) (-0.510,-0.864) (-0.508,-0.856) (-0.506,-0.847) (-0.504,-0.839) (-0.503,-0.830) 
        (-0.501,-0.821) (-0.499,-0.812) (-0.497,-0.803) (-0.496,-0.794) (-0.494,-0.785) (-0.493,-0.776) 
        (-0.491,-0.767) (-0.490,-0.758) (-0.489,-0.749) (-0.488,-0.739) (-0.486,-0.730) (-0.485,-0.721) 
        (-0.484,-0.711) (-0.483,-0.701) (-0.482,-0.692) (-0.482,-0.682) (-0.481,-0.672) (-0.480,-0.663) 
        (-0.480,-0.653) (-0.479,-0.643) (-0.479,-0.633) (-0.478,-0.623) (-0.478,-0.613) (-0.477,-0.602) 
        (-0.477,-0.592) (-0.476,-0.582) (-0.475,-0.572) (-0.475,-0.561) (-0.474,-0.551) (-0.473,-0.540) 
        (-0.471,-0.530) (-0.470,-0.519) (-0.468,-0.509) (-0.466,-0.498) (-0.463,-0.487) (-0.461,-0.476) 
        (-0.457,-0.466) (-0.454,-0.455) (-0.449,-0.444) (-0.444,-0.433) (-0.439,-0.422) (-0.433,-0.411) 
        (-0.427,-0.400) (-0.419,-0.389) (-0.412,-0.378) (-0.403,-0.367) (-0.394,-0.356) (-0.385,-0.345) 
        (-0.375,-0.334) (-0.365,-0.323) (-0.354,-0.312) (-0.343,-0.300) (-0.332,-0.289) (-0.321,-0.278) 
        (-0.310,-0.267) (-0.299,-0.256) (-0.288,-0.244) (-0.277,-0.233) (-0.268,-0.222) (-0.258,-0.211) 
        (-0.250,-0.199) (-0.243,-0.188) (-0.236,-0.177) (-0.231,-0.166) (-0.227,-0.155) (-0.224,-0.144) 
        (-0.222,-0.132) (-0.222,-0.121) (-0.223,-0.110) (-0.226,-0.099) (-0.230,-0.088) (-0.235,-0.077) 
        (-0.242,-0.066) (-0.249,-0.055) (-0.258,-0.044) (-0.268,-0.033) (-0.279,-0.022) (-0.291,-0.011) 
        (-0.303,-0.000) (-0.316,0.010) (-0.329,0.021) (-0.343,0.032) (-0.356,0.042) (-0.370,0.053) 
        (-0.384,0.064) (-0.397,0.074) (-0.410,0.085) (-0.423,0.095) (-0.435,0.105) (-0.447,0.116) 
        (-0.458,0.126) (-0.469,0.136) (-0.479,0.147) (-0.488,0.157) (-0.497,0.167) (-0.505,0.177) 
        (-0.513,0.187) (-0.519,0.197) (-0.525,0.206) (-0.531,0.216) (-0.536,0.226) (-0.540,0.236) 
        (-0.569,0.245)
    } -- cycle;

    \draw[traj] plot coordinates {
        (0.280,0.399) (0.378,0.412) (0.475,0.602) (0.573,0.769) (0.670,0.713) (0.768,0.683) (0.866,0.623) 
        (0.963,0.706) (1.061,0.707) (1.159,0.809) (1.256,0.581) (1.354,0.791) (1.451,0.788) (1.549,0.878) 
        (1.647,0.870) (1.744,0.833) (1.842,0.896) (1.940,0.996) (2.037,0.981) (2.135,0.972) (2.232,1.053) 
        (2.330,0.966) (2.428,0.796) (2.525,0.853) (2.623,0.779) (2.720,0.539) (2.818,0.439) (2.916,0.385) 
        (3.013,0.236) (3.111,0.058) (3.209,0.070) (3.306,0.190) (3.404,0.169) (3.501,0.085) (3.599,0.141) 
        (3.697,0.241) (3.794,0.211) (3.892,0.291) (3.990,0.439) (4.087,0.420) (4.185,0.320) (4.282,0.376) 
        (4.380,0.418) (4.478,0.574) (4.575,0.411) (4.673,0.332) (4.770,0.230) (4.868,0.023) (4.966,0.046) 
        (5.063,0.119) (5.161,0.037) (5.259,0.218) (5.356,0.334) (5.454,0.427) (5.551,0.490) (5.649,0.626) 
        (5.747,0.457) (5.844,0.548) (5.942,0.637) (6.040,0.411) (6.137,0.466) (6.235,0.442) (6.332,0.555) 
        (6.430,0.505)
    };
    \draw[traj] plot coordinates {
        (0.280,0.399) (0.378,0.397) (0.475,0.444) (0.573,0.327) (0.670,0.408) (0.768,0.395) (0.866,0.462) 
        (0.963,0.393) (1.061,0.539) (1.159,0.620) (1.256,0.597) (1.354,0.680) (1.451,0.839) (1.549,1.042) 
        (1.647,0.866) (1.744,0.969) (1.842,1.020) (1.940,1.011) (2.037,0.898) (2.135,1.038) (2.232,0.899) 
        (2.330,0.965) (2.428,1.102) (2.525,0.932) (2.623,0.898) (2.720,0.740) (2.818,0.771) (2.916,0.952) 
        (3.013,1.161) (3.111,0.980) (3.209,0.916) (3.306,0.996) (3.404,1.157) (3.501,1.196) (3.599,1.127) 
        (3.697,1.156) (3.794,1.155) (3.892,1.136) (3.990,1.063) (4.087,1.103) (4.185,1.134) (4.282,1.125) 
        (4.380,1.104) (4.478,0.970) (4.575,0.916) (4.673,1.049) (4.770,1.029) (4.868,0.868) (4.966,1.020) 
        (5.063,1.075) (5.161,1.270) (5.259,1.276) (5.356,1.238) (5.454,1.104) (5.551,1.228) (5.649,1.428) 
        (5.747,1.370) (5.844,1.321) (5.942,1.367) (6.040,1.303) (6.137,1.282) (6.235,1.253) (6.332,1.270) 
        (6.430,1.357)
    };
    \draw[traj] plot coordinates {
        (0.280,0.399) (0.378,0.394) (0.475,0.510) (0.573,0.446) (0.670,0.483) (0.768,0.193) (0.866,0.008) 
        (0.963,0.097) (1.061,0.019) (1.159,0.082) (1.256,0.143) (1.354,0.306) (1.451,0.407) (1.549,0.536) 
        (1.647,0.514) (1.744,0.413) (1.842,0.308) (1.940,0.237) (2.037,0.087) (2.135,0.015) (2.232,-0.002) 
        (2.330,0.021) (2.428,-0.025) (2.525,-0.234) (2.623,-0.246) (2.720,-0.229) (2.818,-0.125) 
        (2.916,-0.068) (3.013,-0.144) (3.111,-0.222) (3.209,-0.014) (3.306,0.069) (3.404,0.295) 
        (3.501,0.571) (3.599,0.455) (3.697,0.499) (3.794,0.553) (3.892,0.620) (3.990,0.683) (4.087,0.702) 
        (4.185,0.717) (4.282,0.513) (4.380,0.483) (4.478,0.376) (4.575,0.385) (4.673,0.316) (4.770,0.391) 
        (4.868,0.493) (4.966,0.527) (5.063,0.562) (5.161,0.567) (5.259,0.500) (5.356,0.471) (5.454,0.397) 
        (5.551,0.441) (5.649,0.436) (5.747,0.506) (5.844,0.322) (5.942,0.433) (6.040,0.324) (6.137,0.220) 
        (6.235,0.188) (6.332,0.110) (6.430,0.139)
    };

    \draw[fibre] (0.280,-0.201) -- (0.280,0.988);
    \draw[blob] plot coordinates {
        (0.280,-0.076) (0.310,-0.067) (0.314,-0.057) (0.318,-0.048) (0.322,-0.039) (0.326,-0.029) 
        (0.331,-0.019) (0.337,-0.010) (0.342,-0.000) (0.348,0.010) (0.355,0.020) (0.362,0.029) 
        (0.369,0.039) (0.377,0.049) (0.385,0.059) (0.393,0.070) (0.402,0.080) (0.411,0.090) (0.421,0.100) 
        (0.431,0.111) (0.441,0.121) (0.451,0.131) (0.462,0.142) (0.473,0.152) (0.483,0.163) (0.494,0.174) 
        (0.505,0.184) (0.516,0.195) (0.527,0.206) (0.537,0.216) (0.548,0.227) (0.558,0.238) (0.568,0.249) 
        (0.577,0.260) (0.586,0.271) (0.594,0.282) (0.602,0.293) (0.609,0.304) (0.615,0.315) (0.621,0.326) 
        (0.626,0.337) (0.630,0.348) (0.633,0.359) (0.636,0.371) (0.637,0.382) (0.638,0.393) (0.638,0.404) 
        (0.636,0.415) (0.634,0.427) (0.631,0.438) (0.628,0.449) (0.623,0.460) (0.618,0.471) (0.611,0.483) 
        (0.605,0.494) (0.597,0.505) (0.589,0.516) (0.580,0.527) (0.571,0.539) (0.562,0.550) (0.552,0.561) 
        (0.541,0.572) (0.531,0.583) (0.520,0.594) (0.509,0.605) (0.499,0.616) (0.488,0.627) (0.477,0.638) 
        (0.466,0.649) (0.455,0.660) (0.445,0.671) (0.435,0.682) (0.425,0.692) (0.415,0.703) (0.406,0.714) 
        (0.397,0.725) (0.388,0.735) (0.380,0.746) (0.372,0.756) (0.365,0.767) (0.358,0.777) (0.351,0.788) 
        (0.345,0.798) (0.339,0.808) (0.334,0.818) (0.329,0.829) (0.324,0.839) (0.320,0.849) (0.316,0.859) 
        (0.312,0.869) (0.308,0.879) (0.280,0.888)
    } -- cycle;
    \node[axlab] at (0.280,-0.201) [below right=-1pt and -1pt] {$\ell$};
    \node[meas] at (-1.074,-0.529) [left] {$\kappa$};
    \node[meas, anchor=south west, inner sep=1pt] at (0.300,1.340) {$\Law(X_0\mid L_0=\ell)$};

    \draw[dens] 
    (6.450,-0.500)
    -- (6.890,-0.500)
    -- plot[smooth] coordinates {
        (6.904,-0.494) (6.916,-0.488) (6.928,-0.482) (6.940,-0.476) (6.952,-0.470) 
        (6.964,-0.464) (6.975,-0.458) (6.987,-0.451) (6.999,-0.445) (7.010,-0.439) (7.022,-0.432) 
        (7.033,-0.426) (7.044,-0.420) (7.055,-0.413) (7.066,-0.407) (7.077,-0.400) (7.087,-0.393) 
        (7.098,-0.387) (7.108,-0.380) (7.117,-0.373) (7.127,-0.366) (7.136,-0.359) (7.145,-0.352) 
        (7.153,-0.345) (7.162,-0.338) (7.170,-0.331) (7.177,-0.324) (7.185,-0.317) (7.192,-0.309) 
        (7.198,-0.302) (7.204,-0.295) (7.210,-0.287) (7.216,-0.280) (7.221,-0.272) (7.225,-0.264) 
        (7.230,-0.257) (7.233,-0.249) (7.237,-0.241) (7.240,-0.233) (7.243,-0.225) (7.245,-0.217) 
        (7.247,-0.209) (7.248,-0.201) (7.249,-0.193) (7.250,-0.184) (7.250,-0.176) (7.250,-0.168) 
        (7.249,-0.159) (7.248,-0.151) (7.247,-0.142) (7.245,-0.133) (7.243,-0.125) (7.241,-0.116) 
        (7.238,-0.107) (7.235,-0.098) (7.232,-0.089) (7.228,-0.080) (7.224,-0.071) (7.220,-0.062) 
        (7.216,-0.053) (7.211,-0.043) (7.206,-0.034) (7.201,-0.025) (7.195,-0.015) (7.190,-0.006) 
        (7.184,0.004) (7.178,0.013) (7.172,0.023) (7.165,0.033) (7.159,0.043) (7.152,0.052) (7.146,0.062) 
        (7.139,0.072) (7.132,0.082) (7.125,0.092) (7.118,0.103) (7.112,0.113) (7.105,0.123) (7.098,0.133) 
        (7.091,0.144) (7.084,0.154) (7.077,0.164) (7.070,0.175) (7.064,0.185) (7.057,0.196) (7.051,0.206) 
        (7.044,0.217) (7.038,0.228) (7.032,0.238) (7.026,0.249) (7.020,0.260) (7.014,0.271) (7.009,0.281) 
        (7.004,0.292) (6.998,0.303) (6.994,0.314) (6.989,0.325) (6.984,0.336) (6.980,0.347) (6.976,0.358) 
        (6.972,0.369) (6.969,0.380) (6.966,0.391) (6.963,0.402) (6.960,0.413) (6.957,0.424) (6.955,0.435) 
        (6.953,0.446) (6.951,0.457) (6.950,0.468) (6.949,0.479) (6.948,0.490) (6.948,0.501) (6.947,0.512) 
        (6.947,0.523) (6.948,0.534) (6.948,0.545) (6.949,0.556) (6.950,0.567) (6.951,0.578) (6.953,0.589) 
        (6.955,0.599) (6.957,0.610) (6.959,0.621) (6.962,0.632) (6.965,0.643) (6.968,0.653) (6.971,0.664) 
        (6.975,0.675) (6.979,0.685) (6.983,0.696) (6.987,0.707) (6.991,0.717) (6.996,0.728) (7.001,0.738) 
        (7.006,0.748) (7.011,0.759) (7.016,0.769) (7.021,0.779) (7.027,0.790) (7.032,0.800) (7.038,0.810) 
        (7.044,0.820) (7.049,0.830) (7.055,0.840) (7.061,0.850) (7.067,0.860) (7.073,0.869) (7.079,0.879) 
        (7.085,0.889) (7.091,0.898) (7.097,0.908) (7.102,0.917) (7.108,0.927) (7.114,0.936) (7.119,0.946) 
        (7.125,0.955) (7.130,0.964) (7.136,0.973) (7.141,0.982) (7.146,0.991) (7.150,1.000) (7.155,1.009) 
        (7.159,1.018) (7.163,1.027) (7.167,1.035) (7.171,1.044) (7.175,1.053) (7.178,1.061) (7.181,1.070) 
        (7.183,1.078) (7.186,1.086) (7.188,1.094) (7.190,1.103) (7.191,1.111) (7.192,1.119) (7.193,1.127) 
        (7.194,1.135) (7.194,1.143) (7.194,1.151) (7.193,1.158) (7.192,1.166) (7.191,1.174) (7.189,1.181) 
        (7.187,1.189) (7.185,1.196) (7.182,1.204) (7.179,1.211) (7.176,1.218) (7.172,1.226) (7.168,1.233) 
        (7.164,1.240) (7.159,1.247) (7.154,1.254) (7.148,1.261) (7.142,1.268) (7.136,1.275) (7.130,1.281) 
        (7.123,1.288) (7.116,1.295) (7.108,1.301) (7.101,1.308) (7.093,1.315) (7.085,1.321) (7.076,1.328) 
        (7.068,1.334) (7.059,1.340) (7.050,1.347) (7.040,1.353) (7.031,1.359) (7.021,1.365) (7.011,1.371) 
        (7.001,1.377) (6.991,1.384) (6.981,1.390) (6.970,1.395) (6.960,1.401) (6.949,1.407) (6.939,1.413) 
        (6.928,1.419) (6.917,1.425) (6.907,1.430) (6.896,1.436) (6.885,1.442) (6.874,1.447) (6.864,1.453) 
        (6.853,1.458) (6.842,1.464) (6.831,1.469) (6.821,1.475) (6.810,1.480) (6.800,1.486) (6.790,1.491) 
        (6.779,1.496) (6.769,1.501) (6.759,1.507) (6.750,1.512) (6.740,1.517) (6.730,1.522) (6.721,1.527) 
        (6.712,1.532) (6.703,1.538) (6.694,1.543) (6.685,1.548) (6.676,1.553) (6.668,1.558) (6.660,1.562) 
        (6.652,1.567) (6.644,1.572) (6.636,1.577) (6.629,1.582) (6.621,1.587) (6.614,1.592) (6.607,1.596) 
        (6.601,1.601) (6.594,1.606) (6.588,1.611) (6.582,1.615) 
    } 
    -- (6.570,1.620)
    -- (6.450,1.620)
    -- cycle;
    \node[meas] at (7.600,-0.176) {$\nu$};
    \node[dynlab] at (3.612,-0.846) {$X$-dynamics};
    \end{scope}

    \draw[map] (8.150,5.100) to[bend left=20] (8.150,1.300);
    \node[maplab] at (8.900,3.150) {$g$};

    \draw[map] (-2.400,1.300) to[bend left=20] (-2.400,5.100);
    \node[maplab] at (-2.750,3.450) [left] {$A=\big(g\star\gamma_{R_0}\big)^{-1}(X_0)$};

\end{tikzpicture}
\medskip

\caption{An $R$-Bass martingale. The reference model is started in $\kappa$ and
shifted by $A=a(R_0)$; the increasing map $g$ carries the shifted terminal law
$\rho^{a}$ onto $\nu$, and $X_t=\E[g(A+Y_1)\mid A,R_t]$ is the resulting
martingale. The shift is pinned down by the martingale condition
$X_0=(g\star\gamma_{R_0})(A)$. Each fibre over the latent axis carries the
conditional law of the initial value given $L_0$; the paths are started on the
front fibre $L_0=\ell$.}
\label{fig:R.Bass}
\end{figure}
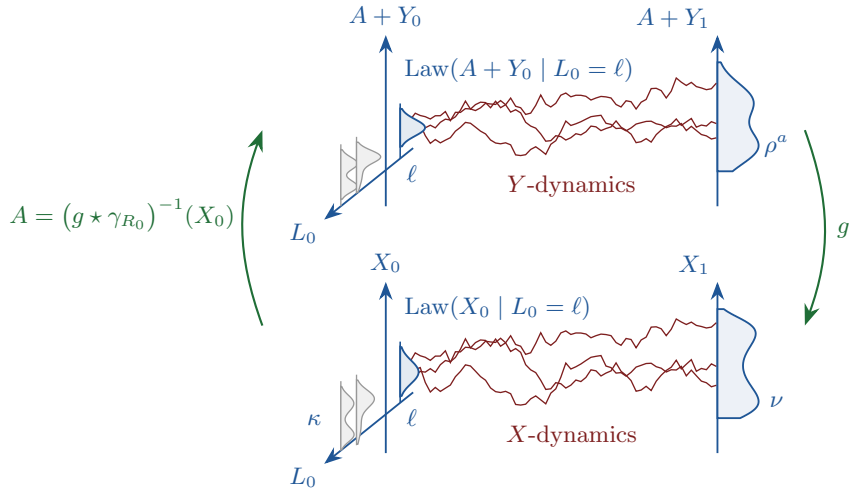

In the case where $R$ is Brownian motion, 
the construction underlying our one-period steps goes back to Bass~\cite{Ba83}, who
transported a Gaussian law onto a prescribed one by a monotone map in order to solve
the Skorokhod embedding problem. {While Bass was only concerned with the deterministic starting case, Bass martingales between general marginals were introduced by Backhoff-Veraguas--Beiglb\"ock--Huesmann--K\"allblad~\cite{BaBeHuKa20} and Guo--Loeper--Wang~\cite{GuLoWa19}.} Extensions beyond the Brownian case are considered in  \cite{CoHe21, Ts24, BePaRi24, HaPaTh26, AcMa26a, AcMa26b}; our formulation has substantial overlap with the concurrent work of Marini, see \cite{Ma26}. The $R$-Bass martingales of \Cref{def:R.Bass} continue this line of thought: we replace the Brownian increment by the conditional increment of a general Markovian reference model, and it is precisely this flexibility which allows the calibrated model to inherit the dynamics of $R$.

\subsection{The Martingale Sinkhorn algorithm}

The structural characterization of optimal bridges in \Cref{thm:R.Bass.characterization} forms the basis for a particularly simple, yet powerful numerical algorithm, which we introduce below. 
 For this it is useful to define for  a shift $a:\R\times \mathsf E\to\R$, 
the shifted reference terminal law
\begin{equation}\label{eq:shifted.reference.law}
    \rho^a
    := 
    \int
       (a(r)+\,\cdot\,)_{\#}\gamma_r\,
       \kappa(\rmd r) = \Law(a(R_0) + Y_1).
\end{equation}
The initial and terminal constraints of an $R$-Bass martingale can then be expressed concisely in terms of $a$ and $g$.
\medskip

\noindent {\textit{Terminal marginal constraint.}} $ g_\#(\rho^a)
    =\nu$.

    \medskip 
    
\noindent {\textit{Martingale constraint.}} $
    x=
    (g\star \gamma_r)(a(r))
    \quad\text{for $\kappa$-a.e. $r=(x,\ell)$}.$

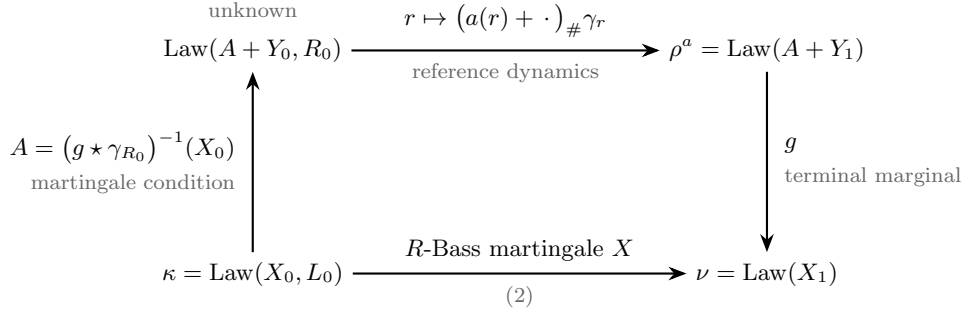
\begin{figure}
\centering
\begin{tikzpicture}[
        meas/.style ={font=\small},
        arr/.style  ={line width=0.8pt, -{Stealth[length=2.4mm]}},
        lab/.style  ={font=\small},
        note/.style ={font=\footnotesize, color=black!60}
    ]

    \node[meas] (s0) at (0,2.95)    {$\Law(A+Y_0, R_0)$};
    \node[meas] (s1) at (6.80,2.95) {$\rho^{a}=\Law(A+Y_1)$};
    \node[meas] (k)  at (0,0)       {$\kappa=\Law(X_0,L_0)$};
    \node[meas] (nu) at (6.80,0)    {$\nu=\Law(X_1)$};

    \node[note, above=1pt of s0] {unknown};

    \draw[arr] (s0) --
        node[lab,  above, yshift=0.4mm] {$r\mapsto\big(a(r)+\,\cdot\,\big)_{\#}\gamma_r$}
        node[note, below, yshift=-0.4mm]{reference dynamics} (s1);

    \draw[arr] (k) --
        node[lab,  above, yshift=0.4mm] {$R$-Bass martingale $X$}
        node[note, below, yshift=-0.4mm]{\eqref{eq:R.Bass}} (nu);

    \draw[arr] (k) --
        node[lab, left, align=right, xshift=-1mm]
            {$A=\big(g\star\gamma_{R_0}\big)^{-1}(X_0)$\\[1.5pt]
             {\footnotesize\color{black!60}martingale condition}} (s0);

    \draw[arr] (s1) --
        node[lab, right, align=left, xshift=1mm]
            {$g$\\[1.5pt]
             {\footnotesize\color{black!60}terminal marginal}} (nu);

\end{tikzpicture}
\medskip

\caption{The Bass diagram behind \Cref{def:R.Bass}, in the notation of
\Cref{fig:R.Bass}. The prescribed data are the marginals $\kappa$
and $\nu$; the unknown is the shift $a$, equivalently the initial law of the
\emph{shifted} reference process. The martingale condition determines $A$ from
the initial state, running the reference dynamics from $A+Y_0$ produces
$\rho^{a}$, and the increasing map $g$ carries $\rho^{a}$ onto $\nu$. Composing
the three outer arrows gives the $R$-Bass martingale along the bottom.}
\label{fig:bass.diagram}
\end{figure}
    \medskip 

This  observation suggests a natural
fixed-point algorithm. 
{Starting from a shift $a_0: \R\times \mathsf E \to \R$, perform the following two updates.}

\medskip

\noindent {\it Terminal-marginal update.}
Let $g_k$ be the increasing transport from $\rho^{a_k}$ to $\nu$, i.e.
 $
    g_k
    :=
    F_\nu^{-1}\circ F_{\rho^{a_k}},
$

\medskip

\noindent{\it Martingale update.}
{For $\kappa$-a.e.~$r=(x,\ell)$, set $a_{k+1}(x,\ell)= (g_k\star \gamma_{(x,\ell)})^{-1}(x)$.}

\medskip

Thus the first update enforces the terminal marginal and the second
enforces the martingale constraint.

\medskip

{In \Cref{thm:MSA.convergence} below we prove convergence of the Martingale Sinkhorn algorithm to the $R$-Bass martingale in the case where $(\mu, \nu) $ is strictly irreducible and $\nu$ is compactly supported, under the additional reference assumptions and normalizations stated there.} In numerical experiments, we observe linear convergence for generic reference processes. 

\medskip

For the classical Brownian Bass construction the above fixed point scheme was introduced by
Conze--Henry-Labord\`ere~\cite{CoHe21} who demonstrate its striking numerical properties for the approximation of the (Bass) LV model. In the Brownian case, linear convergence was proved by
Acciaio--Marini--Pammer~\cite{AcMaPa23}, and a {measure-preserving}, multidimensional
variant is given 
in~\cite{HaJoLoObPa25}. {In the case where the reference dynamics are Brownian, the iterative scheme above reduces precisely to the Conze--Henry-Labord\`ere~\cite{CoHe21} algorithm.}

\subsection{An instantaneous calibration principle --- from SKR to SLV}

Since local volatility is well known to generate unrealistic forward smiles and
spot-volatility dynamics, it is common practice to superimpose a leverage function or leverage adjustment on
a stochastic volatility reference, which leads to stochastic local volatility (SLV) models.
Their calibration amounts to a McKean--Vlasov equation whose  well-posedness and numerics are somewhat
delicate, see \cite{AbTa10, guyon2011smile,GuHe13, JoZh20, CoMaRe19, LaShZh20}.  The present SKR construction addresses the same calibration task, but replaces the McKean--Vlasov equation by a sequence of 1-d WMOT problems.

\begin{figure}[H]
    \centering
    \begin{minipage}[c]{0.41\textwidth}
        \centering
        \includegraphics[width=\linewidth]{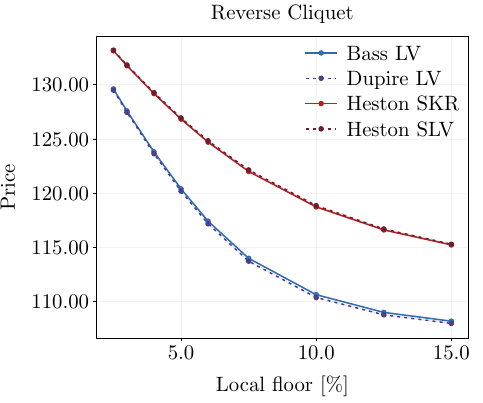}
    \end{minipage}
    \hfill
    \begin{minipage}[c]{0.56\textwidth}
        \centering
        \small
        \sisetup{output-decimal-marker={.},group-digits=false}
        \setlength{\tabcolsep}{3pt}
        \renewcommand{\arraystretch}{1.1}
    
        \begin{tabular}{
            @{}l
            S[table-format=3.3]
            S[table-format=1.3]
            S[table-format=2.3]
            S[table-format=2.4]
            >{\bfseries}r@{}
        }
            \toprule
            & \multicolumn{2}{c}{Valuation}
            & \multicolumn{3}{c}{Time [s]} \\
            \cmidrule(lr){2-3}
            \cmidrule(l){4-6}
            Model & {Price} & {MC SE}
                  & {Calibration} & {Pricing} & {\textbf{Total}} \\
            \midrule
            Bass LV
                & 113.987 & 0.038
                & 0.079 & 0.018 & \textbf{0.097} \\
            Dupire LV
                & 113.879 & 0.038
                & 0.005 & 4.290 & \textbf{4.295} \\
            Heston SKR
                & 122.000 & 0.039
                & 0.197 & 0.076 & \textbf{0.273} \\
            Heston SLV
                & 122.139 & 0.039
                & 18.802 & 6.976 & \textbf{25.778} \\
            \bottomrule
        \end{tabular}
    \end{minipage}

    \caption[Reverse-cliquet comparison of Bass/SKR and LV/SLV]{
        Comparison of Bass LV with Dupire LV and Heston SKR with Heston SLV for the reverse cliquet of \Cref{sec:numerics.cliquets}.
        All models use the SPX-based ESSVI targets of \Cref{sec:numerics.vanillas}, with identical reference parameters for the two Heston models.
        The plot varies the loss cap $L$; the table reports prices, Monte Carlo standard errors, and timings at $L=7.5\%$ for $10^5$ pricing paths.\protect\footnotemark 
    }
    \label{fig:barriers}
\end{figure}
\footnotetext{
    Dupire local volatility is computed analytically, without numerical calibration, from five ESSVI smiles at $0.25$, $0.5$, $1$, $2$, and $3$ years (SPX, 23 June 2020~\cite{FaFeUl22}).
    Heston SLV is calibrated to this local-volatility target using $250\,000$ particles.
    The ESSVI surface calibration from raw quotes is excluded from the timings.
    
    We use QuantLib's C++ implementations of Dupire LV and Heston SLV with $64$ time steps per year, and our C++ implementations of Bass LV and Heston SKR.
    Computations were performed on an Intel Core i5-12500 processor with $32$~GiB RAM under Debian~13.
}

There is a simple probabilistic heuristic which suggests convergence of SKR to SLV. 
We start by interpreting
SLV as the \emph{instantaneously least-distorting Markovian calibration.}
Consider a diffusion reference
$\rmd Y_t=\sigma_R(t,Y_t,L_t)\rmd W_t$
with an autonomous factor $L$, and suppose that a smooth family of target marginals $(\mu_t)$ is prescribed. Write $\sigma_{\mathrm{Dup}}$ for the corresponding Dupire diffusion coefficient, i.e., the volatility depending only on time and state for which the local-volatility martingale
$\rmd Z_t=\sigma_{\mathrm{Dup}}(t,Z_t)\rmd B_t$
has $\Law(Z_t)=\mu_t$.

At a fixed time, retain the joint law of $(X_t,L_t)$ and ask how to modify the reference volatility as little as possible while matching the prescribed marginal evolution. For calibrated dynamics $\rmd X_t=a(t,X_t,L_t)\rmd W_t$, the latter requires
\begin{equation}\label{eq:intro.local.variance}
\E[a(t,X_t,L_t)^2\mid X_t=x]
=\sigma_{\mathrm{Dup}}(t,x)^2.
\end{equation}
Among coefficients satisfying \eqref{eq:intro.local.variance}, the conditional mean-square distortion
$\E[|a(t,x,L_t)-\sigma_R(t,x,L_t)|^2\mid X_t=x]$
is minimized by
\begin{equation}\label{eq:intro.local.projection}
a(t,x,l)=\lambda(t,x)\sigma_R(t,x,l),
\qquad
\lambda(t,x)
=\frac{\sigma_{\mathrm{Dup}}(t,x)}
{\sqrt{\E[\sigma_R(t,X_t,L_t)^2\mid X_t=x]}};
\end{equation}
this is a consequence of Cauchy--Schwarz.
Thus the least-distorting instantaneous adjustment rescales the reference volatility by a factor $\lambda$ depending only on time and price. This is precisely the SLV structure $\rmd X_t=\lambda(t,X_t)\sigma_R(t,X_t,L_t)\rmd W_t.$ 

SKR implements the same principle over finite calibration intervals. Each bridge in \eqref{eq:bridge} compares the calibrated and reference martingale increments, starting from the same price and factor state. Over a short interval, their mean-square difference is governed, to leading order, by the difference of their instantaneous volatilities. Matching the next marginal imposes the variance constraint \eqref{eq:intro.local.variance} to first order. The optimal bridge should therefore select the leverage adjustment \eqref{eq:intro.local.projection}. 
The leverage in \eqref{eq:intro.local.projection} depends on the distribution of the factor $L$ conditional on the current calibrated price. By propagating the joint law of $(X_{T_i},L_{T_i})$, SKR carries this information forward to the next calibration step.

This suggests viewing SKR as a variational time discretization of SLV, with exact calibration at every grid point.
We develop this heuristic in \Cref{sec:SLV_formal2} and establish convergence under regularity assumptions in Appendix~\ref{app:skr:sec:convergence}.

\subsection{Multi-asset calibration with SKR}

The SKR construction extends to several assets while retaining the
scalar structure of the calibration problems. We consider a Markov
reference model $R=(Y,L)$ on $\Rd\times\mathsf E$, whose price
coordinates $Y=(Y^1,\ldots,Y^d)$ form a square integrable martingale.
The calibration data consist of the single-asset marginal laws
$X^j_{T_i}\sim\mu^j_i$, where
$\delta_{x_0^j}=\mu^j_0\leq_c\cdots\leq_c\mu^j_n$ have finite second
moments. These constraints leave the inter-asset dependence
undetermined; the reference model provides the dynamics against
which the joint calibration is compared.

As in the single-asset construction, we propagate the full joint
state law
\(
    \kappa_i
    :=\Law(X^1_{T_i},\ldots,X^d_{T_i},L_{T_i}).
\)
To describe one calibration step, rescale the interval to $[0,1]$,
write $\kappa=\kappa_i$ and $\nu^j=\mu^j_{i+1}$, and start the
reference process with $R_0\sim\kappa$. The bridge problem becomes
\begin{equation}\label{eq:multi.Ubridge}
    \inf_{\substack{X:\,X_0=Y_0, 
                    X^j_1\sim\nu^j}}
    \E\Big[\sum_{j=1}^d[X^j-Y^j]_1\Big],
\end{equation}
with the same martingale and filtration requirements as in the
single-asset case.

The reason the preceding theory applies is that, for fixed
$\kappa$, both the distortion objective and the terminal constraints
separate across assets. For asset $j$, we use the single-asset
construction with the remaining reference price coordinates
included in the latent state. In particular, its conditional
reference laws are taken given the full state
$r=(x^1,\ldots,x^d,\ell)$. The resulting scalar optimizers can all
be realized on the same reference process, so they simultaneously
attain the minimum of the sum in \eqref{eq:multi.Ubridge}.

More explicitly, put
$\gamma_r^j:=\Law(Y_1^j\mid R_0=r)$. Under the assumptions of the
single-asset Bass characterization, imposed coordinatewise,
the unique optimizer has the representation
\begin{equation}\label{eq:multi.R.Bass}
    \begin{aligned}
        X_t^j
        &=
        \E\Bigl[
            g^j\bigl(a^j(R_0)+Y_1^j\bigr)
            \,\Big|\,\mathcal F_t
        \Bigr],\qquad
        a^j(r)
        &=
        (g^j\star\gamma_r^j)^{-1}(x^j),
        \qquad j=1,\ldots,d,
    \end{aligned}
\end{equation}
where $(\mathcal F_t)$ is the reference filtration and the increasing
maps $g^j:\R\to[-\infty,\infty]$ enforce the terminal laws $\nu^j$. The shifts
$a^j$ enforce the martingale constraints and may depend
on all components of the initial state.

For fixed incoming state law, the Martingale Sinkhorn iteration
therefore applies separately to each coordinate, using the kernels
$\gamma_r^j$. The coordinate updates can be performed in parallel,
and no $d$-dimensional transport problem needs to be solved.
The precise formulation and the existence and uniqueness result
are given in \Cref{subsec:bridge.multi}.

\subsection{Closely related literature}
\label{subsec:related.literature}

SLV models combine the exact fit of local volatility with the richer
dynamics of stochastic volatility models. Early constructions include
those of Jex--Henderson--Wang~\cite{JeHeWa99},
Lipton~\cite{Li02}, and Ren--Madan--Qian~\cite{ReMaQi07}.
The leverage relation connects Dupire's local volatility
formula~\cite{Du94} with the mimicking theorem of
Gy\"ongy~\cite{Gy86}. Guyon--Henry-Labord\`ere~\cite{guyon2011smile,GuHe13}
developed a particle method for the resulting conditional
McKean--Vlasov problem. Its well-posedness is delicate; rigorous
results in restricted settings include
\cite{AbTa10,JoZh20,LaShZh20}. A variational approach to calibration
was developed by Guo--Loeper--Wang, first for local
volatility~\cite{GuLoWa19} and subsequently for local-stochastic
volatility~\cite{GuLoWa22}; see also Guo--Loeper~\cite{GuLo21}
for path-dependent calibration.

The Bass construction goes back to Bass~\cite{Ba83}, who used a
monotone transform of Brownian motion to solve the Skorokhod embedding
problem. The Bass LV model for general initial laws was introduced
independently in \cite{BaBeHuKa20} and in the work of Loeper and
Guo--Loeper--Wang~\cite{Lo18,GuLoWa19}, through probabilistic and PDE
approaches to the martingale Benamou--Brenier problem.
Conze--Henry-Labord\`ere~\cite{CoHe21} proposed the fixed-point
iteration that makes Bass calibration particularly efficient.
Loeper~\cite{Lo23} subsequently formulated and interpreted this
iteration as the Martingale Sinkhorn algorithm.
Acciaio--Marini--Pammer~\cite{AcMaPa23} proved existence and
uniqueness of the fixed point, up to translation, and geometric
convergence under suitable assumptions. The multidimensional
formulation and convergence analysis of Martingale Sinkhorn are
developed further by
Hasenbichler--Joseph--Loeper--Ob{\l}{\'o}j--Pammer~\cite{HaJoLoObPa25}.
The relation between increasingly fine Bass interpolations and the
Dupire model is studied in \cite{BePaSc22,AcMaPa23}.

Extensions of the Bass construction beyond Brownian reference dynamics already appear in Conze--Henry-Labord\`ere~\cite[Section~6]{CoHe21}.
Their Theorem~6.1 gives a fixed-point characterization for a general diffusion driver, whose endpoint construction coincides with the one-dimensional fixed-$q$ Bass construction. In Section~6.3 they also consider state-dependent kernels whose specification depends on an additional stochastic factor. Geometric Bass martingales and their change-of-numeraire representation are studied in \cite{BaLoOb24,BePaRi24}.
The variational formulation and duality for $q$-Bass martingales were developed by Tschiderer~\cite{Ts24}, while Acciaio--Marini~\cite{AcMa26a,AcMa26b} establish existence, stability and convergence results under suitable assumptions.
In independent work, Marini~\cite{Ma26_qxBass} allows the reference kernels to depend on the initial value $x$ and refers to the resulting optimizers as $q_x$-Bass martingales. For the associated weak optimal transport problem, he develops the corresponding dual formulation and a generalized Martingale Sinkhorn scheme, establishing convergence and existence under suitable assumptions and deriving a one-dimensional fixed-point theory.
Moreover, Marini is developing a multiperiod extension of $q_x$-Bass martingales and its continuous-time limit, where, unlike in the present work, the intermediate marginals are found by optimization rather than prescribed; see forthcoming work~\cite{Ma26_volMimicking}.

The novelty of our approach lies in extending Bass-type calibration with general stochastic and path-dependent reference dynamics beyond the two-marginal setting. The stochastic Knothe--Rosenblatt construction is the first to achieve this while propagating the full joint price--factor law across arbitrarily many calibration dates. This allows for a highly efficient and robust numerical implementation of the associated SLV model.


\section{The bridge problem as a weak martingale transport problem} \label{sec:bridge.wmot}

As in the introduction, we rescale the calibration interval to $[0,1]$. 
The reference process $R=(Y,L)$ has initial law $\kappa\in\cP_2(\R\times\mathsf E)$, whose first marginal we denote by $\mu:=\pr^X_\#\kappa$.
We work with the usual augmentation $(\mathcal F_t)$ of the natural filtration generated by $R$, and assume that $R$ is Markov with respect to this filtration.
For a prescribed $\nu\in\cP_2(\R)$, the class $\Acal(\kappa,\nu)$ consists of square-integrable $(\mathcal F_t)$-martingales $X$ satisfying $X_0=Y_0$ and $X_1\sim\nu$.

For $r=(x,\ell)\in\R\times\mathsf E$, write $\gamma_r:=\Law(Y_1\mid R_0=r)$ for the conditional terminal law of the reference process. We work under the following assumptions on the conditional reference laws:
\begin{enumerate}[label={\rm(A\arabic*)}]
    \item \label{aspt:lsc}
    $(\rho,r)\mapsto\Wcal_2(\rho,\gamma_r)$ is lower semicontinuous on $\cP_2(\R)\times(\R\times\mathsf E)$.
    \item \label{aspt:moments}
    $\iint z^2\,\gamma_r(\rmd z)\,\kappa(\rmd r)<\infty$, and $\bar\gamma_r=x$ for $\kappa$-a.e.~$r=(x,\ell)$.
    \item \label{aspt:atomless}
    $\gamma_r$ is atomless for $\kappa$-a.e.~$r$.
\end{enumerate}

Since $X_0=Y_0$, integration by parts shows that the bridge problem is equivalent to
\begin{equation}\label{eq:wmot.bridge}
    P(\kappa,\nu)
    :=
    \sup_{X\in\Acal(\kappa,\nu)}\E[X_1Y_1].
\end{equation}
To obtain a static formulation, consider an admissible bridge $X$ and set $\pi:=\Law(R_0,X_1)$, with disintegration $\pi(\rmd r,\rmd y)=\kappa(\rmd r)\pi_r(\rmd y)$.
The martingale property implies
\begin{equation}\label{eq:conditional.means}
    \int y\,\pi_r(\rmd y)=x
    \quad\text{for $\kappa$-a.e.~$r=(x,\ell)$.}
\end{equation}
Thus $\pi$ belongs to $\Mcal(\kappa,\nu)$, the set of couplings in $\Cpl(\kappa,\nu)$ satisfying the martingale property \eqref{eq:conditional.means}.

Conditionally on $R_0=r$, the pair $(X_1,Y_1)$ is a coupling of $\pi_r$ and $\gamma_r$ leading us to consider the maximal covariance
\[
    \MCov(\rho,\gamma)
    :=
    \sup_{q\in\Cpl(\rho,\gamma)}
    \int yz\,q(\rmd y,\rmd z).
\]
Adapting the dynamic-to-static reduction used for the martingale Benamou--Brenier problem in \cite{BaBeHuKa20}, we obtain the following static formulation of the bridge problem.
\begin{theorem}[Existence and uniqueness for the primal problem]\label{thm:wmot.bridge.existence.uniqueness}
    Suppose that \ref{aspt:lsc}, \ref{aspt:moments} and \ref{aspt:atomless} hold, and let $\mu \le_{\rm cvx} \nu$. Then
    \begin{equation}\label{eq:wmot.bridge.static}\tag{wBridge}
        P(\kappa,\nu)
        =
        \max_{\pi\in\Mcal(\kappa,\nu)}
        \int\MCov(\pi_r,\gamma_r)\,\kappa(\rmd r).
    \end{equation}
    Moreover, both the dynamic and static problems admit optimizers. The optimizer of \eqref{eq:wmot.bridge.static} is unique, and the optimizer of \eqref{eq:wmot.bridge} is unique in law.
\end{theorem}
The maximization problem on the right-hand side of \eqref{eq:wmot.bridge.static} fits into the weak optimal transport framework of \cite{GoRoSaTe17}.
More specifically, it falls within the extended weak martingale optimal transport framework of \cite[Section~2.1]{JoPa24}, with $\ell$ as the additional information parameter.

\begin{proof}
    Conditioning on $R_0$ and using the definition of $\MCov$ yields
    \begin{equation}\label{ineq:proof:bridge.dynamic.static}
        P(\kappa,\nu)
        \leq
        \sup_{\pi\in\Mcal(\kappa,\nu)}
        \int\MCov(\pi_r,\gamma_r)\,\kappa(\rmd r).
    \end{equation}
    The static objective equals
    $C-\frac12\int\Wcal_2^2(\pi_r,\gamma_r)\,\kappa(\rmd r)$,
    where $C<\infty$ is independent of $\pi$ by \ref{aspt:moments}.
    Since $(r,\rho)\mapsto\Wcal_2^2(\rho,\gamma_r)$ is nonnegative, jointly lower semicontinuous by \ref{aspt:lsc}, and convex in $\rho$, its integral is $\Wcal_2$-lower semicontinuous by \cite[Proposition~2.8]{BaBePa19}.
    Thus the objective is upper semicontinuous on the nonempty $\Wcal_2$-compact set $\Mcal(\kappa,\nu)$ and admits a maximizer $\hat\pi$.

    By \ref{aspt:atomless}, the conditional monotone couplings can be realized on the reference probability space by setting
    \[
        X_1:=F^{-1}_{\hat\pi_{R_0}}\bigl(F_{\gamma_{R_0}}(Y_1)\bigr),
        \qquad X_t:=\E[X_1\mid\mathcal F_t],
    \]
    where $F_\rho$ and $F^{-1}_\rho$ denote the distribution and quantile functions of $\rho$.
    The barycenter constraint gives $X_0=Y_0$, while $X_1\sim\nu$, so $X\in\Acal(\kappa,\nu)$.
    By construction,
    \[
        \E[X_1Y_1]
        =
        \int\MCov(\hat\pi_r,\gamma_r)\,\kappa(\rmd r).
    \]
    Hence equality holds in \eqref{ineq:proof:bridge.dynamic.static}, and both problems are attained.

    By \ref{aspt:atomless}, strict concavity of $\rho\mapsto\MCov(\rho,\gamma_r)$ gives uniqueness of $\hat\pi$.
    Any dynamic optimizer $X$ must satisfy $\Law(R_0,X_1)=\hat\pi$, and its conditional terminal coupling must attain $\MCov(\hat\pi_r,\gamma_r)$.
    This is the unique monotone coupling, so $X_1$ is uniquely determined, up to null sets, as a function of $(R_0,Y_1)$.
    Since $X_t=\E[X_1\mid\mathcal F_t]$, the optimizer is unique, and in particular unique in law.
\end{proof}

In particular, applying \Cref{thm:wmot.bridge.existence.uniqueness} to every bridge between maturities establishes \Cref{thm:skr}, the existence and uniqueness of the SKR model.

We next characterize optimal bridges as $R$-Bass martingales through the dual problem, following the approach of \cite{BaBeScTs26} for the martingale Benamou--Brenier problem.
Under the assumptions below, this gives a concrete representation of the optimal bridge in terms of a single convex dual potential $\psi$, with generating function $g=\partial_x\psi^*$.
Here, $\psi^*$ denotes the convex conjugate of $\psi$.
This representation provides the basis for the martingale Sinkhorn algorithm and its convergence analysis.

Here and below, the inverse in Definition~\ref{def:R.Bass} denotes a measurable choice of a finite shift $a(r)$ satisfying $(g\star\gamma_r)(a(r))=x$ for $\kappa$-a.e.~$r=(x,\ell)$. We require $g(a(R_0)+Y_1)\in L^2$; in particular, the convolution at $a(r)$ is absolutely convergent for $\kappa$-a.e.~$r$.

To derive the dual problem, first let $\psi\in L^1(\nu)$ be proper, convex and lower semicontinuous. For $r=(x,\ell)$, set
\[
    (\psi^*\star\gamma_r)(a)
    :=
    \int\psi^*(a+z)\,\gamma_r(\rmd z),
    \qquad
    \psi^C(r):=(\psi^*\star\gamma_r)^*(x).
\]
Fenchel's inequality, together with the barycenter constraint $\int y\,\pi_r(\rmd y)=x$, yields
\[
    \MCov(\pi_r,\gamma_r)
    \leq
    \int\psi\,\rmd\pi_r-\psi^C(r).
\]
Consequently,
\[
    \sup_{\pi\in\Mcal(\kappa,\nu)} \int\MCov(\pi_r,\gamma_r)\,\kappa(\rmd r)
    \leq
    \int\psi\,\rmd\nu-\int\psi^C\,\rmd\kappa.
\]
Taking the infimum over such $\psi$ gives the associated dual problem.
However, a limiting potential obtained from a minimizing sequence need not belong to $L^1(\nu)$, so the difference $\int\psi\,\rmd\nu-\int\psi^C\,\rmd\kappa$ may no longer be well defined.
We therefore use the generalized-integral construction of Beiglb\"ock and Juillet \cite[Appendix~A.3]{BeJu16}, further developed by Beiglb\"ock, Nutz and Touzi \cite[Section~4]{BeNuTo17}: instead of evaluating the two integrals separately, we integrate the conditional difference.

Write $I_\nu:=\operatorname{int}\co(\supp\nu)$.
A \emph{dual potential} is a proper, convex and lower semicontinuous function $\psi$ that is finite on $I_\nu$ and equal to $+\infty$ outside $\bar I_\nu$.
We refer to this domain restriction as the \emph{support convention}.
Adapting the relaxed dual functional in \cite[Proposition~4.2]{BaBeScTs26} to the present setting, we define, for an irreducible pair $(\mu,\nu)$, a dual potential $\psi$ and $\pi\in\Mcal(\kappa,\nu)$,
\begin{equation}\label{eq:relaxed.functional}
    \Ecal(\psi)
    :=
    \int \Bigl( \int\psi\,\rmd\pi_r-\psi^C(r) \Bigr) \,\kappa(\rmd r).
\end{equation}
\Cref{lem:properties.Ecal} shows that $\Ecal(\psi)$ is well-defined and independent of the choice of $\pi\in\Mcal(\kappa,\nu)$.
For $\psi\in L^1(\nu)$, it reduces to $\int\psi\,\rmd\nu-\int\psi^C\,\rmd\kappa$.

\begin{theorem}[Duality and characterization of $R$-Bass martingales]\label{thm:wmot.bridge.duality}
    Suppose that \ref{aspt:lsc}, \ref{aspt:moments} and \ref{aspt:atomless} hold, and let $(\mu,\nu)$ be irreducible. Then
    \begin{equation}\label{eq:wmot.bridge.dual}\tag{dBridge}
    P(\kappa,\nu)
    =
    \min_{\psi\ {\rm dual\ potential}}
    \Ecal(\psi).
    \end{equation}
    A martingale $X\in\Acal(\kappa,\nu)$ solves \eqref{eq:wmot.bridge} if and only if it is an $R$-Bass martingale.
    
    Moreover, suppose that $\supp \gamma_r$ is either $\R$ or $[0,\infty)$ for $\kappa$-a.e.~$r$. Then the dual minimizer is unique up to addition of an affine function on $I_\nu$.
\end{theorem}

In particular, under these assumptions, to every $(\kappa,\nu)$ there is at most one $R$-Bass martingale $X \in \Acal(\kappa,\nu)$.
To show \Cref{thm:wmot.bridge.duality}, we begin by identifying the value of the static problem with the dual problem over integrable convex potentials.

\begin{lemma}[No duality gap]\label{lem:bridge.no.duality.gap}
    Suppose that \ref{aspt:lsc} and \ref{aspt:moments} hold and that $\mu\le_{\rm cvx}\nu$.
    Then
    \begin{equation}\label{eq:wmot.bridge.strong.duality}
        \sup_{\pi\in\Mcal(\kappa,\nu)}
        \int\MCov(\pi_r,\gamma_r)\,\kappa(\rmd r)
        =
        \inf_{\substack{v\in L^1(\nu)\\v\ {\rm cvx,\ lsc}}}
        \Bigl\{
            \int v\,\rmd\nu-\int v^C\,\rmd\kappa
        \Bigr\}.
    \end{equation}
\end{lemma}

\begin{proof}
    For a Gaussian reference law, this dual representation was established in \cite[Sections~3--4]{BaBeScTs26}.
    We adapt that argument to the conditional reference laws $\gamma_r$.

    The cost $C(r,\rho):=\frac12\Wcal_2^2(\rho,\gamma_r)$ when $\bar\rho=x$, and $+\infty$ otherwise, is nonnegative, jointly lower semicontinuous and convex in $\rho$.
    For such costs, \cite[Theorem~1.3]{BaBePa19} identifies the primal value with a dual over continuous test functions.
    After cancelling the second-moment terms, this gives
    \begin{gather*}
        \sup_{\pi\in\Mcal(\kappa,\nu)}
        \int\MCov(\pi_r,\gamma_r)\,\kappa(\rmd r)
        =
        \inf_v\Bigl\{
            \int v\,\rmd\nu-\int v_r(x)\,\kappa(\rmd r)
        \Bigr\},
        \quad
        v_r(z):=
        \inf_{\substack{\rho\in\cP_2(\R)\\\bar\rho=z}}
        \Bigl\{
            \int v\,\rmd\rho-\MCov(\rho,\gamma_r)
        \Bigr\},
    \end{gather*}
    where $v$ is continuous and satisfies $c+\frac12|\cdot|^2\le v\le a+b|\cdot|^2$ for some $a,b,c\in\R$.

    The argument of Pramenkovi\'c \cite[Theorem~3.2]{Pr25} shows that, for weak costs decreasing in convex order, replacing a dual potential by its convex envelope cannot worsen the dual objective.
    Since $\MCov(\cdot,\gamma_r)$ is increasing in convex order, this argument allows us to restrict the infimum to convex $v$.

    It remains to identify $v_r$. 
    The calculation in the proof of \cite[Proposition~5.6]{BePaRiSc25} removes the barycenter constraint by conjugation and yields, under our quadratic bounds,
    \[
        v_r^*(a)
        =
        \int v^*(a+y)\,\gamma_r(\rmd y)
        =
        (v^*\star\gamma_r)(a).
    \]
    The quadratic bounds make $v_r$ finite, while mixing competitors gives convexity.
    Thus Fenchel--Moreau yields $v_r=(v^*\star\gamma_r)^*$ and hence $v_r(x)=v^C(r)$.

    Finally, Fenchel's inequality gives weak duality for every convex lower semicontinuous $v\in L^1(\nu)$, so enlarging the class of potentials leaves the value unchanged.
    In the irreducible case used below, setting $v=+\infty$ outside $\bar I_\nu$ preserves $\int v\,\rmd\nu$ and can only increase $v^C$, so the support convention also leaves the infimum unchanged.
\end{proof}

The preceding lemma identifies the static value with the dual problem over $\nu$-integrable convex potentials. 
To pass to the relaxed dual problem, we first verify that the functional $\Ecal$ is well defined and does not depend on the coupling $\pi$ used in its definition.

\begin{lemma}\label{lem:properties.Ecal}
    Suppose that $(\mu,\nu)$ is irreducible and that \ref{aspt:moments} holds. Let $v: \R \to \R \cup \{+\infty\}$ be proper, convex and lower semicontinuous with $I_\nu = \operatorname{int} \dom(v)$. Then $\Ecal(v)\in[0,\infty]$ is well defined and is independent of the choice of $\pi\in\Mcal(\kappa,\nu)$.
    Moreover, $\Ecal$ is invariant under addition of affine functions.
\end{lemma}

\begin{proof}
    Let $\pi \in \Mcal(\kappa,\nu)$.
    For $\kappa$-a.e.~$r=(x,\ell)$, we have $(v^*\star\gamma_r)^*(x) \le v(x) - x^2$. By Jensen's inequality,
    \[
        \int v(y)\,\pi_r(\rmd y) - (v^*\star\gamma_r)^*(x) \ge x^2,
    \]
    and so $\Ecal(v) \in [0,\infty]$. Moreover, define $\chi(\rmd x,\rmd y) := \int \pi_{x,\ell}(\rmd y) \, \kappa(\rmd x,\rmd \ell)$, which is a martingale coupling of $(\mu,\nu)$. Then
    \begin{align*}
        \Ecal(v)
        &=
        \int\Bigl(\int v(y)\,\pi_{x,\ell}(\rmd y)-v(x)\Bigr)\kappa(\rmd x,\rmd\ell)
        +
        \int\bigl(v(x)-v^C(x,\ell)\bigr)\,\kappa(\rmd x,\rmd\ell)
        \\
        &=
        \int\Bigl(\int v(y)\,\chi_{x}(\rmd y)-v(x)\Bigr)\mu(\rmd x)
        +
        \int\bigl(v(x)-v^C(x,\ell)\bigr)\,\kappa(\rmd x,\rmd\ell).
    \end{align*}
    \cite[Lemma A.15]{BeJu16} asserts that the first integral depends only on $v$ and the marginals $(\mu,\nu)$, but not on the particular coupling $\chi$. 
    Since the second integral is independent of $\pi$ as well, $\Ecal(v)$ is independent of the choice $\pi \in \Mcal(\kappa,\nu)$.

    Finally, let $c(x):=kx+d$ for some $k,d\in\R$. Then $ (v+c)^*(y)=v^*(y-k)-d$, and hence
    \[
        ((v+c)^*\star\gamma_r)(y)
        =
        (v^*\star\gamma_r)(y-k)-d.
    \]
    Taking conjugates gives
    \[
        ((v+c)^*\star\gamma_r)^*(x)
        =
        (v^*\star\gamma_r)^*(x)+kx+d
        =
        v^C(r)+c(x).
    \]
    Since $\bar\pi_r=x$, we therefore obtain
    \[
        \int (v+c)\,\rmd\pi_r-(v+c)^C(r)
        =
        \int v\,\rmd\pi_r-v^C(r),
    \]
    for $\kappa$-a.e.~$r=(x,\ell)$. Integrating with respect to $\kappa$ yields $\Ecal(v+c)=\Ecal(v)$.
\end{proof}

To obtain a dual optimizer, we will pass to a locally uniform limit of normalized convex potentials on $I_\nu$. The next lemma shows that $\Ecal$ is lower semicontinuous with respect to such sequences.

\begin{lemma}\label{lem:Ecal.lsc}
    Suppose that $(\mu,\nu)$ is irreducible and that \ref{aspt:moments} holds. Let $\hat\psi$ and $(\hat\psi_n)_{n\in\N}$ be dual potentials with $\hat\psi,\hat\psi_n\ge0$ and $\hat\psi(\bar\mu)=\hat\psi_n(\bar\mu)=0$ for every $n$.
    Suppose that $\hat\psi_n\to\hat\psi$ locally uniformly on $I_\nu$.
    Then
    \[
        \Ecal(\hat\psi)
        \le
        \liminf_{n\to\infty}\Ecal(\hat\psi_n).
    \]
\end{lemma}

\begin{proof}
    Fix $r=(x,\ell)$ with $x\in I_\nu$, $\gamma_r\in\cP_2(\R)$ and $\bar\gamma_r=x$, as holds for $\kappa$-a.e.~$r$. We first show that
    \[
        \limsup_{n \to \infty} (\hat\psi_n^*\star\gamma_r)^*(x) \le (\hat\psi^*\star\gamma_r)^*(x).
    \]
    Pass to a subsequence along which the limsup is attained as a limit. If this limit is $-\infty$, there is nothing to show. Otherwise, $(\hat\psi_n^*\star\gamma_r)^*(x)$ is bounded below for all sufficiently large $n$.
    
    Choose $\delta > 0$ such that $[x-\delta,x+\delta]\subset I_\nu$. Since $(\hat\psi_n^*\star\gamma_r)^*(z) \le \hat\psi_n(z) - xz$ for every $z \in I_\nu$, $(\hat\psi_n^*\star\gamma_r)^*$ is uniformly bounded above on $[x-\delta,x+\delta]$.
    Let $a_n$ such that $(\hat\psi_n^*\star\gamma_r)^*(x) \le a_n x - (\hat\psi_n^*\star\gamma_r)(a_n) + \frac1n$ for every $n\in\N$.  Since $(\hat\psi_n^*\star\gamma_r)^*(z) \ge a_n z - (\hat\psi_n^*\star\gamma_r)(a_n)$ for every $z \in I_\nu$, the sequence $(a_n)_{n \in \N}$ is bounded. Passing to a further subsequence, let $a_n \to a$. Then, for every $z$ and every $y \in I_\nu$, 
    \[
        \liminf_{n \to \infty} \hat\psi_n^*(a_n+z) \ge \liminf_{n \to \infty} \bigl((a_n+z)y - \hat\psi_n(y)\bigr) = (a+z)y - \hat\psi(y),
    \]
    and so $\liminf_{n \to \infty} \hat\psi_n^*(a_n+z) \ge \hat\psi^*(a+z)$. Moreover, $\hat\psi_n^*(y) \ge y \bar\mu$, and therefore Fatou's lemma gives 
    \[
        \int \hat\psi^*(a+y) \,\gamma_r(\rmd y) \le \liminf_{n \to \infty} \int \hat\psi_n^*(a_n+y) \,\gamma_r(\rmd y).
    \]
    Consequently, 
    \[
        \limsup_{n \to \infty} (\hat\psi_n^*\star\gamma_r)^*(x) 
        \le
        a x - \int \hat\psi^*(a+y) \,\gamma_r(\rmd y) 
        \le 
        (\hat\psi^*\star\gamma_r)^*(x).
    \]
    
    Let $\pi \in \Mcal(\kappa,\nu)$. By convexity and lower semicontinuity, $\hat\psi(y)\le\liminf_n\hat\psi_n(y)$ also at the endpoints of $I_\nu$. Applying Fatou's lemma two more times then yields
    \[
        \Ecal(\hat\psi) 
        = 
        \int \Bigl(\int \hat \psi \,\rmd\pi_{x,\ell} - (\hat\psi^*\star\gamma_r)^*(x)\Bigr) \, \kappa(\rmd x, \rmd \ell)
        \le
        \liminf_{n \to \infty} \Ecal(\hat\psi_n). \qedhere
    \]
\end{proof}

The following elementary lemma will be used for showing uniqueness of the dual optimizer.

\begin{lemma}\label{lem:monotone.transport.uniqueness}
    Let $\gamma\in\cP(\R)$ be atomless and suppose that $\supp\gamma$ is an interval. Let $\psi,v:\R\to\R\cup\{+\infty\}$ be proper, convex and lower semicontinuous. Suppose that, for some $a,b\in\R$ and $\rho\in\cP(\R)$,
    \[
        \bigl(\partial_x\psi^*(a+\cdot)\bigr)_\#\gamma
        =
        \bigl(\partial_x v^*(b+\cdot)\bigr)_\#\gamma
        =
        \rho.
    \]
    Then $\psi-v$ is affine on $I_\rho:=\operatorname{int}\operatorname{co}(\operatorname{supp}\rho)$.
\end{lemma}

\begin{proof}
    Set $\psi_a:=\psi-a\,\operatorname{id}$ and $v_b:=v-b\,\operatorname{id}$, so that $\psi_a^*=\psi^*(a+\cdot)$ and
    $v_b^*=v^*(b+\cdot)$.
    Since $\gamma$ is atomless, the monotone transport from $\gamma$ to $\rho$ is $\gamma$-a.e.~unique \cite{Mc95}. Hence $\partial_x\psi_a^*=\partial_x v_b^*$, $\gamma$-a.e.
    As $J:=\supp\gamma$ is an interval, every set of full $\gamma$-measure is dense in $J$. Monotonicity of the derivatives of convex functions therefore yields $\psi_a^*=v_b^*+c$ on $\operatorname{int} J$ for some $c\in\R$.
    
    Now fix $y\in I_\rho$. Then $\rho((-\infty,y))>0$ and $\rho((y,\infty))>0$. 
    Since $\rho=(\partial_x\psi_a^*)_\#\gamma$, monotonicity of $\partial\psi_a^*$ yields some $z\in \operatorname{int} J$ such that $y\in\partial\psi_a^*(z)$. 
    As $\psi_a^*=v_b^*+c$ on $\operatorname{int} J$, also $y\in\partial v_b^*(z)$. Fenchel equality therefore gives
    \[
        \psi_a(y)+\psi_a^*(z)=yz
        =v_b(y)+v_b^*(z),
    \]
    and hence $\psi_a(y)-v_b(y)=-c$. Thus $\psi(y)-v(y)=(a-b)y-c$ for every $y \in I_\rho$.
\end{proof}

The shift in the $R$-Bass representation will be obtained from the subgradient of $(\psi^*\star\gamma_r)^*$. Whenever $\supp \gamma_r$ is either $\R$ or $[0,\infty)$ for $\kappa$-a.e.~$r$, this subgradient is in fact a singleton on the relevant interval, as shown next.

\begin{lemma}\label{lem:c-conjugate.differentiable}
    Let $\gamma\in\cP_1(\R)$ have support either $\R$ or $[0,\infty)$, and let $\psi:\R\to\R\cup\{+\infty\}$ be proper, convex and lower semicontinuous.
    Suppose that $(\psi^*\star\gamma)^*$ is proper.
    Then $(\psi^*\star\gamma)^*(z)\le\psi(z)-z\bar\gamma$ for every $z\in\R$ and, in particular, $\dom\psi\subset\dom(\psi^*\star\gamma)^*$.
    Moreover, $(\psi^*\star\gamma)^*$ is differentiable on $\operatorname{int}\dom\psi$.
\end{lemma}

\begin{proof}
    Since $(\psi^*\star\gamma)^*$ is proper, so is $\psi^*\star\gamma$.
    For every $a\in\R$, Jensen's inequality gives $(\psi^*\star\gamma)(a)\ge\psi^*(a+\bar\gamma)$.
    Taking convex conjugates yields $(\psi^*\star\gamma)^*(z)\le\psi(z)-z\bar\gamma$ for every $z\in\R$.
    In particular, $\dom\psi\subset\dom(\psi^*\star\gamma)^*$.

    Suppose that $(\psi^*\star\gamma)^*$ is not differentiable at some $x\in \operatorname{int}\dom\psi$, and choose $a_1<a_2$ in $\partial(\psi^*\star\gamma)^*(x)$.
    By Fenchel--Moreau,
    \[
        (\psi^*\star\gamma)(a_1)-a_1x
        =
        -(\psi^*\star\gamma)^*(x)
        =
        (\psi^*\star\gamma)(a_2)-a_2x.
    \]
    Since $a\mapsto(\psi^*\star\gamma)(a)-ax$ is convex and attains its minimum at both $a_1$ and $a_2$, it is constant on $[a_1,a_2]$.
    Thus there exists $b\in\R$ such that
    \begin{equation}\label{eq:c-conjugate.differentiable:affine}
        (\psi^*\star\gamma)(a)=xa+b,
        \qquad a\in[a_1,a_2].
    \end{equation}
    In particular,
    \[
        \int\psi^*\Bigl(\frac{a_1+a_2}{2}+z\Bigr)\,\gamma(\rmd z)
        =
        \int\Bigl(
            \frac12\psi^*(a_1+z)
            +
            \frac12\psi^*(a_2+z)
        \Bigr)\,\gamma(\rmd z).
    \]
    By convexity, equality holds pointwise for $\gamma$-a.e.~$z$, and hence $\psi^*$ is affine on $[a_1+z,a_2+z]$ for every such $z$.

    Set $J:=a_1+\supp\gamma$.
    The set of such $z$ is dense in $\supp\gamma$, so the interiors of these intervals cover $\operatorname{int}J$.
    The affine representations agree on overlaps, and therefore $\psi^*$ is affine on $\operatorname{int}J$ and, by lower semicontinuity, on $J$.
    Write $\psi^*(u)=su+d$ for $u\in J$.
    Since $J$ is unbounded above, Fenchel--Moreau gives $\psi(y)\ge\sup_{u\in J}\{u(y-s)-d\}=+\infty$ for every $y>s$.
    Hence $\sup\dom\psi\le s$.
    On the other hand, for every $u\in\operatorname{int}J$, we have $s\in\partial\psi^*(u)$ and hence $u\in\partial\psi(s)$, so $s\in\dom\psi$.
    It follows that $s=\sup\dom\psi$.

    Since $x\in\operatorname{int}\dom\psi$, we have $x<s$.
    Moreover, $a+\supp\gamma\subset J$ for every $a\in[a_1,a_2]$, so
    \[
        (\psi^*\star\gamma)(a)
        =
        \int\psi^*(a+z)\,\gamma(\rmd z)
        =
        sa+s\bar\gamma+d,
        \qquad a\in[a_1,a_2].
    \]
    Comparing with \eqref{eq:c-conjugate.differentiable:affine} gives $s=x$, a contradiction.
\end{proof}

We now turn to the proof of \Cref{thm:wmot.bridge.duality}. We first prove attainment of the relaxed dual problem and show that a dual optimizer $\psi$ generates, through $g=\partial_x\psi^*$, an $R$-Bass martingale attaining the static value. The same structure also yields the converse: if an admissible $R$-Bass martingale is generated by an increasing function $g$, then the corresponding convex potential $\psi$ with $g=\partial_x\psi^*$ is a dual optimizer.

\begin{proof}[Proof of \Cref{thm:wmot.bridge.duality}]
    By \Cref{thm:wmot.bridge.existence.uniqueness}, there exists $\hat\pi\in\Mcal(\kappa,\nu)$ attaining the static problem. Together with \Cref{lem:bridge.no.duality.gap}, this gives
    \begin{equation}\label{ineq:proof:bridge.initial}
        P(\kappa,\nu)
        =
        \int\MCov(\hat\pi_r,\gamma_r)\,\kappa(\rmd r)
        =
        \inf_{\substack{v\in L^1(\nu)\\ v\ {\rm cvx,\ lsc}}}\Ecal(v).
    \end{equation}
    
    We first prove dual attainment in~\eqref{eq:wmot.bridge.dual}. Let $(\psi_n)_{n\in\N}$ be a minimizing sequence in \eqref{ineq:proof:bridge.initial}, chosen according to the support convention. 
    Since $\Ecal$ is invariant under addition of affine functions by \Cref{lem:properties.Ecal}, we may normalize it as follows: Choose $s_n\in\partial\psi_n(\bar\mu)$ and set $c_n(x):=\psi_n^*(s_n)-s_nx$ as well as $\hat\psi_n:=\psi_n+c_n$.
    Then $\hat\psi_n\geq0$, $\hat\psi_n(\bar\mu)=0$, and $\Ecal(\hat\psi_n)=\Ecal(\psi_n)$.
    
    To extract a converging subsequence, we use that $(\mu,\nu)$ is irreducible. By \cite[Proposition 3.3]{HaJoLoObPa25}, there exists a martingale coupling $\pi'$ of $(\mu,\nu)$ whose conditional laws charge every open set charged by $\nu$. 
    We regard it as an element of $\Mcal(\kappa,\nu)$ by setting $\pi'_{x,\ell}:=\pi'_x$. Since $(\hat\psi_n^*\star\gamma_{x,\ell})^*(x) \leq\hat\psi_n(x)-x^2$, \Cref{lem:properties.Ecal} gives
    \begin{equation}\label{inequ:proof:thm.E&U.controlled.potentials}
        0
        \leq
        \sup_n \int \Bigl( \int\hat\psi_n\,\rmd\pi'_x-\hat\psi_n(x) \Bigr) \, \mu(\rmd x)
        \leq
        \sup_n\Ecal(\hat\psi_n) - \int x^2\,\mu(\rmd x)
        <\infty.
    \end{equation}
    The tightness result \cite[Lemma 3.9]{HaJoLoObPa25} states that normalized proper, convex and lower semicontinuous functions satisfying this bound are locally bounded on $I_\nu$. After passing to a subsequence, we therefore obtain a dual potential $\psi$ such that
    $\hat\psi_n\to\psi$ locally uniformly on $I_\nu$.
    By \Cref{lem:Ecal.lsc},
    \[
        \Ecal(\psi)
        \leq
        \liminf_{n\to\infty}\Ecal(\hat\psi_n)
        =
        \int\MCov(\hat\pi_r,\gamma_r)\,\kappa(\rmd r).
    \]
    Conversely, Fenchel's inequality gives $\MCov(\hat\pi_r,\gamma_r) \leq \int\psi\,\rmd\hat\pi_r-\psi^C(r)$, and hence the reverse inequality after integration. Thus $\psi$ attains \eqref{eq:wmot.bridge.dual}, and moreover
    \begin{equation}\label{eq:proof:bridge.conditional.equality}
        \MCov(\hat\pi_r,\gamma_r)
        =
        \int\psi \,\rmd\hat\pi_r
        -
        (\psi^*\star\gamma_r)^*(x)
    \end{equation}
    for $\kappa$-a.e.~$r=(x,\ell)$.
    
    \medskip
    
    We next use \eqref{eq:proof:bridge.conditional.equality} to construct the optimal $R$-Bass martingale. 
    Since $x\in I_\nu\subset\operatorname{int}\dom\psi$, the subdifferential $\partial(\psi^*\star\gamma_r)^*(x)$ is nonempty. Choose measurably $a(r)\in\partial(\psi^*\star\gamma_r)^*(x)$.
    Fenchel equality gives
    \[
        \MCov(\hat\pi_r,\gamma_r)
        =
        \int\bigl(\psi(y)-a(r)y\bigr)\,\hat\pi_r(\rmd y)
        +
        \int\psi^*(a(r)+z)\,\gamma_r(\rmd z).
    \]
    Hence $\psi-a(r)\operatorname{id}$ is a dual optimizer of $\MCov(\hat\pi_r,\gamma_r)$.
    Let $g$ be the left derivative of $\psi^*$ on $J:=\operatorname{int}\dom\psi^*$, extended monotonically to $\R$ by the values $-\infty$ and $+\infty$ to the left and right of $J$, respectively. The preceding equality implies that $\psi^*(a(r)+z)$ is finite for $\gamma_r$-a.e.~$z$. By \ref{aspt:atomless}, $a(r)+z$ therefore belongs to $J$ and avoids the countable set of nondifferentiability points of $\psi^*$ for $\gamma_r$-a.e.~$z$. Equality in the preceding covariance duality thus yields, for $\kappa$-a.e.~$r$,
    \begin{equation}\label{eq:proof:bridge.existence:char.R-Bass}
        \bigl(g(a(r)+\cdot)\bigr)_\#\gamma_r
        =
        \hat\pi_r.
    \end{equation}
    Set $A:=a(R_0)$ and $X_t:=\E[g(A+Y_1)\mid A,R_t]$.
    By \eqref{eq:proof:bridge.existence:char.R-Bass}, $\Law(R_0,X_1)=\hat\pi$ and $g(A+Y_1)\in L^2$. Since $\int y \, \hat\pi_r(\rmd y)=x$, we have $X_0=Y_0$ a.s., and $X_1\sim\nu$.
    Thus $X\in\Acal(\kappa,\nu)$ and is an $R$-Bass martingale. Moreover, $\Law(X_1,Y_1\mid R_0=r)$ is the monotone coupling of $\hat\pi_r$ and $\gamma_r$, and hence
    \[
        \E[X_1Y_1]
        =
        \int\MCov(\hat\pi_r,\gamma_r)\,\kappa(\rmd r)
        =
        P(\kappa,\nu).
    \]
    In particular, $X$ attains \eqref{eq:wmot.bridge}.
    
    \medskip
    
    Conversely, let $X\in\Acal(\kappa,\nu)$ be an $R$-Bass martingale, generated by $g$ and $A=a(R_0)$, and set $\pi:=\Law(R_0,X_1)$.
    Let $h$ be the lower semicontinuous convex extension of a primitive of $g$ on $J:=\operatorname{int}\{g\in\R\}$, with $h=+\infty$ outside $\bar J$. By \ref{aspt:atomless} and finiteness of $g(A+Y_1)$, $J$ is nonempty and $a(r)+z\in J$ for $\gamma_r$-a.e.~$z$, for $\kappa$-a.e.~$r$.
    Writing $y=g(a(r)+z)$, we have $y\in\partial h(a(r)+z)$ and hence
    \[
        h^*(y)+h(a(r)+z)=(a(r)+z)y.
    \]
    In particular, $h^*$ is finite $\nu$-a.e., and hence on $I_\nu$. Define $\psi$ to agree with $h^*$ on $\bar I_\nu$ and to equal $+\infty$ elsewhere. Then $\psi$ is a dual potential. Since $y\in\bar I_\nu$ almost surely, the preceding equality implies $\psi^*(a(r)+z)=h(a(r)+z)$ and
    \begin{equation}\label{eq:proof:bridge.converse.fenchel}
        \psi(y)+\psi^*(a(r)+z)=(a(r)+z)y,
        \qquad y\in\partial\psi^*(a(r)+z).
    \end{equation}

    Fix $r=(x,\ell)$ outside a $\kappa$-null set. The product $(a(r)+z)y$ is $\gamma_r$-integrable by the conditional second-moment bounds. Since $\psi$ and $\psi^*$ admit affine lower bounds, \eqref{eq:proof:bridge.converse.fenchel} shows that both terms on its left-hand side are integrable. Integrating the subgradient inequality and using $\int y\,\pi_r(\rmd y)=x$ gives
    \[
        x\in\partial(\psi^*\star\gamma_r)(a(r)),
        \qquad
        \psi^C(r)=a(r)x-\int\psi^*(a(r)+z)\,\gamma_r(\rmd z).
    \]
    Moreover, the increasing map $z\mapsto g(a(r)+z)$ induces the monotone coupling of $\gamma_r$ and $\pi_r$. Integrating \eqref{eq:proof:bridge.converse.fenchel} therefore yields
    \[
        \MCov(\pi_r,\gamma_r)
        =
        \int\psi\,\rmd\pi_r-\psi^C(r).
    \]
    By \Cref{lem:properties.Ecal}, we may evaluate $\Ecal(\psi)$ using this particular $\pi$. Hence
    \[
        \int\MCov(\pi_r,\gamma_r)\,\kappa(\rmd r)
        =
        \Ecal(\psi).
    \]
    Since the static and dual values coincide, $\psi$ is a dual minimizer and $\pi$ is the static maximizer. Finally, conditionally on $R_0=r$, the pair
    $(X_1,Y_1)$ is precisely the monotone coupling above, and therefore
    \[
        \E[X_1Y_1]
        =
        \int\MCov(\pi_r,\gamma_r)\,\kappa(\rmd r)
        =
        P(\kappa,\nu).
    \]
    Thus $X$ attains \eqref{eq:wmot.bridge}.
    
    \medskip
    
    It remains to prove uniqueness of the dual minimizer under the additional assumption that $\supp \gamma_r$ is either $\R$ or $[0,\infty)$.
    Let $v$ be another minimizer of \eqref{eq:wmot.bridge.dual}. By \Cref{lem:c-conjugate.differentiable}, $b(r):=\partial_x(v^*\star\gamma_r)^*(x)$ is well defined for $\kappa$-a.e.~$r=(x,\ell)$. As above,
    \[
        \bigl(\partial_x v^*(b(r)+\cdot)\bigr)_\#\gamma_r
        =
        \hat\pi_r
        =
        \bigl(\partial_x\psi^*(a(r)+\cdot)\bigr)_\#\gamma_r.
    \]
    By \Cref{lem:monotone.transport.uniqueness}, $\psi-v$ is affine on $I_r:=\operatorname{int}\operatorname{co}(\supp\hat\pi_r)$.
    Fix $y\in I_\nu$. Irreducibility gives
    \[
        0
        <
        \int|z-y|\,\nu(\rmd z) - \int|z-y|\,\mu(\rmd z)
        =
        \int \Bigl( \int|z-y|\,\hat\pi_{x,\ell}(\rmd z)-|x-y| \Bigr) \kappa(\rmd x,\rmd\ell).
    \]
    The integrand is nonnegative by Jensen's inequality, so it is strictly positive on a set of positive $\kappa$-measure. For every such $r$, $\hat\pi_r(-\infty,y) > 0$ and $\hat\pi_r(y,\infty) > 0$ as otherwise $z \mapsto |z-y|$ is affine on its support and Jensen's inequality is an equality. Therefore $y\in I_r$ for some $r$, and thus every $y\in I_\nu$ has a neighbourhood on which $\psi-v$ is affine. Since $I_\nu$ is an interval, $\psi-v$ is affine on $I_\nu$. By convexity and lower semicontinuity, this identifies the two dual optimizers on $\co\supp(\nu)$ up to addition of an affine function.
\end{proof}

\section{Multi-asset calibration with SKR}
\label{sec:multi.SKR}
\label{subsec:bridge.multi}

Let $R=(Y,L)$ be a reference process on $\Rd\times\mathsf E$, Markov with respect to the reference filtration $(\mathcal F_t)$, whose price coordinates $Y=(Y^1,\ldots,Y^d)$ form a square-integrable martingale.
We fix one calibration interval, rescale it to $[0,1]$, and start $R$ with law $\kappa\in\cP_2(\Rd\times\mathsf E)$.
As in the single-asset case, $(\mathcal F_t)$ is the usual augmentation of the natural filtration generated by $R$.
Write $\mu^j:=\pr^{X^j}_\#\kappa$ and let $\nu^1,\ldots,\nu^d\in\cP_2(\R)$ be the prescribed terminal marginals.

The class $\Acal(\kappa,\nu^1,\ldots,\nu^d)$ consists of square-integrable $\Rd$-valued martingales $X$ on the same filtered probability space as $R$, satisfying $X_0=Y_0$ and $X_1^j\sim\nu^j$ for every $j=1,\ldots,d$.
We emphasize that only the 1-dimensional terminal marginals are prescribed, but not the joint law of $X_1$.
As before, minimizing $\E[\sum_{j=1}^d[X^j-Y^j]_1]$ is equivalent to
\begin{equation}\label{eq:multi-bridge}\tag{multi-Bridge}
    P(\kappa,\nu^1,\ldots,\nu^d)
    :=
    \sup_{X\in\Acal(\kappa,\nu^1,\ldots,\nu^d)}
    \E[X_1\cdot Y_1].
\end{equation}
For $r=(x^1,\ldots,x^d,\ell) \in \R^d \times \mathsf E$, write
\[
    \gamma_r:=\Law(Y_1\mid R_0=r),
    \qquad
    \gamma_r^j:=\Law(Y_1^j\mid R_0=r).
\]
Thus the reference law and its one-dimensional projections are conditioned on the full initial state.

For an admissible bridge $X$, the coupling $\pi^j:=\Law(R_0,X_1^j)$ has disintegration $\pi^j(\rmd r,\rmd y)=\kappa(\rmd r)\pi_r^j(\rmd y)$ with
\[
    \int y\,\pi_r^j(\rmd y)=x^j
    \quad\text{for $\kappa$-a.e.~$r$.}
\]
We denote by $\Mcal^j(\kappa,\nu^j)$ the set of couplings in $\Cpl(\kappa,\nu^j)$ satisfying this martingale property.
For fixed $\kappa$, both the objective and the terminal constraints separate across assets.
We therefore apply the single-asset construction to each coordinate, including the remaining reference price coordinates in the latent state, to obtain the following result.

\begin{theorem}[Existence and uniqueness for the primal problem]
\label{thm:multi-bridge.existence}
    Suppose that \ref{aspt:lsc}, \ref{aspt:moments} and \ref{aspt:atomless} hold coordinatewise, and that $\mu^j\le_{\rm cvx}\nu^j$ for every $j=1,\ldots,d$. Then
    \begin{equation}\label{eq:multi-bridge.weak}\tag{multi-wBridge}
        P(\kappa,\nu^1,\ldots,\nu^d)
        =
        \sum_{j=1}^d
        \max_{\pi^j\in\Mcal^j(\kappa,\nu^j)}
        \int\MCov(\pi_r^j,\gamma_r^j)\,\kappa(\rmd r).
    \end{equation}
    Every single-asset problem admits a unique optimizer $\hat\pi^j$. The dynamic problem admits an optimizer, which is unique in law.
\end{theorem}

\begin{proof}
    Apply \Cref{thm:wmot.bridge.existence.uniqueness} to each single-asset problem, and denote its unique static optimizer by $\hat\pi^j$.
    Each coordinate of an admissible multi-asset bridge is admissible for the corresponding single-asset problem, so
    \begin{equation}\label{ineq:proof:multi-bridge.existence:trivial}
        P(\kappa,\nu^1,\ldots,\nu^d)
        \leq
        \sum_{j=1}^d
        \int\MCov(\hat\pi_r^j,\gamma_r^j)\,\kappa(\rmd r).
    \end{equation}
    On the common filtered probability space, define
    \[
        T_r^j:=F^{-1}_{\hat\pi_r^j}\circ F_{\gamma_r^j},
        \qquad
        X_t^j:=\E[T_{R_0}^j(Y_1^j)\mid\mathcal F_t],
        \qquad j=1,\ldots,d.
    \]
    By \ref{aspt:atomless}, $\Law(X_1^j\mid R_0=r)=\hat\pi_r^j$. The barycenter constraints give $X_0=Y_0$, and each coordinate has terminal law $\nu^j$ and attains its single-asset problem.
    Thus $X=(X^1,\ldots,X^d)$ is admissible and attains the bound in \eqref{ineq:proof:multi-bridge.existence:trivial}.

    Finally, let $\tilde X$ be another optimizer.
    It must attain every single-asset optimum, so $\Law(R_0,\tilde X_1^j)=\hat\pi^j$ for every $j$.
    Its conditional terminal coupling with $Y_1^j$ must also be optimal. By \ref{aspt:atomless}, this is the unique monotone coupling, and hence
    \[
        \tilde X_1^j=T_{R_0}^j(Y_1^j)=X_1^j
        \quad\text{a.s.,}
        \qquad j=1,\ldots,d.
    \]
    Since $X$ and $\tilde X$ are martingales with respect to the same filtration, their c\`adl\`ag versions are indistinguishable.
    Hence, the optimizer is unique in law.
\end{proof}

We next characterize optimal bridges as multi-$R$-Bass martingales through the coordinatewise dual problems.
Write $I_{\nu^j}:=\operatorname{int}\co(\supp\nu^j)$. A dual potential for $\nu^j$ is a proper, convex and lower semicontinuous function $v$ that is finite on $I_{\nu^j}$ and equal to $+\infty$ outside $\bar I_{\nu^j}$.
For such a potential, set
\[
    (v^*\star\gamma_r^j)(a):=\int v^*(a+z)\,\gamma_r^j(\rmd z),
    \qquad
    v^{C,j}(r):=(v^*\star\gamma_r^j)^*(x^j).
\]
For an irreducible pair $(\mu^j,\nu^j)$, the corresponding relaxed objective is
\[
    \Ecal^j(v)
    :=
    \int\Bigl(
        \int v\,\rmd\pi_r^j
        -
        v^{C,j}(r)
    \Bigr)\,\kappa(\rmd r),
    \qquad \pi^j\in\Mcal^j(\kappa,\nu^j).
\]
By Lemma~\ref{lem:properties.Ecal}, $\Ecal^j(v)$ is well-defined, independent of $\pi^j$ and invariant under addition of affine functions.

\begin{definition}[Multi-$R$-Bass martingale]\label{def:multi.R.Bass}
    An admissible martingale $X\in\Acal(\kappa,\nu^1,\ldots,\nu^d)$ is called a \emph{multi-$R$-Bass martingale}, or a \emph{separable $R$-Bass martingale}, if there exist increasing functions $g^1,\ldots,g^d:\R\to[-\infty,\infty]$ and measurable shifts $a^j:\Rd\times\mathsf E\to\R$ such that
    \[
        (g^j\star\gamma_r^j)(a^j(r))=x^j
        \quad\text{for $\kappa$-a.e.~$r$,}
    \]
    and, writing $A^j:=a^j(R_0)$ and $A:=(A^1,\ldots,A^d)$, we have
    \begin{equation}\label{eq:multi.R.Bass.definition}
        X_t^j
        =
        \E[g^j(A^j+Y_1^j)\mid\mathcal F_t]
        =
        \E[g^j(A^j+Y_1^j)\mid A,R_t],
        \qquad 0\leq t\leq1,
    \end{equation}
    for every $j=1,\ldots,d$.
\end{definition}

The shifts enforce the martingale property. 
As in the single-asset case, we write $a^j(r)=(g^j\star\gamma_r^j)^{-1}(x^j)$, where the inverse denotes a measurable choice of the inverse.
The $L^2$ condition ensures that the convolution at this shift is absolutely convergent for $\kappa$-a.e.~$r$.
The second equality in \eqref{eq:multi.R.Bass.definition} follows from the Markov property of the reference process.
Although each generating function $g^j$ is one-dimensional, its shift $a^j$ may depend on all components of the initial state.

\begin{theorem}[Duality and characterization of multi-$R$-Bass martingales]
\label{thm:multi-bridge.duality}
    Suppose that \ref{aspt:lsc}, \ref{aspt:moments} and \ref{aspt:atomless} hold coordinatewise, and that $(\mu^j,\nu^j)$ is irreducible for every $j=1,\ldots,d$. Then
    \begin{equation}\label{eq:multi-bridge.dual}\tag{multi-dBridge}
        P(\kappa,\nu^1,\ldots,\nu^d)
        =
        \sum_{j=1}^d
        \min_{v^j\ {\rm dual\ potential}}
        \Ecal^j(v^j).
    \end{equation}
    An admissible martingale solves \eqref{eq:multi-bridge} if and only if it is a multi-$R$-Bass martingale.

    Moreover, suppose that $\supp\gamma_r^j$ is either $\R$ or $[0,\infty)$ for $\kappa$-a.e.~$r$ and every $j$. Then each dual minimizer is unique up to addition of an affine function on $I_{\nu^j}$.
\end{theorem}

\begin{proof}
    Apply \Cref{thm:wmot.bridge.duality} to each single-asset problem, and denote its static and dual optimizers by $\hat\pi^j$ and $\psi^j$. Together with \Cref{thm:multi-bridge.existence}, this gives
    \[
        P(\kappa,\nu^1,\ldots,\nu^d)
        =
        \sum_{j=1}^d
        \int\MCov(\hat\pi_r^j,\gamma_r^j)\,\kappa(\rmd r)
        =
        \sum_{j=1}^d\Ecal^j(\psi^j).
    \]
    This proves \eqref{eq:multi-bridge.dual}.

    Let $a^j$ be the shifts obtained in the single-asset construction and set $A:=(a^1(R_0),\ldots,a^d(R_0))$.
    On the common filtered probability space, define
    \[
        X_t^j
        :=
        \E[\partial_x(\psi^j)^*(A^j+Y_1^j)\mid\mathcal F_t],
        \qquad j=1,\ldots,d.
    \]
    Each coordinate has terminal law $\nu^j$ and attains its single-asset problem.
    Thus $X=(X^1,\ldots,X^d)$ is admissible and attains the bound in \eqref{ineq:proof:multi-bridge.existence:trivial}.
    The Markov property gives the representation \eqref{eq:multi.R.Bass.definition}, with $g^j=\partial_x(\psi^j)^*$.

    An admissible bridge attains this bound precisely when each coordinate attains its single-asset optimum.
    The single-asset Bass characterization therefore gives the claimed equivalence.

    Under the additional support assumption, uniqueness of each dual minimizer follows from \Cref{thm:wmot.bridge.duality}.
\end{proof}

We observe that the optimizer has joint conditional terminal law
\[
    \Law(X_1\mid R_0=r)
    =
    (G_r)_\#\gamma_r,
    \qquad
    G_r(z^1,\ldots,z^d):=\bigl(g^1(a^1(r)+z^1),\ldots,g^d(a^d(r)+z^d)\bigr).
\]
Thus the single-asset reduction does not impose independence and the assets are coupled through the reference process.

\paragraph{Coordinatewise martingale Sinkhorn iteration.}
For fixed $\kappa$, the martingale Sinkhorn iteration applies separately to each coordinate.
For a shift $a^j:\Rd\times\mathsf E\to\R$, define
\begin{equation}\label{eq:multi.shifted.reference.law}
    \rho^{a^j}
    :=
    \int(a^j(r)+\,\cdot\,)_\#\gamma_r^j\,\kappa(\rmd r)
    =
    \Law(a^j(R_0)+Y_1^j).
\end{equation}
Starting from shifts $a_0^1,\ldots,a_0^d\in L^2(\kappa)$, the coordinate updates, whenever well defined, are
\begin{align*}
    g_k^j
    &:=
    F_{\nu^j}^{-1}\circ F_{\rho^{a_k^j}},\\
    a_{k+1}^j(r)
    &:=
    (g_k^j\star\gamma_r^j)^{-1}(x^j),
    \qquad j=1,\ldots,d.
\end{align*}
The first update chooses the increasing transport to $\nu^j$; the second enforces the martingale constraint.
The coordinate updates can be performed in parallel, using the kernels $\gamma_r^j$ conditioned on the full initial state, and no $d$-dimensional transport problem needs to be solved.
The additional assumptions and normalizations for convergence are given in \Cref{sec:MSA.convergence} and apply asset-wise.

\paragraph{Calibration across maturities.}
For maturities $0=T_0<T_1<\cdots<T_n$, prescribe marginal laws in $\cP_2(\R)$ satisfying
\[
    \delta_{x_0^j}
    =
    \mu_0^j
    \leq_{\rm cvx}\cdots\leq_{\rm cvx}
    \mu_n^j,
    \qquad j=1,\ldots,d.
\]
After calibration up to $T_i$, restart the reference dynamics from the full joint state law
\[
    \kappa_i:=\Law(X_{T_i},L_{T_i}).
\]
On $[T_i,T_{i+1}]$, solve the multi-bridge problem with starting law $\kappa_i$ and terminal marginals $\mu_{i+1}^1,\ldots,\mu_{i+1}^d$.
The joint evolution of the calibrated bridge and the reference latent process determines $\kappa_{i+1}=\Law(X_{T_{i+1}},L_{T_{i+1}})$, which supplies the initial state law for the next step.

\section{Convergence of the Martingale Sinkhorn algorithm}\label{sec:MSA.convergence}

We next study convergence of the Martingale Sinkhorn algorithm. Throughout this section, we assume that $\mu\le_{\rm cvx}\nu$ and that $(\mu,\nu)$ is irreducible. In addition to \ref{aspt:lsc}, \ref{aspt:moments} and \ref{aspt:atomless}, we impose the following assumptions:
\begin{enumerate}[label={\rm(A\arabic*)},start=4]
    \item \label{aspt:MSA.support}
    For $\kappa$-a.e.~$r$, $\supp \gamma_r$ is either $\R$ or $[0,\infty)$.

    \item \label{aspt:MSA.compact}
    $\supp\nu$ is compact, and $\supp \mu \subset I_\nu = \operatorname{int} \co\supp(\nu)$.
\end{enumerate}
Assumption~\ref{aspt:MSA.compact} gives the uniform bounds on the shifts needed in the convergence proof. For related estimates without compact support in the Brownian case, see \cite{HaJoLoObPa25}.
For an irreducible pair $(\mu,\nu)$, the additional condition $\supp\mu\subset I_\nu$, equivalently $\co\supp\mu\subset I_\nu$, is often referred to as \emph{strict irreducibility}.

The iteration is as follows.

\begin{algorithm}[h]
    \caption{Martingale Sinkhorn algorithm}
    \label{alg:MSA}
    \begin{algorithmic}[1]
    \State Choose an initial shift $a_0 \in L^2(\kappa)$ with $\E[a_0(R_0)]=0$.
    \For{$i=1,2,\ldots$}
        \medskip
        \State \textbf{Terminal-marginal update:}
        \Statex\hspace{\algorithmicindent}Set $\rho_i:=\Law\bigl(a_{i-1}(R_0)+Y_1\bigr)$, and choose a dual potential $\psi_i$, satisfying the support convention above, such that $(\partial_x\psi_i^*)_\#\rho_i=\nu$.
        \Statex\hspace{\algorithmicindent}Normalize by $\psi_i\gets\psi_i-\int(\psi_i^*\star\gamma_r)^*(\bar\mu)\,\kappa(\rmd r)$.
        \medskip
    
        \State \textbf{Martingale update:}
        \Statex\hspace{\algorithmicindent}For $r=(x,\ell)$, set $a_i(r):=\partial_x(\psi_i^*\star\gamma_r)^*(x)$.
        \Statex\hspace{\algorithmicindent}Normalize by $a_i(r)\gets a_i(r)-\E[a_i(R_0)]$.
        \medskip
    \EndFor
    \end{algorithmic}
\end{algorithm}

The terminal-marginal update enforces the prescribed terminal law. In one dimension it is simply the monotone rearrangement from $\rho_i$ to $\nu$:
\[
    \partial_x\psi_i^*
    =
    F_\nu^{-1}\circ F_{\rho_i}.
\]
Here we use $F_\nu^{-1}(0):=\min\supp\nu$ and $F_\nu^{-1}(1):=\max\supp\nu$. In particular, if $\supp\rho_i$ is bounded below, the rearrangement is extended constantly by $\min\supp\nu$ below that support.
Thus this step is explicit once $\rho_i$ is known. The potential $\psi_i$ itself is determined only up to an additive constant. Since the implementation uses only $\partial_x\psi_i^*$, this constant has no effect on the iteration; the normalization in \Cref{alg:MSA} merely fixes a convenient representative.
The martingale update then restores the martingale constraint. Write
\[
    b_i(r):=\partial_x(\psi_i^*\star\gamma_r)^*(x),
    \qquad
    c_i:=\int b_i\,\rmd\kappa,
    \qquad
    a_i=b_i-c_i.
\]
Before centering the shift, Fenchel--Moreau gives, for $\kappa$-a.e.~$r=(x,\ell)$,
\[
    x
    =
    \int \partial_x\psi_i^*(b_i(r)+z)\,\gamma_r(\rmd z),
\]
so the conditional law of $\partial_x\psi_i^*(b_i(r)+Y_1)$ given $R_0=r$ has barycenter $x$.

The conditional laws depend on the shift and the potential through the map
\[
    z\longmapsto\partial_x\psi_i^*(b_i(r)+z).
\]
Replacing $b_i$ by $a_i=b_i-c_i$ subtracts $c_i$ from the argument of this map.
We compensate by replacing $\psi_i$ with
\[
    \tilde\psi_i(y):=\psi_i(y)-c_i(y-\bar\mu),
\]
which gives $\partial_x\tilde\psi_i^*(u)=\partial_x\psi_i^*(u+c_i)$.
The two changes therefore cancel, leaving the conditional laws unchanged.
The dual value is also unchanged because the correction to $\psi_i$ is affine; see \Cref{lem:properties.Ecal}.

We are allowed to center the shifts because the dual objective is invariant under addition of affine functions.
Indeed, if $c(y)=sy+d$, then $(\psi+c)^*(z)=\psi^*(z-s)-d$; see \Cref{lem:properties.Ecal} and the proof thereof.
Thus the potential paired with the centered shift is $\tilde\psi_i(y):=\psi_i(y)-c_i(y-\bar\mu)$, since
\[
    \partial_x\tilde\psi_i^*(a_i(r)+z)
    =
    \partial_x\psi_i^*(b_i(r)+z).
\]
The pairs $(\tilde\psi_i,a_i)$ and $(\psi_i,b_i)$ therefore represent the same conditional laws, and affine invariance gives $\Ecal(\tilde\psi_i)=\Ecal(\psi_i)$. The algorithm retains $\psi_i$ as the terminal-marginal potential and uses the centered shift $a_i$ in the next iteration.

The main result of this section is the convergence of these iterates to the solution of the dual problem~\eqref{eq:wmot.bridge.dual}.

\begin{theorem}[Convergence of the Martingale Sinkhorn algorithm]\label{thm:MSA.convergence}
    Let $(\mu,\nu)$ be irreducible. Suppose further that \ref{aspt:lsc}, \ref{aspt:moments}, \ref{aspt:atomless},
    \ref{aspt:MSA.support} and \ref{aspt:MSA.compact} hold.
    Then $\Ecal(\psi_i)\downarrow P(\kappa,\nu)$, and the potentials $(\psi_i)_i$ generated by \Cref{alg:MSA} converge locally uniformly on $I_\nu$ to a dual minimizer of~\eqref{eq:wmot.bridge.dual}.
\end{theorem}

By \Cref{thm:wmot.bridge.duality}, the dual minimizer is unique up to addition of an affine function. The normalizations in \Cref{alg:MSA} select a unique limiting representative, as shown in the proof below. We denote this normalized dual minimizer by $\psi$. Setting
\[
    g:=\partial_x\psi^*,
    \qquad
    a(r):=\partial_x(\psi^*\star\gamma_r)^*(x),
    \quad r=(x,\ell),
\]
this yields the $R$-Bass martingale $X_t=\E\bigl[g(a(R_0)+Y_1)\mid a(R_0),R_t\bigr]$, and hence the unique optimal bridge in law; see \Cref{thm:wmot.bridge.existence.uniqueness}.

The main mechanism behind the convergence proof is the following observation: each Martingale Sinkhorn iteration strictly decreases $\Ecal$ unless the current potential already attains \eqref{eq:wmot.bridge.dual}. Before proving this, we record the regularity properties of the iterates that will be used throughout the argument.

\begin{lemma}[Regularity of the iterates]\label{lem:sinkhorn.integrability}
    Suppose that the assumptions of \Cref{thm:MSA.convergence} hold. Let
    \[
        L:=\max\{|y|:y\in\supp\nu\},
        \qquad
        m(r):=\int|y|\,\gamma_r(\rmd y),
        \qquad
        \bar m:=\int m(r)\,\kappa(\rmd r).
    \]
    Then $\psi_i^*$ is $L$-Lipschitz for every $i$, and for $u\in I_\nu$ and $\kappa$-a.e.~$r$,
    \begin{equation}\label{eq:regularity.transform.bound}
        \bigl|(\psi_i^*\star\gamma_r)^*(u)-\psi_i(u)\bigr| \le Lm(r).
    \end{equation}
    Moreover, there exist a compact set $K\subset I_\nu$ and $\delta>0$, independent of $i$, such that for $\kappa$-a.e.~$r$,
    \begin{equation}\label{inequ:regularity.iterates}
        |a_i(r)|
        \le
        \frac{1}{\delta} \Bigl( 4\sup_{u\in K}|\psi_i(u)| + 2L(m(r)+\bar m) \Bigr).
    \end{equation}
    Consequently, $a_i\in L^2(\kappa)$, $\rho_{i+1}\in\cP_2(\R)$ is atomless, and $\psi_i\in L^1(\nu)$, $\psi_i^C\in L^1(\kappa)$, and $\psi_i^*,\psi_{i+1}^*\in L^1(\rho_{i+1})$.
    In particular, $\Ecal(\psi_i)<\infty$.
\end{lemma}

\begin{proof}
    By the support convention, $\dom\psi_i\subset\co\supp\nu\subset[-L,L]$, so the conjugate $\psi_i^*$ is $L$-Lipschitz. Hence, for every $a\in\R$ and $\kappa$-a.e.~$r$,
    \[
        \bigl|(\psi_i^*\star\gamma_r)(a)-\psi_i^*(a)\bigr| \le Lm(r).
    \]
    Taking conjugates gives \eqref{eq:regularity.transform.bound}. 
    Since $\supp\mu$ is a compact subset of $I_\nu$, choose $\delta>0$ such that $K:=\supp\mu+[-\delta,\delta]\subset I_\nu$.
    Set $M_i:=\sup_{u\in K}|\psi_i(u)|$. By \eqref{eq:regularity.transform.bound},
    \[
        \bigl|(\psi_i^*\star\gamma_r)^*(u)\bigr|
        \le M_i+Lm(r),
        \qquad u\in K.
    \]
    In particular, taking $u=x$ shows that $\psi_i^C\in L^1(\kappa)$.
    
    The shift $\tilde a_i(r):=\partial_x(\psi_i^*\star\gamma_r)^*(x)$ before normalization is well defined by \Cref{lem:c-conjugate.differentiable}. The subgradient inequality at $x\pm\delta$ gives
    \[
        \delta|\tilde a_i(r)| \le 2M_i+2Lm(r).
    \]
    Since $a_i(r)=\tilde a_i(r)-\int\tilde a_i(e)\,\kappa(\rmd e)$ we obtain
    \[
        \delta|a_i(r)| \le 4M_i+2L\bigl(m(r)+\bar m\bigr),
    \]
    which is \eqref{inequ:regularity.iterates}. 
    By Jensen's inequality and \ref{aspt:moments}, $m\in L^2(\kappa)$, and hence $a_i\in L^2(\kappa)$. It follows that $\rho_{i+1}=\Law(a_i(R_0)+Y_1)\in\cP_2(\R)$. 
    Moreover, $\rho_{i+1}(\{z\}) = \int\gamma_r(\{z-a_i(r)\})\,\kappa(\rmd r) = 0$, so $\rho_{i+1}$ is atomless.
    
    Since $\psi_j^*$ is $L$-Lipschitz,
    \[
        |\psi_j^*(z)| \le |\psi_j^*(0)|+L|z|,
    \]
    and therefore $\psi_i^*,\psi_{i+1}^*\in L^1(\rho_{i+1})$. By construction, $(\partial_x\psi_{i+1}^*)_\#\rho_{i+1}=\nu$, and Fenchel--Moreau yields
    \begin{align*}
        \int|\psi_{i+1}|\,\rmd\nu
        &=
        \int
        \bigl|z\partial_x\psi_{i+1}^*(z)-\psi_{i+1}^*(z)\bigr|
        \,\rho_{i+1}(\rmd z)
        \le
        L\int|z|\,\rho_{i+1}(\rmd z)
        +
        \int|\psi_{i+1}^*|\,\rmd\rho_{i+1}
        <\infty.
    \end{align*}
    Since $a_0\in L^2(\kappa)$, the same argument applies to $\rho_1=\Law(a_0(R_0)+Y_1)$ and gives $\psi_1\in L^1(\nu)$. Thus $\psi_i\in L^1(\nu)$ for every $i$. Together with $\psi_i^C\in L^1(\kappa)$, this also shows that $\Ecal(\psi_i)<\infty$.
\end{proof}

\begin{lemma}[Strict descent]\label{lem:strict.descent}
    Under the assumptions of \Cref{thm:MSA.convergence}, $\Ecal(\psi_{i+1}) \le \Ecal(\psi_i)$ for every $i \in \N$.
    Moreover, equality holds if and only if $\psi_i$ attains \eqref{eq:wmot.bridge.dual}.
\end{lemma}

\begin{proof}
    We follow the argument of \cite[Lemma 3.5]{HaJoLoObPa25}. Fix $i\in\N$. By \Cref{lem:sinkhorn.integrability}, all integrals below are finite.
    Set
    \[
        c_i
        :=
        \int \partial_x(\psi_i^*\star\gamma_r)^*(x)\,\kappa(\rmd r),
        \qquad
        \tilde\psi_i(y)
        :=
        \psi_i(y)-c_i(y-\bar\mu).
    \]
    By construction of the normalized shift, $a_i(r) = \partial_x(\tilde\psi_i^*\star\gamma_r)^*(x)$ for $\kappa$-a.e.~$r=(x,\ell)$, while affine invariance yields $\Ecal(\tilde\psi_i)=\Ecal(\psi_i)$. 
    Hence Fenchel--Moreau and the definition of $\rho_{i+1}$ give
    \[
        \Ecal(\psi_i)
        =
        \underbrace{\Bigl( \int\tilde\psi_i\,\rmd\nu + \int\tilde\psi_i^*\,\rmd\rho_{i+1} \Bigr)}_{\displaystyle \mathrm{(I)}}
        -
        \underbrace{\int xa_i(r)\,\kappa(\rmd r)}_{\displaystyle \mathrm{(II)}}.
    \]
    Since $(\partial_x\psi_{i+1}^*)_\#\rho_{i+1}=\nu$,
    \begin{equation}\label{eq:strict.descent.I}
        \mathrm{(I)}
        \ge
        \MCov(\nu,\rho_{i+1})
        =
        \int\psi_{i+1}\,\rmd\nu
        +
        \int\psi_{i+1}^*\,\rmd\rho_{i+1}.
    \end{equation}
    Moreover, Fenchel's inequality gives
    \begin{equation}\label{eq:strict.descent.II}
        \mathrm{(II)}
        \le
        \int\Bigl(
            (\psi_{i+1}^*\star\gamma_r)(a_i(r))
            +
            (\psi_{i+1}^*\star\gamma_r)^*(x)
        \Bigr)\,\kappa(\rmd r).
    \end{equation}
    Because $\int\psi_{i+1}^*\,\rmd\rho_{i+1} = \int (\psi_{i+1}^*\star\gamma_r)(a_i(r)) \,\kappa(\rmd r)$, combining \eqref{eq:strict.descent.I} and \eqref{eq:strict.descent.II} yields
    \[
        \Ecal(\psi_i)\ge\Ecal(\psi_{i+1}).
    \]
    
    Suppose that equality holds. Then equality holds in \eqref{eq:strict.descent.I}, so $\tilde\psi_i$ also attains the dual problem for $\MCov(\nu,\rho_{i+1})$. Since $\rho_{i+1}$ is atomless, $(\partial_x\tilde\psi_i^*)_\#\rho_{i+1}=\nu$.
    On the other hand, by the definition of $a_i$ and \Cref{lem:sinkhorn.integrability},
    \[
        x
        =
        \partial_x(\tilde\psi_i^*\star\gamma_r)(a_i(r))
        =
        \int
        \partial_x\tilde\psi_i^*(a_i(r)+z)\,
        \gamma_r(\rmd z)
    \]
    for $\kappa$-a.e.~$r=(x,\ell)$. Thus $\partial_x\tilde\psi_i^*$ and $a_i$ generate an admissible $R$-Bass martingale, and \Cref{thm:wmot.bridge.duality} shows that $\tilde \psi_i$ attains \eqref{eq:wmot.bridge.dual}. Since $\Ecal(\tilde \psi_i ) = \Ecal(\psi_i)$, $\psi_i$ attains \eqref{eq:wmot.bridge.dual} too.
    Conversely, if $\psi_i$ attains \eqref{eq:wmot.bridge.dual}, then
    \[
        P(\kappa,\nu)
        =
        \Ecal(\psi_i)
        \ge
        \Ecal(\psi_{i+1})
        \ge
        P(\kappa,\nu),
    \]
    and hence equality holds.
\end{proof}

We now prove convergence.
The bound on the dual values $(\Ecal(\psi_i))_i$ gives local boundedness of the normalized potentials $(\psi_i)_i$, while the difference between successive dual values controls their Fenchel gaps.
This will identify every accumulation point of $(\psi_i)_i$ as a dual optimizer.

\begin{proof}[Proof of \Cref{thm:MSA.convergence}]
    The proof has three steps. We first obtain relative compactness of the iterates. We then use the vanishing descent gaps to identify every accumulation point as the appropriately normalized dual minimizer. Finally, we pass to the limit in the dual objective.    
    By \Cref{lem:strict.descent}, $(\Ecal(\psi_i))_{i\in\N}$ is decreasing and bounded below by $P(\kappa,\nu)$. Hence
    \begin{equation}\label{eq:MSA.vanishing.descent}
        \Ecal(\psi_i)-\Ecal(\psi_{i+1})\longrightarrow 0.
    \end{equation}

    We first prove relative compactness. Choose $p_i\in\partial\psi_i(\bar\mu)$ and set $\varphi_i(y):=\psi_i(y)-\psi_i(\bar\mu)-p_i(y-\bar\mu)$.
    Then $\varphi_i\ge0$ and $\varphi_i(\bar\mu)=0$. By \cite[Proposition 3.3]{HaJoLoObPa25}, there exists a martingale coupling $\pi'$ of $(\mu,\nu)$ whose conditional laws $\pi'_x$ charge every open set charged by $\nu$. Since affine functions cancel from the corresponding Jensen gap,
    \[
        0
        \le
        \int\Bigl(\int\varphi_i\,\rmd\pi'_x-\varphi_i(x)\Bigr)\,\mu(\rmd x)
        =
        \int\psi_i\,\rmd\nu-\int\psi_i\,\rmd\mu
        \le
        \Ecal(\psi_i)-\int x^2\,\mu(\rmd x).
    \]
    The right-hand side is uniformly bounded by \Cref{lem:strict.descent}. The compactness result \cite[Lemma 3.9]{HaJoLoObPa25} states that normalized convex functions satisfying such a uniform Jensen-gap bound are locally bounded on $I_\nu$. Hence $(\varphi_i)_i$ is locally bounded there.
    
    Retain the notation $L,m,\bar m$ from \Cref{lem:sinkhorn.integrability}. By \eqref{eq:regularity.transform.bound} and the additive normalization in \Cref{alg:MSA},
    $|\psi_i(\bar\mu)|\le L\bar m$. Moreover, applying \eqref{inequ:regularity.iterates} to the affine representatives $\varphi_i$, which induce the same normalized shifts $a_i$, gives
    \begin{equation}\label{eq:MSA.uniform.shift.bound}
        |a_i(r)|\le C(1+m(r))
    \end{equation}
    for some $C>0$, all $i$, and $\kappa$-a.e.~$r$. Thus the second moments of $(\rho_i)_i$ are uniformly bounded.
    
    It remains to control $p_i$. Choose $y_-<\bar\mu<y_+$ with $\nu((-\infty,y_-])>0$ and $\nu([y_+,\infty))>0$. 
    Since $(\partial_x\psi_i^*)_\#\rho_i=\nu$ and $\bar\mu\in\partial\psi_i^*(p_i)$, monotonicity gives
    \[
        \rho_i((-\infty,p_i])\ge\nu((-\infty,y_-]),
        \qquad
        \rho_i([p_i,\infty))\ge\nu([y_+,\infty)).
    \]
    The uniform second-moment bound on $\rho_i$ therefore implies $\sup_i|p_i|<\infty$. Hence $(\psi_i)_i$ is locally bounded on $I_\nu$, and convexity gives local equicontinuity. It is therefore relatively compact for locally uniform convergence.
    
    The same estimates also give a uniform linear-growth bound
    \begin{equation}\label{eq:MSA.conjugate.bound}
        |\psi_i^*(z)|\le C+L|z|,
        \qquad i\in\N,\ z\in\R.
    \end{equation}
    
    Let $\psi_{i_k}\to\psi$ locally uniformly on $I_\nu$. We extend $\psi$ lower semicontinuously to $\bar I_\nu$ and set it equal to $+\infty$ elsewhere. By \cite[Theorem 7.17]{RoWe98}, this implies epi-convergence on $\R$; convex conjugation preserves epi-convergence by \cite[Theorem 11.34]{RoWe98}. Since $\psi^*$ is finite on $\R$, another application of \cite[Theorem 7.17]{RoWe98} yields
    $\psi_{i_k}^*\to\psi^*$ locally uniformly on $\R$.
    
    As in the proof of \Cref{lem:strict.descent}, write
    \[
        b_i(r):=\partial_x(\psi_i^*\star\gamma_r)^*(x),
        \qquad
        c_i:=\int b_i(r)\,\kappa(\rmd r),
        \qquad
        a_i=b_i-c_i.
    \]
    The subgradient estimate in the proof of \Cref{lem:sinkhorn.integrability}, together with the local bounds above, gives
    $|b_i(r)|\le C(1+m(r))$. By \eqref{eq:MSA.conjugate.bound}, dominated convergence gives local uniform convergence
    $\psi_{i_k}^*\star\gamma_r\to\psi^*\star\gamma_r$ for $\kappa$-a.e.~$r$. Conjugation therefore gives locally uniform convergence of $(\psi_{i_k}^*\star\gamma_r)^*$ to $(\psi^*\star\gamma_r)^*$ on $I_\nu$. Hence \Cref{lem:c-conjugate.differentiable} and \cite[Exercise 12.40(b)]{RoWe98} give
    \[
        b_{i_k}(r)\longrightarrow
        b(r):=\partial_x(\psi^*\star\gamma_r)^*(x).
    \]
    Dominated convergence therefore yields $c_{i_k}\to c:=\int b\,\rmd\kappa$ and $a_{i_k}\to a:=b-c$ in $L^2(\kappa)$. Consequently, with $\rho:=\Law(a(R_0)+Y_1)$, $\Wcal_2(\rho_{i_k+1},\rho)\to 0$. The law $\rho$ is atomless.
    Set $\tilde\psi_i(y):=\psi_i(y)-c_i(y-\bar\mu)$ and $\tilde\psi(y):=\psi(y)-c(y-\bar\mu)$. Then $\tilde\psi_{i_k}\to\tilde\psi$ locally uniformly and $a(r)=\partial_x(\tilde\psi^*\star\gamma_r)^*(x)$.
    
    We now identify the limit using the two nonnegative gaps in the proof of \Cref{lem:strict.descent}. Denote by $G_i$ the gap in \eqref{eq:strict.descent.I} and by $H_i$ the Fenchel gap in \eqref{eq:strict.descent.II}. That proof gives
    \[
        \Ecal(\psi_i)-\Ecal(\psi_{i+1})=G_i+H_i,
        \qquad G_i,H_i\ge0.
    \]
    Hence \eqref{eq:MSA.vanishing.descent} implies $G_i,H_i\to0$.
    
    Let $\chi_k\in\Cpl(\nu,\rho_{i_k+1})$ attain
    $\MCov(\nu,\rho_{i_k+1})$. Passing to a further subsequence, assume
    $\chi_k\to\chi\in\Cpl(\nu,\rho)$ weakly. Since
    \[
        G_{i_k}
        =
        \int\bigl(
            \tilde\psi_{i_k}(x)
            +
            \tilde\psi_{i_k}^*(y)
            -
            xy
        \bigr)\,\chi_k(\rmd x,\rmd y),
    \]
    epi-convergence and Fatou's lemma, applied to the nonnegative Fenchel gaps on a Skorokhod representation of the weakly convergent couplings, imply $\tilde\psi(x)+\tilde\psi^*(y)=xy$ for $\chi$-a.e.~$(x,y)$.
    Since $\rho$ is atomless, $(\partial_x\tilde\psi^*)_\#\rho=\nu$.
    On the other hand, the definition of $a$ and the regularity established in \Cref{lem:sinkhorn.integrability} give, for $\kappa$-a.e.~$r=(x,\ell)$,
    \[
        x
        =
        \partial_x(\tilde\psi^*\star\gamma_r)(a(r))
        =
        \int\partial_x\tilde\psi^*(a(r)+z)\,\gamma_r(\rmd z).
    \]
    Thus $\partial_x\tilde\psi^*$ and $a$ generate an admissible $R$-Bass martingale. By \Cref{thm:wmot.bridge.duality}, $\tilde\psi$, and hence $\psi$, attains \eqref{eq:wmot.bridge.dual}.
    
    It remains only to identify the affine representative. We claim that $c_i\to0$. Indeed, from any subsequence we may pass further so that $\psi_i\to u$ and $\psi_{i+1}\to v$ locally uniformly on $I_\nu$. Along this subsequence, the preceding argument gives $a_i\to a$ $\kappa$-a.e.~and in $L^2(\kappa)$, with $\int a\,\rmd\kappa=0$. The same argument applies to $\psi_{i+1}$ and gives locally uniform convergence of the convolutions and their conjugates. Since $H_i\to0$, Fatou's lemma applied to its nonnegative integrand yields
    \[
        0
        \le
        \int\Bigl(
            (v^*\star\gamma_r)(a(r))
            +(v^*\star\gamma_r)^*(x)
            -xa(r)
        \Bigr)\,\kappa(\rmd r)
        \le
        \liminf_i H_i
        =0.
    \]
    Hence Fenchel equality and \Cref{lem:c-conjugate.differentiable} give $a(r)=\partial_x(v^*\star\gamma_r)^*(x)$ for $\kappa$-a.e.~$r=(x,\ell)$. The convergence of the raw shifts in $L^2(\kappa)$ therefore implies $c_{i+1}\to\int a\,\rmd\kappa=0$ along this subsequence. Since the original subsequence was arbitrary, $c_{i+1}\to0$, and thus $c_i\to0$.
    
    Thus every accumulation point $\psi$ satisfies
    \[
        \int(\psi^*\star\gamma_r)^*(\bar\mu)\,\kappa(\rmd r)=0,
        \qquad
        \int\partial_x(\psi^*\star\gamma_r)^*(x)\,\kappa(\rmd r)=0.
    \]
    The first identity follows from the normalization in \Cref{alg:MSA} and dominated convergence, while the second is $c=0$. By \Cref{thm:wmot.bridge.duality}, dual minimizers differ only by an affine function, and these two conditions fix one affine normalization. Hence all accumulation points coincide, so $(\psi_i)_i$ converges locally uniformly on $I_\nu$.
    
    Finally, let $\psi$ denote this limit. Since $c_i\to0$, the preceding argument gives $a_i\to a$ in $L^2(\kappa)$ and $\rho_{i+1}\to\rho$ in $\Wcal_2$, where
    $a(r)=\partial_x(\psi^*\star\gamma_r)^*(x)$ and $\rho=\Law(a(R_0)+Y_1)$.
    Moreover, $\tilde\psi_i^*\to\psi^*$ locally uniformly on $\R$, with a uniform linear-growth bound inherited from \eqref{eq:MSA.conjugate.bound}. Consequently,
    \[
        \int\tilde\psi_i^*\,\rmd\rho_{i+1}
        \longrightarrow
        \int\psi^*\,\rmd\rho.
    \]
    Since $(\partial_x\psi^*)_\#\rho=\nu$, Fenchel--Moreau and the same growth bound give $\psi\in L^1(\nu)$ and
    \[
        \int\psi\,\rmd\nu+\int\psi^*\,\rmd\rho
        =\MCov(\nu,\rho).
    \]
    As $\int\tilde\psi_i\,\rmd\nu=\int\psi_i\,\rmd\nu$, the definition of $G_i$ gives
    \[
        \int\psi_i\,\rmd\nu
        =
        \MCov(\nu,\rho_{i+1})
        -\int\tilde\psi_i^*\,\rmd\rho_{i+1}
        +G_i
        \longrightarrow
        \int\psi\,\rmd\nu,
    \]
    using continuity of maximal covariance in $\Wcal_2$ and $G_i\to0$.
    Moreover, \eqref{eq:regularity.transform.bound}, local uniform convergence on the compact set $\supp\mu\subset I_\nu$, and dominated convergence give
    \[
        \int(\psi_i^*\star\gamma_r)^*(x)\,\kappa(\rmd r)
        \longrightarrow
        \int(\psi^*\star\gamma_r)^*(x)\,\kappa(\rmd r).
    \]
    Therefore $\Ecal(\psi_i)\downarrow\Ecal(\psi)=P(\kappa,\nu)$.
\end{proof}

\section{Comparison with SLV models}\label{sec:SLV_formal2}

We explain formally how the Bass structure of a short SKR bridge leads to the instantaneous calibration principle described in the introduction. Let $(\mu_t)_{0\le t\le T}$ be a smooth family of target marginals, and write $\sigma_{\mathrm{Dup}}$ for the corresponding Dupire diffusion coefficient. Consider a diffusion reference
\[
    \rmd Y_t=\sigma_R(t,Y_t,L_t)\,\rmd W_t,
\]
with an autonomous factor $L$. For $h=T/N$, let $X^h$ be the SKR model calibrated at the maturities $t_j=jh$. The factor and Brownian drivers are retained throughout the construction.

\paragraph{A short Bass bridge.}
Fix a calibration interval $[t,t+h]$ and write
$x=X_t^h$, $l=L_t$, and $\E_t=\E[\cdot\mid\mathcal F_t]$.
On this interval the reference price is restarted at $x$:
\[
    Y_t^{t,h}=x,
    \qquad
    \rmd Y_s^{t,h} =\sigma_R(s,Y_s^{t,h},L_s)\,\rmd W_s.
\]
Its increment $M_{t,h}=Y_{t+h}^{t,h}-x$ has conditional mean zero and, to first order, is given by
$\sigma_R(t,x,l)(W_{t+h}-W_t)$.

Let $G_{t,h}$ be the increasing Bass map for this bridge. For the following informal calculation, assume that the map is smooth and invertible, with the spatial bounds needed for the short-time expansions below. Write
$\xi_{t,h}=G_{t,h}^{-1}$ and $k_{t,h}=\xi_{t,h}'$.
The terminal value of the bridge can then be written as
\begin{equation}\label{eq:slv.formal.centered.bridge}
X_{t+h}^h
=
G_{t,h}\bigl(\xi_{t,h}(x)+b_{t,h}(x,l)+M_{t,h}\bigr),
\qquad
\E_t[X_{t+h}^h]=x.
\end{equation}
Here the shift $b_{t,h}$ is chosen to preserve the conditional mean.

The reference fluctuations are of order $\sqrt h$. Applying the curved map $G_{t,h}$ changes their conditional mean by order $h$, which is compensated by the centering shift. The leading random fluctuation is therefore determined by the slope of $G_{t,h}$. Since
$G_{t,h}'(\xi_{t,h}(x))=1/k_{t,h}(x)$, linearization of \eqref{eq:slv.formal.centered.bridge} suggests
\begin{equation}\label{eq:slv.formal.increment}
X_{t+h}^h-X_t^h
\simeq
\frac{\sigma_R(t,X_t^h,L_t)}
{k_{t,h}(X_t^h)}
(W_{t+h}-W_t).
\end{equation}
Thus the reciprocal inverse Bass slope plays the role of a leverage function.

\paragraph{Identification of the leverage.}
The bridge reproduces the next prescribed marginal exactly. To first order, this imposes the conditional variance constraint \eqref{eq:intro.local.variance}. Hence, if the joint laws and inverse slopes admit a sufficiently regular limit, \eqref{eq:slv.formal.increment} gives
\[
    \frac{\E[\sigma_R(t,x,L_t)^2\mid X_t=x]}{k(t,x)^2} =\sigma_{\mathrm{Dup}}(t,x)^2.
\]
Consequently, the candidate limiting dynamics are
\begin{equation}\label{eq:slv.formal.limit}
\rmd X_t
=\lambda(t,X_t)\sigma_R(t,X_t,L_t)\,\rmd W_t,
\qquad
\lambda(t,x)
=\frac1{k(t,x)}
=\frac{\sigma_{\mathrm{Dup}}(t,x)}
{\sqrt{\E[\sigma_R(t,x,L_t)^2\mid X_t=x]}}.
\end{equation}
This is the SLV prescription \eqref{eq:intro.local.projection}. The reference volatility is evaluated at the calibrated price because each bridge starts there. Propagating the joint law of price and factor supplies the conditional distribution needed to update the leverage.

\paragraph{Accumulation of the errors.}
The martingale condition also explains why a local approximation can be useful over many steps. Once a smooth SLV target $X$ is specified, put
$f(t,x,l)=\lambda(t,x)\sigma_R(t,x,l)$ and consider the endpoint errors
\[
    \epsilon_j
    =
    X_{t_{j+1}}^h-X_{t_j}^h -f(t_j,X_{t_j}^h,L_{t_j})(W_{t_{j+1}}-W_{t_j}).
\]
Each error has conditional mean zero given $\mathcal F_{t_j}$. Errors from different intervals are therefore orthogonal:
\begin{equation}\label{eq:slv.formal.error.accumulation}
\E\left|\sum_{j=0}^{N-1}\epsilon_j\right|^2
=
\sum_{j=0}^{N-1}\E|\epsilon_j|^2.
\end{equation}
Moreover, within each interval the Bass interpolation is the conditional expectation of its terminal value. Endpoint estimates consequently control the intervening paths as well. Together with stability of the target dynamics, these observations reduce convergence to quantitative estimates for the centered bridge endpoints; see \Cref{app:skr:prop:endpoint-criterion}.

Appendix~\ref{app:skr:sec:convergence} carries out this program for a specified SLV target with one autonomous real factor and constant correlation in $(-1,1)$. It compares the exact bridges with smooth centered maps whose inverse slopes are $1/\lambda(t_j,\cdot)$, and controls the difference through quantitative calibration. Under the coefficient and regular-start assumptions stated in Subsection~\ref{app:skr:sec:model}, \Cref{app:skr:thm:main} yields
\[
    \E\sup_{0\le t\le T}|X_t^h-X_t|^2\le C_T h.
\]
This gives strong order one half on the common Brownian drivers. The general factor result requires the stated regular initial joint density. When the reference price coefficient is independent of the factor, \Cref{app:skr:cor:autonomous-price} gives the same rate for any square-integrable initial pair, including a deterministic initial price.

\section{Numerical experiments}
\label{sec:numerics}

We compare Bass LV with the Heston, $3/2$, and one-factor Bergomi SKR models, all calibrated to the same European option prices.
We report the vanilla fits and the prices of three-year reverse cliquets and memory autocallables with quarterly observations.

\subsection{Vanilla calibration}
\label{sec:numerics.vanillas}

We fit European call prices corresponding to the SPX smiles of 23 June 2020 reported by Farkas, Ferrari and Ulrych~\cite{FaFeUl22}, at maturities $(T_1,\ldots,T_5)=(0.25,0.5,1,2,3)$ years.
We work with forward-normalized prices $X_t=S_t/F_t$.

\paragraph{Reference calibration.}
We first fit the Heston, $3/2$, and one-factor Bergomi reference models to these vanilla prices, obtaining the parameters in \Cref{tab:numerics.references}.
For Bergomi, the price volatility is $\sqrt{V_t}$, where
\begin{equation}
    V_t
    =
    \xi_0(t)\exp\Bigl(
        \eta Z_t-\frac{\eta^2}{4\kappa}
        \bigl(1-e^{-2\kappa t}\bigr)
    \Bigr),
    \label{eq:numerics.bergomi}
\end{equation}
with $\rmd Z_t=-\kappa Z_t\,\rmd t+\rmd B_t$ and $Z_0=0$.
The correlation between the price and volatility drivers is $\rho$.
The initial forward variance $\xi_0$ is constant on each interval listed in \Cref{tab:numerics.references}.
Its levels are obtained from successive differences of the total variance implied by log-contract prices, divided by the corresponding interval lengths.
This curve is fixed when fitting the remaining Bergomi parameters.

\begin{table}[H]
    \centering
    \small
    \sisetup{output-decimal-marker={.},group-digits=false}
    \setlength{\tabcolsep}{8pt}
    \renewcommand{\arraystretch}{1.1}

    \begin{tabular}{
        @{}l S[table-format=-1.4]
        @{\hspace{12mm}} l S[table-format=-1.4]
        @{\hspace{12mm}} l S[table-format=-1.4]@{}
    }
        \toprule
        \multicolumn{2}{c}{Heston}
        & \multicolumn{2}{c}{$3/2$}
        & \multicolumn{2}{c}{Bergomi-1F} \\
        \cmidrule(lr){1-2}
        \cmidrule(lr){3-4}
        \cmidrule(l){5-6}
        $\kappa$ &  1.2484
            & $\kappa$ &  0.3600
            & $\kappa$ &  2.4971 \\
        $\theta$ &  0.0988
            & $\theta$ &  5.0050
            & $\eta$ &  3.6191 \\
        $\eta$ &  1.1175
            & $\eta$ &  7.6870
            & $\rho$ & -0.8796 \\
        $\rho$ & -0.8038
            & $\rho$ & -1.0000
            & & \\
        $v_0$ &  0.1020
            & $v_0$ &  0.0877
            & & \\
        \bottomrule
    \end{tabular}

    \medskip

    \begin{tabular}{@{}l *{5}{S[table-format=1.8]}@{}}
        \toprule
        $t$ [years]
            & {$[0,0.25)$}
            & {$[0.25,0.5)$}
            & {$[0.5,1)$}
            & {$[1,2)$}
            & {$[2,3]$} \\
        \midrule
        $\xi_0(t)$
            & 0.10209961
            & 0.11963007
            & 0.09081965
            & 0.08462666
            & 0.08383247 \\
        \bottomrule
    \end{tabular}

    \caption{
        Fitted reference parameters and the piecewise-constant initial forward-variance curve for Bergomi.
        Forward-variance levels are annualized variances.
    }
    \label{tab:numerics.references}
\end{table}

\paragraph{Marginal calibration.}
With the reference parameters fixed, we apply \Cref{alg:MSA} successively between calibration maturities, alternating the terminal-marginal and martingale updates.
At each maturity, the joint law $\kappa_i=\Law(X_{T_i},L_{T_i})$ is passed to the next calibration interval.
Values at intermediate observation dates are obtained from the conditional-expectation representation~\eqref{eq:R.Bass}.
Only the five stated maturities impose marginal constraints.
\Cref{fig:numerics.smiles} reports the repriced calls in implied-volatility terms.
All four models closely reproduce the target smiles over the displayed strike ranges.

\begin{figure}[H]
    \centering
    \includegraphics[width=0.41\textwidth]{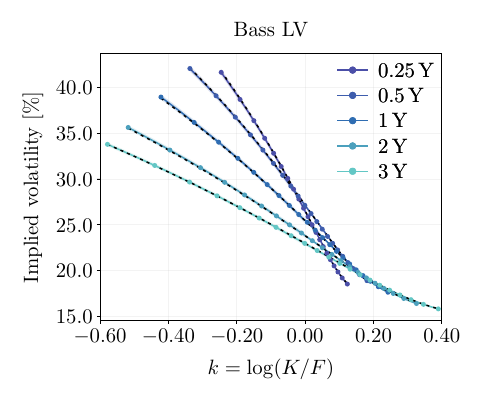}
    \hspace{1em}
    \includegraphics[width=0.41\textwidth]{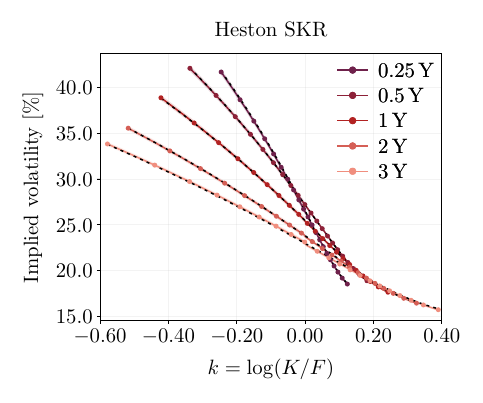}

    \includegraphics[width=0.41\textwidth]{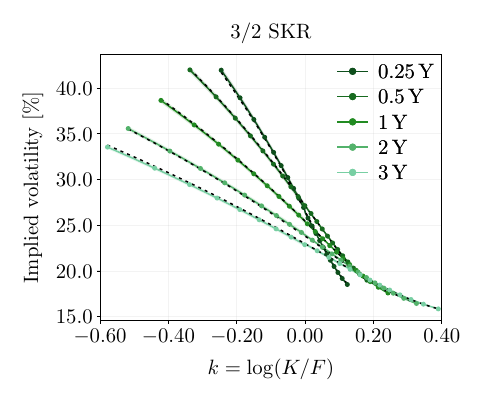}
    \hspace{1em}
    \includegraphics[width=0.41\textwidth]{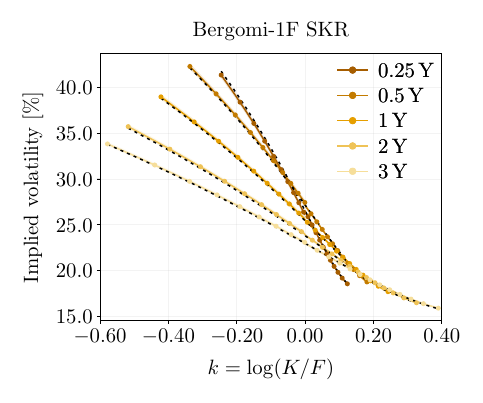}

    \caption{
        Fit of Bass LV and the three SKR models to the SPX implied-volatility smiles.
        Black dashed curves show the ESSVI targets; coloured curves with markers show implied volatilities recovered from Monte Carlo option prices under the calibrated models.
        Within each panel, darker to lighter shades indicate maturities $0.25$, $0.5$, $1$, $2$, and $3$ years.
        Shading denotes pointwise $95\%$ Monte Carlo confidence bands.
    }
    \label{fig:numerics.smiles}
\end{figure}

For both products below, we take $S_0=1$, a constant interest rate $r=2\%$, and no dividends, with discount factor $D(t)=e^{-rt}$.
All product prices use $10^5$ Monte Carlo paths per model.
\Cref{tab:numerics.timings} reports reference path simulation and fixed-point iteration separately, together with the pricing times for both products.

\begin{table}[H]
    \centering
    \small
    \sisetup{output-decimal-marker={.},group-digits=false}
    \setlength{\tabcolsep}{8pt}
    \renewcommand{\arraystretch}{1.1}

    \begin{tabular}{
            @{}l 
            S[table-format=1.4] 
            S[table-format=1.4] 
            S[table-format=1.4] 
            S[table-format=1.4]@{}
        }
        \toprule
        & \multicolumn{2}{c}{Calibration [s]} & \multicolumn{2}{c}{Pricing [s]} \\
        \cmidrule(lr){2-3}
        \cmidrule(l){4-5}
        Model & {Reference simulation} & {Fixed-point iteration} & {Cliquet} & {Autocall} \\
        \midrule
        Bass LV        & 0.0011 & 0.0801 & 0.0220 & 0.0219 \\
        Heston SKR     & 0.0334 & 0.1626 & 0.0757 & 0.0756 \\
        $3/2$ SKR      & 0.1421 & 0.1952 & 0.0753 & 0.0753 \\
        Bergomi-1F SKR & 0.0313 & 0.9204 & 1.0074 & 1.0073 \\
        \bottomrule
    \end{tabular}

    \caption[Calibration and pricing times]{
        Calibration and pricing times using $10^5$ paths.\protect\footnotemark
        Pricing reuses the reference paths generated during calibration.
        Each product's pricing time includes evaluating the model price process at the observation dates and averaging the discounted payoffs.
    }
    \label{tab:numerics.timings}
\end{table}
\footnotetext{
    Computations were performed on an Intel Core i5-12500 processor with $32$~GiB RAM under Debian~13.
}

\subsection{Reverse cliquets}
\label{sec:numerics.cliquets}

We consider a three-year reverse cliquet with notional $N=100$ and quarterly observation dates $t_j=j/4$, $j=0,\ldots,12$.
Writing $R_j=S_{t_j}/S_{t_{j-1}}-1$, its terminal payoff is
\begin{equation}
    \Pi_T^{\mathrm{cliq}}
    =
    N\Bigl[
        1+
        \Bigl(
            C_0-\sum_{j=1}^{12}\min\{(-R_j)^+,L\}
        \Bigr)^+
    \Bigr],
    \qquad T=t_{12},
    \label{eq:numerics.cliquet}
\end{equation}
where $C_0$ is the aggregate coupon budget and $L>0$ is the local loss cap, corresponding to a local return floor of $-L$.
The price is $D(T)\E[\Pi_T^{\mathrm{cliq}}]$.
\Cref{tab:numerics.cliquet} reports prices and Monte Carlo standard errors for $C_0=50\%$ and $L=7.5\%$.

\begin{table}[H]
    \centering
    \small
    \sisetup{output-decimal-marker={.},group-digits=false}
    \setlength{\tabcolsep}{8pt}
    \renewcommand{\arraystretch}{1.1}

    \begin{tabular}{
        @{}l
        S[table-format=3.3, table-number-alignment=right]
        S[table-format=1.3, table-number-alignment=right]@{}
    }
        \toprule
        Model & {Price} & {MC SE} \\
        \midrule
        Bass LV
            & 113.987 & 0.038 \\
        Heston SKR
            & 122.000 & 0.039 \\
        $3/2$ SKR
            & 116.991 & 0.035 \\
        Bergomi-1F SKR
            & 120.144 & 0.036 \\
        \bottomrule
    \end{tabular}

    \caption{
        Three-year quarterly reverse cliquet with $N=100$, $C_0=50\%$, and $L=7.5\%$.
        MC SE denotes the Monte Carlo standard error.
    }
    \label{tab:numerics.cliquet}
\end{table}

\Cref{fig:numerics.cliquet} varies $L$ at fixed $C_0=50\%$.
Each quarterly deduction is a put spread on $R_j$ with strikes $0$ and $-L$, so the comparison tests the downside part of the forward smile.
Raising $L$ increases deductions in quarters whose loss exceeds the cap and hence lowers the price.
Once $12L>C_0$, the aggregate coupon floor also makes the price sensitive to how losses accumulate across quarters.
At $L=7.5\%$, the SKR models generate fewer negative quarters and lower expected clipped losses than Bass LV.
The reduction in expected clipped losses explains their higher prices; the smaller aggregate-floor contribution partly offsets this effect.

\begin{figure}[H]
    \centering
    \includegraphics[width=0.41\textwidth]{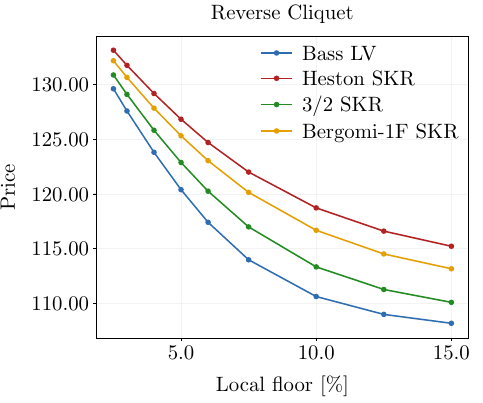}

    \caption{
        Local-floor sensitivity of the three-year quarterly reverse cliquet under Bass LV and the three SKR models, with $N=100$ and $C_0=50\%$.
        The horizontal axis shows the loss cap $L$, corresponding to the return floor $-L$.
    }
    \label{fig:numerics.cliquet}
\end{figure}

\subsection{Memory autocallables}
\label{sec:numerics.autocalls}

We consider memory autocallables with the same notional, maturity, and quarterly observations, and a quarterly coupon rate $c=5\%$.
While the note remains outstanding, the coupon $Nc$ and all previously unpaid coupons are paid at an observation if $S_{t_j}\geq bS_0$, where $b$ is the coupon barrier.
The note redeems at par at the first quarterly observation before maturity for which $S_{t_j}\geq S_0$, including any coupon due on that date.
The protection barrier $b_{\mathrm P}=60\%$ is tested only at maturity; any coupons still unpaid then expire.

Let $\tau$ be the redemption date, equal to $T$ if the note is not called, and let $C_j$, $j=1,\ldots,12$, be the coupon payment at $t_j$, including arrears.
Writing $s_T=S_T/S_0$, the discounted payoff is
\begin{equation}
    \Pi^{\mathrm{ac}}
    =
    \sum_{t_j\leq\tau}D(t_j)C_j
    +ND(\tau)
    -ND(T)(1-s_T)
        \mathbf{1}_{\{\tau=T,\ s_T<b_{\mathrm P}\}}.
    \label{eq:numerics.autocall}
\end{equation}
The final term is the capital loss for notes surviving to maturity below the protection barrier.
The price is $\E[\Pi^{\mathrm{ac}}]$.
\Cref{tab:numerics.autocall} reports results for $b=60\%$.

\begin{table}[H]
    \centering
    \small
    \sisetup{output-decimal-marker={.},group-digits=false}
    \setlength{\tabcolsep}{8pt}
    \renewcommand{\arraystretch}{1.1}

    \begin{tabular}{
        @{}l
        S[table-format=3.3, table-number-alignment=right]
        S[table-format=1.3, table-number-alignment=right]@{}
    }
        \toprule
        Model & {Price} & {MC SE} \\
        \midrule
        Bass LV
            & 104.899 & 0.069 \\
        Heston SKR
            & 106.834 & 0.070 \\
        $3/2$ SKR
            & 106.444 & 0.066 \\
        Bergomi-1F SKR
            & 106.800 & 0.067 \\
        \bottomrule
    \end{tabular}

    \caption{
        Three-year quarterly memory autocallable with $N=100$, coupon $5\%$, coupon barrier $60\%$, autocall barrier $100\%$, and terminal protection barrier $60\%$.
        Autocalling is possible from the first quarterly observation.
        MC SE denotes the Monte Carlo standard error.
    }
    \label{tab:numerics.autocall}
\end{table}

\Cref{fig:numerics.autocall} varies $b$ while keeping the autocall and protection barriers fixed.
This isolates coupon-payment and recovery risk: raising $b$ delays coupon payments or leaves more coupons unpaid at maturity.
The sensitivity depends on reaching $bS_0$ while the note remains outstanding, particularly after missed coupons; varying $b$ changes neither the redemption date nor the principal payoff on a given path.

Across models, price levels additionally depend on call times and terminal downside on uncalled paths.
At $b=60\%$, the higher SKR prices reflect larger expected discounted coupon receipts and smaller expected discounted capital losses, which outweigh the effect of later principal repayment.

\begin{figure}[H]
    \centering
    \includegraphics[width=0.41\textwidth]{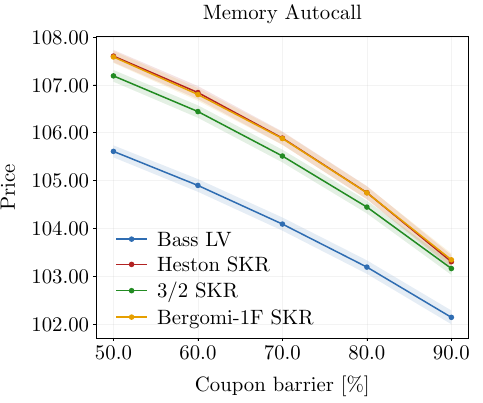}

    \caption{
        Coupon-barrier sensitivity of the three-year quarterly memory autocallable under Bass LV and the three SKR models, with $N=100$ and quarterly coupon $c=5\%$.
        The autocall and terminal protection barriers are $100\%$ and $60\%$ of $S_0$, respectively.
        Shading denotes pointwise $95\%$ Monte Carlo confidence bands.
    }
    \label{fig:numerics.autocall}
\end{figure}

\clearpage

\appendix
\begingroup

\section{Convergence to the stochastic local volatility model}
\label{app:skr:sec:convergence}

This appendix makes the single-asset comparison in
\Cref{sec:SLV_formal2} quantitative. Existence and uniqueness of each
bridge, together with its $R$-Bass representation, are supplied by
\Cref{thm:wmot.bridge.existence.uniqueness,thm:wmot.bridge.duality};
what remains is to control the effect of concatenating the bridges. Centered
endpoint errors give a direct martingale criterion for path
convergence. The quantitative argument controls those errors through
the calibrated inverse slopes for a smooth, uniformly nondegenerate
model with a regular initial law.

\subsection{The model and the quantitative conclusion}
\label{app:skr:sec:model}

Fix a finite horizon $T$ and a deterministic constant
$\rho\in(-1,1)$. On a probability space carrying independent Brownian
motions $W,W^\perp$, the reference model has dynamics
\begin{align}
 \rmd Y_t
   &=\alpha(t,L_t)\vartheta(t,Y_t)\,\rmd W_t,
   \label{app:skr:eq:reference}\\
 \rmd L_t
   &=b(t,L_t)\,\rmd t
     +\beta(t,L_t)
       \bigl(\rho\,\rmd W_t+\sqrt{1-\rho^2}\,\rmd W_t^\perp\bigr).
   \label{app:skr:eq:factor}
\end{align}
Here the factor is real-valued. The target model and its price marginals
are specified by
\begin{equation}
 \rmd X_t=f(t,X_t,L_t)\,\rmd W_t,
 \qquad
 \sigma_R(t,x,l)=\alpha(t,l)\vartheta(t,x),
 \qquad
 f(t,x,l)=\lambda(t,x)\sigma_R(t,x,l),
 \qquad
 \mu_t=\Law(X_t).
 \label{app:skr:eq:target}
\end{equation}
The factor evolves autonomously and is shared by the reference and target.
The initial pair $(X_0,L_0)$ is square-integrable and independent of the
Brownian increments. We write $\bar\kappa_t=\Law(X_t,L_t)$ for the target
joint law, and use the completed filtration generated by the initial
pair and the two Brownian motions.

The following coefficient assumptions are imposed for the SLV
convergence statements: $\alpha,b,\beta$ are smooth on $[0,T]\times\R$,
$\vartheta,\lambda$ are smooth on $[0,T]\times\R$, the functions and all
their derivatives are bounded, and
\[
 0<c\le\alpha,\beta,\vartheta,\lambda\le C<\infty.
\]
Constants below may depend on these bounds, on $T$, and on the fixed
correlation, but are independent of the number of resets. A fixed-order
estimate uses only finitely many derivative bounds. The independent
WMOT theorem in Subsection~\ref{app:skr:sec:wmot-theorem} uses its
own stated reference and base-density assumptions; it does not use
the target leverage or the SLV marginal curve.

For $h=T/N$ and $t_j=jh$, let $X^h$ be the SKR construction of the main
text for the marginals $(\mu_{t_j})_{j=0}^N$. It propagates the incoming
joint state law $\kappa_j=\Law(X^h_{t_j},L_{t_j})$ and starts each
reference interval at $Y_{t_j}=X^h_{t_j}$, retaining the factor $L$.
Its initial pair agrees with that of the target, and all
reference intervals use the same Brownian motions as the target.
The bounded Lipschitz coefficients give continuous transition kernels
in $W_2$ and finite second moments. Uniform nondegeneracy gives
atomless reference price transitions with support $\R$. Moreover,
successive target marginals are irreducible: conditional strict Jensen
for $|X_{t_{j+1}}-k|$, using the full support of the target transitions,
gives strict convex order at every $k\in\R$.
These properties verify \ref{aspt:lsc}, \ref{aspt:moments} and
\ref{aspt:atomless} at every reset, so
\Cref{thm:wmot.bridge.existence.uniqueness} supplies the unique bridge.
Irreducibility additionally gives its $R$-Bass representation by
\Cref{thm:wmot.bridge.duality}. This realizes the bridge on the common
drivers with the conditional-expectation interpolation
\eqref{eq:R.Bass}. No new existence argument is needed.

The general quantitative conclusion uses one additional
\emph{regular-start assumption}. Define
\[
 Z_w=\int_\R e^{-\sqrt{1+x^2}}\,\rmd x,
 \qquad
 w(x)=Z_w^{-1}e^{-\sqrt{1+x^2}},
 \qquad
 \omega(l)=Z_w^{-1}e^{-\sqrt{1+l^2}},
 \qquad
 \mathfrak w(x,l)=w(x)\omega(l).
\]
The initial joint law is assumed to have a density
$p_0=\mathfrak w\mathsf r_0$, where $\mathsf r_0$ is smooth, all its
derivatives are bounded, and $c\le\mathsf r_0\le C$; the density has
total mass one. This is an additional assumption for the general
quantitative theorem. The endpoint criterion and the autonomous-price result
below also apply to the deterministic initial price of the main text.
Constants in estimates using this assumption may also depend on these
initial-density bounds and on finitely many derivatives of $\mathsf r_0$.

Under this additional assumption, write $p(t,x,l)$ for the target joint
density and
\[
 m_t(x)=\int_\R p(t,x,l)\,\rmd l,
 \qquad
 n_t(l)=\int_\R p(t,x,l)\,\rmd x.
\]
In particular, $\mu_t(\rmd x)=m_t(x)\,\rmd x$.
The weighted Feynman--Kac formula in
Subsection~\ref{app:skr:sec:density-flow} establishes
\begin{equation}
 p(t,x,l)=\mathfrak w(x,l)\mathsf r(t,x,l),
 \qquad
 0<c_T\le\mathsf r\le C_T,
 \label{app:skr:eq:target-density}
\end{equation}
with bounded space-time derivatives of every order.

\begin{theorem}[Strong convergence of SKR]
\label{app:skr:thm:main}
Under the regular-start assumption, for all sufficiently small $h$,
\begin{equation}
 \E\sup_{0\le t\le T}|X_t^h-X_t|^2\le Ch.
 \label{app:skr:eq:rate}
\end{equation}
Consequently, $(X^h,L)$ converges to $(X,L)$ in probability in the
uniform path topology, and hence in distribution.
\end{theorem}

The rate is strong order one half. In particular, the theorem concerns
bounded uniformly positive coefficients and a regular initial joint law;
it does not supply the general stochastic-factor rate from a singular
initial law. The role of that extra assumption is confined to the
quantitative calibration argument. We first give the probabilistic
reduction for an arbitrary square-integrable initial pair.

\paragraph{Structure of the proof and the techniques used.}
The argument has three parts. Centered endpoint comparisons first
reduce path convergence to a bound on the calibrated inverse slopes.
A standalone WMOT theorem then controls the joint-density error in
the explicit weighted norm
\[
 \mathcal N_b(z)^2=\sup_n w(n)^{-1}
       \int\frac{\chi_n(x)^2z(x,l)^2}
                       {\sigma_R(x,l)^2p(x,l)}\,dx\,dl.
\]
Here p is the known comparison density and $\sigma_R^2$ is the reference
price variance at the beginning of the step. Finally, an explicit
smooth comparison step applies that theorem along the SLV flow.
The central cancellation is simple: inverse-variance weighting makes
the highest-order scalar calibration term independent of the factor;
the common terminal price marginal then makes that term vanish.
A full two-dimensional heat interpolation implements this cancellation
without replacing the transport of a rough factor density by a frozen
factor approximation.

\smallskip\noindent
\emph{Probability and WMOT: from endpoints to paths.}
\begin{itemize}[leftmargin=*,itemsep=3pt,topsep=4pt]
\item \textbf{Subsection~\ref{app:skr:sec:model}: model and scope.}
The assumptions specify the regular initial law and the separated
reference and shared autonomous factor used in the SLV application.
The conclusion is a strong, uniform-path estimate on every fixed
finite horizon. The later standalone WMOT theorem permits a more
general reference diffusion; this does not enlarge the scope of the
stated SLV convergence theorem.
\item \textbf{Subsection~\ref{app:skr:sec:endpoint-convergence}: a centered endpoint criterion.}
A known smooth comparison slope produces a centered endpoint with
an $O(h^2)$ squared Euler error. Martingale orthogonality and Doob's
inequality accumulate the endpoint errors; Lipschitz-SDE estimates
then give the maximal path bound. Comparing two centered inverse
maps shows that a uniform slope error $\varepsilon$ costs only
$Ch\varepsilon^2$ per step. These are probabilistic estimates and
an elementary monotone-map calculation.
\item \textbf{The autonomous-price corollary in the same subsection.}
A target continuation is then an admissible WMOT competitor.
Optimality and a known uniformly convex potential bound its endpoint
error by $Ch^2$. This short WMOT argument also covers deterministic
initial price, without the regular-density assumption.
\end{itemize}

\smallskip\noindent
\emph{The standalone WMOT theorem: a density energy and its quantitative proof.}
\begin{itemize}[leftmargin=*,itemsep=3pt,topsep=4pt]
\item \textbf{Subsection~\ref{app:skr:sec:wmot-theorem}: the law-level statement.}
A reference experiment on $[0,1]$ represents physical duration
$\delta$: it is a known smooth nondegenerate diffusion with
martingale price coordinate. A known smooth joint law and positive
inverse slope generate its terminal price law. Nearby incoming laws
have the same price marginal. The theorem controls their WMOT
endpoints in the explicit variance-weighted density distance above,
with amplification at most $1+C\delta$ on a neighborhood of radius
$\delta^{3/4}$, which contains the $O(\delta)$ errors needed later. The metric uses only the known density and reference
variance; it requires no regularity assumption on an unknown optimizer.
\item \textbf{Subsection~\ref{app:skr:sec:finite}: endpoint calculus.}
Rescaling the exact endpoint and calendar formulas gives $C^2$
bounds in elementary $L^2$ and $L^\infty$ spaces. Scalar Gaussian
averaging supplies a factor $h^{-1/4}$; a neighborhood of radius
$h^{1/4}$ pays for it. The fixed price marginal cancels the sole
unsmoothed term in the calendar equation. This is analytic work,
using Gaussian bounds and product estimates without density derivatives.
\item \textbf{Subsection~\ref{app:skr:sec:channel}: the exact density derivative.}
Tilting the endpoint experiment writes the incoming-density derivative
as a conditional expectation plus a small mass-preserving error.
The mean constraint expresses the slope derivative as a scalar
calibration term. The associated known-kernel price commutator is
exactly symmetric. These identities follow from the conditional
endpoint formulas in physical coordinates, retaining the frozen
initial shift label.
\item \textbf{Subsection~\ref{app:skr:sec:native-kernel}: full heat jumps and centered kernels.}
An auxiliary reversible diffusion retains the full known price--factor
covariance. Its jump form is comparable to a two-dimensional Gaussian
jump form, and a weighted Poincar\'e inequality gives a gap of order h.
A Gaussian-midpoint estimate converts exact column cancellation and
localized kernel bounds into an inverse-noise estimate. The diffusion
and resampling interpretations are probabilistic; the Gaussian bounds,
weighted inequalities and kernel estimates are quantitative analysis.
\item \textbf{Subsection~\ref{app:skr:sec:pressure}: the scalar inverse.}
A crossing-flux identity turns price calibration into an equation for
the scalar slope. An occupation series for a Gaussian-mixture walk
with positive killing probability constructs its inverse. The
representation is probabilistic; the Gaussian parametrix and weighted
Schur bounds that make it uniform are analytic.
\item \textbf{Subsection~\ref{app:skr:sec:localscalar}: exact calibration.}
The scalar inverse and the Banach implicit-function theorem give the
local calibrated branch on the $h^{1/4}$ neighborhood. This analytic
construction permits signed densities and uses only the known smooth
base and the explicit endpoint formulas.
\item \textbf{Subsection~\ref{app:skr:sec:strongchannel}: the weighted heat cancellation.}
A forward--backward heat interpolation connects a density perturbation
to a price-calibrated endpoint. Differentiating its weighted squared
norm twice gives a nonnegative principal part and a known first-order
commutator. At the terminal time, the price-marginal constraint cancels
the leading scalar calibration term. Integration by parts then yields
an energy bound with output jump dissipation. Full Gaussian kernel
comparison, cutoff identities and a global scalar correction transfer
this bound to the exact density derivative at cost $O(h)$.
This is an analytic energy argument for a known probabilistic
experiment; it neither differentiates an input density nor freezes
its factor coordinate.
\item \textbf{Subsection~\ref{app:skr:sec:quadratic}: the nonlinear step.}
Changing the slope leaves the reference terminal factor unchanged.
This gives exact column cancellation in the first and second density
derivatives. The centered-kernel estimate and integration of the
$C^2$ second derivative along the signed branch give a quadratic
remainder. Pairing it with the retained output dissipation and using
Young's inequality produces a fourth-order energy error, hence a
cubic norm error. A uniformly convex WMOT identity finally identifies
positive calibrated endpoints as the unique optimizers. The only
differentiability used here comes from the explicit scalar and
endpoint formulas.
\end{itemize}

\smallskip\noindent
\emph{Application to the specified smooth SLV target.}
\begin{itemize}[leftmargin=*,itemsep=3pt,topsep=4pt]
\item \textbf{Subsection~\ref{app:skr:sec:wmot-consistency}: smooth consistency.}
A strong first-order endpoint expansion identifies the target
generator. A first directional calendar expansion identifies the
leading response to a smooth slope change. Gaussian derivative
transfer establishes these estimates at known smooth inputs, in the
finite orders needed by the scalar preconditioner and density norm.
No regularity of an unknown calibrated endpoint is used.
\item \textbf{Subsection~\ref{app:skr:sec:density-flow}: bounds for the target.}
Conjugating the target forward equation by the fixed density weight
gives a Feynman--Kac formula with bounded coefficients and potential.
Positivity gives the upper and lower density bounds; ordinary forward
SDE variations give the required derivatives. This is a probabilistic
representation after an elementary differential-operator calculation.
\item \textbf{Subsection~\ref{app:skr:sec:iteration}: comparison and iteration.}
An explicit pointwise slope correction gives a known smooth
comparison step with joint and calibration defects $O(h^2)$.
The existing local branch corrects its target at a cost $O(h^2)$,
even at the actual rough incoming density. The triangle inequality combines these defects with stability.
Adjacent norms compare by the pointwise ratio of the known
weights $p_t\sigma_R(t)^2$, independently of the comparison slope.
The standalone theorem therefore yields
\[
 e_{j+1}\le(1+Ch)e_j+Ch^2.
\]
An $O(h)$ bootstrap closes for every fixed finite horizon, well inside
the calibration neighborhood. The slope bound and the initial
endpoint criterion then give the strong path rate.
\end{itemize}

Probability and WMOT supply the endpoint comparison, the conditional
identities, the scalar occupation representation and the path closure.
The quantitative proof also uses Gaussian kernel estimates, a weighted
Poincar\'e inequality, a heat energy calculation, localization and the
Banach implicit-function theorem. The regular-start and coefficient
assumptions above remain part of the convergence statement.

\subsection{From centered endpoint errors to path convergence}
\label{app:skr:sec:endpoint-convergence}

Write $x_j=X^h_{t_j}$, $\mathbb E_j=\mathbb E[\cdot\mid\mathcal F_{t_j}]$,
$\Delta W_j=W_{t_{j+1}}-W_{t_j}$ and
$f_j=f(t_j,x_j,L_{t_j})$. On the jth reference interval let
$M_j=Y_{t_{j+1}}-x_j$ and $U_j=X^h_{t_{j+1}}$. Introduce the
known comparison inverse slope and its centered endpoint
\[
 k_j^0(x)=1/\lambda(t_j,x),\qquad
 V_j=G_j\bigl(\xi_j(x_j)+b_j+M_j\bigr),
 \qquad \mathbb E_jV_j=x_j,
\]
where $\xi_j'=k_j^0$ and $G_j=\xi_j^{-1}$. This endpoint need
not have the prescribed terminal marginal.

Only the known comparison map is differentiated in the following
elementary estimate:
\begin{equation}
 \mathbb E_j|V_j-x_j-f_j\Delta W_j|^2\le Ch^2,
 \qquad
 \mathbb E_j[V_j-x_j-f_j\Delta W_j]=0.
 \label{app:skr:eq:trial-endpoint-estimate}
\end{equation}
Indeed reference It\^o isometry gives
$\mathbb E_j|M_j-\sigma_R(t_j,x_j,L_{t_j})\Delta W_j|^2\le Ch^2$
and $\mathbb E_j|M_j|^4\le Ch^2$. Since $\mathbb E_jM_j=0$,
Taylor expansion of the smooth $G_j$ first gives $|b_j|\le Ch$
and then the displayed squared error. Its conditional mean is exactly
zero by construction.

\begin{proposition}[An endpoint criterion]
\label{app:skr:prop:endpoint-criterion}
For any square-integrable initial pair and its genuine Bass
interpolation,
\begin{equation}
 \mathbb E\sup_{t\le T}|X_t^h-X_t|^2
 \le C_T\left(h+\sum_{j<N}\mathbb E|U_j-V_j|^2\right).
 \label{app:skr:eq:endpoint-convergence}
\end{equation}
\end{proposition}
\begin{proof}
The residuals $\epsilon_j=U_j-x_j-f_j\Delta W_j$ have
$\mathbb E_j\epsilon_j=0$ and
$\mathbb E\epsilon_j^2\le C\mathbb E|U_j-V_j|^2+Ch^2$.
The within-bridge process is $X_t^h=\mathbb E[U_j\mid\mathcal F_t]$.
Thus its difference from the corresponding frozen-coefficient
stochastic integral is the martingale
\[
 R_t=\sum_{i<j}\epsilon_i+
             \mathbb E[\epsilon_j\mid\mathcal F_t],
 \qquad t_j\le t\le t_{j+1}.
\]
Orthogonality and Doob's inequality bound
$\mathbb E\sup|R|^2$ by $4\sum_j\mathbb E\epsilon_j^2$.
Subtract the target equation from
\[
 X_t^h=X_0+\int_0^t
 f(t_{\lfloor s/h\rfloor},X^h_{t_{\lfloor s/h\rfloor}},
                  L_{t_{\lfloor s/h\rfloor}})\,dW_s+R_t.
\]
The target coefficient is globally Lipschitz, and the target price
and common autonomous factor have squared increments $Ch$.
Burkholder--Davis--Gundy and Gronwall prove
\eqref{app:skr:eq:endpoint-convergence}. This uses the frozen-label
conditional expectation, not a Markov assumption on current physical
price and factor. Only unconditional bounds on the squared residuals
are needed.
\end{proof}

\begin{lemma}[Stability of a centered Bass endpoint]
\label{app:skr:lem:centered-slope-stability}
Condition on an incoming state with price x. Let M have mean zero
and second moment at most $Ch$. For $c\le k_i\le C$, $i=0,1$,
put
\[
 R_i(u)=\int_0^u k_i(x+s)\,ds,\qquad J_i=R_i^{-1},\qquad
 U_i=x+J_i(b_i+M),\quad \mathbb E U_i=x.
\]
Then
\begin{equation}
 \mathbb E|U_1-U_0|^2
       \le Ch\|k_1-k_0\|_\infty^2.
 \label{app:skr:eq:centered-slope-stability}
\end{equation}
\end{lemma}
\begin{proof}
For $\varepsilon=\|k_1-k_0\|_\infty$, the uniform inverse
bounds give $|J_1(s)-J_0(s)|\le C\varepsilon|s|$.
Each mean map $b\mapsto\mathbb E J_i(b+M)$ has all difference
quotients bounded above and below by positive constants. Consequently
$|b_i|\le C\mathbb E|M|\le C\sqrt h$. At $b_0$, the two
means differ by at most $C\varepsilon\sqrt h$, so
$|b_1-b_0|\le C\varepsilon\sqrt h$. Hence
$|U_1-U_0|\le C\varepsilon(|M|+\sqrt h)$, which proves
\eqref{app:skr:eq:centered-slope-stability}. No spatial derivative
of either slope is used.
\end{proof}

\begin{corollary}[Autonomous price volatility]
\label{app:skr:cor:autonomous-price}
If $\alpha$ is independent of the factor, the strong bound
\eqref{app:skr:eq:rate} holds for any square-integrable initial pair,
including a deterministic initial price.
\end{corollary}
\begin{proof}
Start the target price equation from $x_j$ with the same Brownian
increments, and let $Z_j$ be its terminal value. Its price dynamics
are autonomous, so its terminal law is $\mu_{t_{j+1}}$ and
$\mathbb E_jZ_j=x_j$. It is therefore an admissible martingale
endpoint for the static WMOT comparison.

Choose $\psi_j'=\xi_j$ and write
$B_j(z,v)=\psi_j(z)-\psi_j(v)-\psi_j'(v)(z-v)$.
The identity $\xi_j(V_j)=Y_{t_{j+1}}+\zeta_j(x_j,L_{t_j})$
and the conditional means show that, among admissible endpoints Z,
$\mathbb E B_j(Z,V_j)$ is a fixed constant minus
$\mathbb E[ZY_{t_{j+1}}]$. WMOT optimality therefore gives
\[
 \mathbb E|U_j-V_j|^2
 \le C\mathbb E B_j(U_j,V_j)
 \le C\mathbb E B_j(Z_j,V_j)
 \le C\mathbb E|Z_j-V_j|^2\le Ch^2.
\]
Here the curvature bounds for the known $\psi_j$ give the two
quadratic comparisons; all terms are integrable by their quadratic
growth. The final estimate follows from the centered Euler
approximations for the smooth target and for $V_j$.
Summation and Proposition~\ref{app:skr:prop:endpoint-criterion}
prove the assertion. No density assumption is used in this corollary.
\end{proof}

The factor-dependent price case will use the centered-slope estimate after
quantitative calibration. We now prove a standalone theorem for the
short WMOT update that supplies that calibration bound.

\subsection{A one-period WMOT calibration and stability theorem}
\label{app:skr:sec:wmot-theorem}

We isolate the one-step result from the convergence problem. Its data
are a reference experiment, a smooth incoming density and a smooth
positive inverse Bass slope. The terminal price law is initially the
one generated by this known base. The smooth target application uses
an explicit comparison slope and a small prescribed-target correction
in Subsection~\ref{app:skr:sec:iteration}.
Neither a target SLV model nor a subdivision enters the statement.

\paragraph{The reference experiment on $[0,1]$.}
Fix the probability weights $w,\omega$ defined above and let
$0<\delta\le\delta_0$. Conditional on $(Y_0,L_0)=(x,l)$, the reference
experiment is a two-dimensional diffusion on $[0,1]$ with generator
\begin{equation}
 \delta\mathcal L_{\delta u},\qquad
 \mathcal L_s=\tfrac12 A^R_s:D^2+b^R_s\partial_l,
 \qquad 0\le u\le1.
 \label{app:skr:eq:unit-reference}
\end{equation}
The covariance matrix $A^R_s(x,l)$ and drift $b^R_s(x,l)$ are smooth
with bounded derivatives in physical time and space, $b^R$ is bounded,
and $cI\le A^R\le CI$. The first coordinate has zero drift and is
therefore a martingale. Write $\sigma_R^2=(A^R)_{11}$.
Thus $\delta$ is the physical duration represented by the unit
interval. Time-translated data are permitted; arbitrary fast time
dependence on u is not included. Neither a separated price volatility
nor an autonomous second coordinate is required here.

The known base is a probability density $p=\mathfrak w r$ and a slope
$k$, where $c\le r,k\le C$ and all derivatives of $r,k$ are bounded.
A \emph{uniform class} means that these bounds and the coefficient
bounds are fixed. The data may depend on $\delta$ without any
regularity in $\delta$, provided the bounds remain uniform. Only finitely many of these smooth bounds enter the theorem.

Set
\[
 \xi_k(x)=\int_0^x k(v)\,dv,\qquad G_k=\xi_k^{-1}.
\]
There is a unique centering shift $b_{\delta,k}(x,l)$ such that
\begin{equation}
 U^b=G_k\bigl(\xi_k(x)+b_{\delta,k}(x,l)+Y_1-x\bigr),
 \qquad \E[U^b\mid x,l]=x.
 \label{app:skr:eq:unit-base}
\end{equation}
In the notation of \Cref{sec:bridge.wmot}, the generator is
$g=G_k=\partial_x\psi_k^*$, where $\psi_k'=\xi_k$ and hence
$\psi_k''=k$. The bounds on $k$ make $G_k$ globally finite and
Lipschitz. The corresponding shift is
$a_k(x,l)=\xi_k(x)+b_{\delta,k}(x,l)-x$, so that
$U^b=G_k(a_k(x,l)+Y_1)$; this is the shift denoted by $\zeta_k$
below.
With incoming density $p$, denote the density of $(U^b,L_1)$ by
$q^b=\mathcal O_\delta(p,k)$ and its price marginal by $\nu^b$.
The full joint reference kernel
$(x,l)\mapsto\Law(Y_1,L_1\mid x,l)$ is part of the data: it determines
the joint price--factor output of the WMOT optimizer.

For a nonnegative incoming density $\widetilde p$ and terminal law
$\nu\in\mathcal P_2(\R)$, the WMOT problem can equivalently be written as
\[
 \sup\bigl\{\E[UY_1]:\ \Law(U)=\nu,\quad
          \E[U\mid Y_0,L_0]=Y_0\bigr\},
\]
where the reference experiment has initial density $\widetilde p$ and
the supremum is over couplings of $U$ with that experiment. The martingale constraint and the prescribed reference law
are fixed in this optimization.

\paragraph{Joint laws and the distance used below.}
Write $\gamma_b=p(x,l)\,dx\,dl$, $\mu=p^X$, and
\[
 \mathcal P_\mu=\{\gamma\in\mathcal P_2(\mathbb R^2):\gamma^X=\mu\}.
\]
Here $\gamma$ denotes an incoming joint state law, corresponding to
$\kappa$ in \Cref{sec:bridge.wmot}.
The WMOT update sends the incoming joint law to the joint law of the
optimal price endpoint and the terminal reference factor:
\[
 \Phi_\delta:\mathcal P_\mu\longrightarrow\mathcal P_{\nu^b}.
\]
Only its restriction to the neighborhood below is asserted.
The compared incoming laws have the same price marginal; the factor
law and the dependence between price and factor may change.
The output factor law is the terminal second-coordinate law of the
reference experiment with that full incoming joint law. It is
independent of the choice of the endpoint slope.
The base endpoint $U^b$ is itself optimal for its generated target,
as the convex-potential identity below verifies.

Fix $\chi\in C_c^\infty(\mathbb R)$ with $0\le\chi\le1$,
$\chi=1$ on $[-1,1]$ and support in $[-2,2]$, and put
$\chi_n(x)=\chi(x-n)$ for $n\in\mathbb Z$. For joint laws absolutely continuous with
respect to $\gamma_b$, write $f_\gamma=d\gamma/d\gamma_b$ and define
\begin{equation}
 d_b(\gamma,\bar\gamma)^2
 =\sup_n\frac1{w(n)}\,\mathbb E_{\gamma_b}
 \left[\frac{\chi_n^2}{\sigma_R(0)^2}
             (f_\gamma-f_{\bar\gamma})^2\right].
 \label{app:skr:eq:unit-distance}
\end{equation}
All functions in the expectation are evaluated at the incoming pair.
This distance measures the squared relative-density error, locally in
price, weighted by the reciprocal reference price variance. It is
defined for arbitrary joint price marginals and is used at both
endpoints of the update. Its formula uses only the known input law
and reference variance; the optimizing slope does not enter.

\begin{theorem}[Stability of a short WMOT update]
\label{app:skr:thm:wmot-calibration}
Fix a uniform class of known reference experiments and smooth bases
as above. There are constants $C,\delta_*>0$ such that, for
$0<\delta\le\delta_*$ and every incoming law
$\gamma\in\mathcal P_\mu$, $\gamma\ll\gamma_b$, with
$d_b(\gamma,\gamma_b)\le\delta^{3/4}$, the WMOT update exists
uniquely and
\begin{equation}
 d_b(\Phi_\delta\gamma,\Phi_\delta\gamma_b)
 \le(1+C\delta)d_b(\gamma,\gamma_b).
 \label{app:skr:eq:unit-stability}
\end{equation}
\end{theorem}

Thus a short regular WMOT update has amplification at most
$1+C\delta$ throughout a specified shrinking neighborhood. This
neighborhood contains all errors of order $\delta$. The estimate
compares each perturbed input with the known base; it does not assert
pairwise contraction on a fixed neighborhood.

\paragraph{The density norm behind the distance.}
For the proof, write the distance as
$d_b(\gamma,\bar\gamma)=\mathcal N_b(p_\gamma-p_{\bar\gamma})$, where
\begin{equation}
 \mathcal N_b(z)^2
 =\sup_n w(n)^{-1}\int_{\mathbb R^2}
       \frac{\chi_n(x)^2z(x,l)^2}
            {\sigma_R(0,x,l)^2p(x,l)}\,dx\,dl.
 \label{app:skr:eq:explicit-density-norm}
\end{equation}
For signed density differences, put
\[
 \|z\|_E=\sup_n\left\|\chi_n(x)
          \frac{z(x,l)}{w(x)\omega(l)}\right\|_{L^2(dx\,\omega(l)dl)},
 \qquad
 \mathcal Z=\{z\in E:z^X=0\},
\]
Then $\mathcal N_b$ is a norm on E, uniformly equivalent to
$\|\cdot\|_E$. No density derivative
is required. The fixed-price constraint is used for calibration;
the norm itself is defined on all E.

For a family of known bases $b_a$, parameterized by an interval,
a uniform bound on $|\partial_a\log(p_a\sigma_{R,a}(0)^2)|$ gives
\begin{equation}
 e^{-C|a-a'|}\mathcal N_{b_a}(z)
 \le\mathcal N_{b_{a'}}(z)
 \le e^{C|a-a'|}\mathcal N_{b_a}(z),
 \qquad z\in E.
 \label{app:skr:eq:unit-metrics}
\end{equation}
The constant may additionally depend on these family bounds.

\paragraph{What the proof constructs.}
The proof first constructs the calibrated density map $\mathcal S_\delta$
on the larger signed
neighborhood $\|z\|_E\le r\delta^{1/4}$. Its scalar inverse slope is
measured in
\[
 \|v\|_{K_\delta}=\|v\|_{L^2_{\rm uloc}}
                       +\delta^{1/4}\|v\|_\infty.
\]
An auxiliary reversible diffusion keeps the full joint covariance.
Its weighted heat interpolation has a lower bound on its second
energy derivative, and terminal price calibration gives dissipation.
The exact density channel differs by centered short-jump errors.
Conservation of the terminal factor law places the quadratic
remainder in the dual of the same jump energy. Pairing it with the linear
output and applying Young's inequality gives
\[
 \mathcal N_b(\mathcal S_\delta(p+z)-\mathcal S_\delta p)
 \le(1+C\delta)\mathcal N_b(z)
                  +C\delta^{-1/2}\mathcal N_b(z)^3.
\]
For $\mathcal N_b(z)\le\delta^{3/4}$, the last term is at most
$C\delta\mathcal N_b(z)$. The physical distance is the explicit one in
\eqref{app:skr:eq:explicit-density-norm}. Scalar branch differentiability
is proved as part of the argument.

Subsections~\ref{app:skr:sec:finite}--\ref{app:skr:sec:localscalar}
construct the scalar branch. The auxiliary heat geometry and
linear dissipation are proved in
Subsections~\ref{app:skr:sec:channel}--\ref{app:skr:sec:strongchannel}.
The centered-kernel $C^2$ argument in
Subsection~\ref{app:skr:sec:quadratic} completes the theorem.
The subsequent consistency and SLV application are separate results.

\subsection{Calculus for one endpoint experiment}
\label{app:skr:sec:finite}

We use physical time, writing $h=\delta$ on an interval $[t,t+h]$.
The case $t=0$ is the experiment on $[0,1]$ in
Theorem~\ref{app:skr:thm:wmot-calibration}; another $t$ translates the
coefficient data. Write $q_{t,h}^{x,l}(m,l')$ for the density of
$(Y_{t+h}-x,L_{t+h})$ conditional on $(Y_t,L_t)=(x,l)$ and
$r_{t,h}^{x,l}(m)=\int q_{t,h}^{x,l}(m,l')\,dl'$.
The initial-price dependence is retained throughout. We use the
reference generator $\mathcal L_t$ in
\eqref{app:skr:eq:unit-reference}. In the SLV application it specializes to
\begin{equation}
 A^R=\begin{pmatrix}
 \sigma_R^2&\rho\sigma_R\beta\\
 \rho\sigma_R\beta&\beta^2
 \end{pmatrix},\qquad b^R=b,\qquad \sigma_R=\alpha\vartheta.
 \label{app:skr:eq:generator}
\end{equation}
All one-period estimates below use only the general reference
assumptions of the standalone theorem.

For the positive inverse slope $k$, put
$R_k(x,u)=\int_0^u k(x+v)\,dv$. The centering shift is characterized by
\begin{equation}
 \int u\,r_{t,h}^{x,l}(R_k(x,u)-b_k(x,l))k(x+u)\,du=0.
 \label{app:skr:eq:meanshift}
\end{equation}
If $J_{k,x}$ is the inverse of $u\mapsto R_k(x,u)$ and $M$ has
density $r_{t,h}^{x,l}$, this says
$U=x+J_{k,x}(b_k+M)$ and $\E[U\mid x,l]=x$. Moreover,
\begin{equation}
 \partial_b\E[J_{k,x}(b+M)]
 =\int r_{t,h}^{x,l}(R_k(x,u)-b)\,du\in[C^{-1},c^{-1}].
 \label{app:skr:eq:shiftderivative}
\end{equation}
Thus the shift is unique. For a known smooth slope, the martingale
identity $\E[M\mid x,l]=0$ and Taylor's formula give
\begin{equation}
 \|b_k\|_\infty\le Ch.
 \label{app:skr:eq:shiftbound}
\end{equation}
The exact joint endpoint density is
\begin{equation}
 \mathcal O_{t,h}(p,k)(z,l')
 =k(z)\int p(x,l)
 q_{t,h}^{x,l}\left(\int_x^z k(v)\,dv-b_k(x,l),l'\right)\,dx\,dl.
 \label{app:skr:eq:exactkernel}
\end{equation}
It is linear in $p$, preserves mass and global price mean, and has
exactly the reference experiment's terminal factor marginal,
independently of $k$. Indeed, changing variables in the terminal price gives
\[
 \int\mathcal O_{t,h}(P,k)(u,l')\,du
 =\int P(x,l)q^L_{t,h}(x,l;l')\,dx\,dl,
 \qquad q^L_{t,h}(x,l;l')=\int q_{t,h}^{x,l}(m,l')\,dm.
\]
The factor follows the reference path, even if its coefficients
depend on the reference price. The endpoint is
$G_k(Y_{t+h}+\zeta_k(x,l))$, where
$\zeta_k=\xi_k(x)+b_k-x$. This initial label stays frozen within the
step; the physical interpolation need not be Markov in price and
factor alone.

Write $h=\delta$, $\varepsilon=\sqrt h$ and $\tau=h^{1/4}$.
For the relative density and the scalar slope use
\[
 \|P\|_{E}=\sup_n\left\|\chi_n(x)
       \frac{P(x,l)}{w(x)\omega(l)}\right\|_{L^2(dx\,\omega(l)dl)},
 \qquad
 \|v\|_{K_h}=\|v\|_{L^2_{\rm uloc}}+\tau\|v\|_\infty.
\]
Thus $K_h$ is $L^\infty$ equipped with this $h$-dependent norm.
No derivative of an input density or a slope is included. Slopes are
measurable representatives bounded above and below; their primitives
are absolutely continuous, strictly increasing bi-Lipschitz maps.
Fix a known smooth base $(p,k_0)$ with the usual positive weighted
bounds. Throughout, density perturbations belong to the single space
$\mathcal Z=\{z\in E:z^X=0\}$. No factor-marginal constraint is
required.

Define the normalized centered calendar by
\[
 J_h(P,k)(z)=h^{-1}\int P(x,l)\mathbb E\bigl[(U-z)^+-(x-z)^+
 -\mathbf1_{\{x>z\}}(U-x)\mid x,l\bigr]\,dx\,dl.
\]
The subtracted linear term has conditional mean zero. Set
$J_b=J_h(p,k_0)$ and
$m_b=\mathcal O_{t,h}(p,k_0)^X/k_0$; keep this smooth base weight
fixed when varying the density and slope.
Choose once for the uniform class
$a_*=(1+2\overline v^{\,2})^{-1}$, where $\overline v$ bounds
$\sigma_R/k$ on the permitted positive slope range, and put
\[
 \mathcal F_b(P,k)
 =(a_*I-h\partial_{xx})\frac{J_h(P,k)-J_b}{m_b},
\]
where $J_b$ is the known base calendar. The derivatives in this
preconditioner are distributional; the identity used below identifies
the result as a function in $K_h$.

\begin{lemma}[Scaled finite-step calculus]
\label{app:skr:lem:finite}
There is a fixed small $r>0$ such that on
\[
 P\in p+\mathcal Z,\qquad
 \|P-p\|_{E}+\|k-k_0\|_{K_h}<r\tau,
\]
the endpoint map into $E$ and the scalar map into $K_h$ are $C^2$,
and
\[
 \|D^j\mathcal O_h\|+\|D^j\mathcal F_b\|
       \le C\tau^{1-j},\qquad j=1,2.
\]
The constants are independent of small $h$. Density arguments may be
signed. The fixed incoming price marginal in the scalar assertion is
essential.
\end{lemma}
\begin{proof}
Choose $r$ small enough that
$\|k-k_0\|_\infty\le\tau^{-1}\|k-k_0\|_{K_h}<r$ preserves uniform
positive slope bounds. Only these bounds, not slope derivatives, will
enter the kernel estimates in this proof.

\paragraph{The conditional mean equation.}
Use the scaled reference kernels
\[
 Q_h(x,l;r,\eta)=h q_{t,h}^{x,l}(\varepsilon r,l+\varepsilon\eta),
 \qquad \bar Q_h(x,l;r)=\varepsilon r_{t,h}^{x,l}(\varepsilon r),
\]
which have uniform polynomial Gaussian majorants for the finitely
many fast derivatives required below. Define
\[
 A_k(x,r)=\int_0^r k(x+\varepsilon s)\,ds,\qquad
 \beta_k=b_k/\varepsilon.
\]
The exact mean equation is
\[
 \int r\,\bar Q_h(x,l;A_k(x,r)-\beta_k(x,l))
                    k(x+\varepsilon r)\,dr=0.
\]
Its derivative in $\beta$ is
\[
 \int\bar Q_h(x,l;A_k(x,r)-\beta_k(x,l))\,dr\in[C^{-1},c^{-1}].
\]
The pointwise mean root satisfies $|\beta_k|\le C$: the scaled
reference increment has bounded first moment and $A_k^{-1}$ has
uniform slope bounds. Implicit differentiation is therefore legitimate
in $L^\infty$ and yields the required parameter derivatives.

More precisely, for a fixed enlarged Gaussian averaging operator
$\mathcal G_h f(x)=\int G(s)|f(x+\varepsilon s)|\,ds$, the explicit
implicit formulas give, for $1\le j\le2$,
\begin{equation}
 \sup_l|D_k^j\beta_k[v_1,\ldots,v_j](x,l)|
 \le C\sum_{i=1}^j
       \left(\prod_{a\ne i}\|v_a\|_\infty\right)
             \mathcal G_h|v_i|(x).
 \label{app:skr:eq:mean-envelope}
\end{equation}
Indeed, every slope direction is a fixed translation or an integral
of fixed translations. Put all directions except one in $L^\infty$;
Gaussian Fubini bounds the remaining direction by the displayed
average. Products of lower mean-shift derivatives have the same form,
and division by the positive mean derivative preserves it. This is
a zeroth-order averaged-direction estimate. No spatial derivative of
$\beta_k$, $k$ or a density is taken.

\paragraph{Three elementary kernel estimates.}
The endpoint formula after $x=u-\varepsilon r$ and
$l=l'-\varepsilon\eta$ contains only fixed translations of $P$ and
of the slope directions. Its kernel is bounded by a fixed Gaussian
in $(r,\eta)$ because $|A_k(x,r)-\beta_k|\ge c|r|-C$.
Weight ratios are absorbed by Gaussian moments. The same description
holds after parameter differentiation, with polynomial displacement
factors and the averages in \eqref{app:skr:eq:mean-envelope}.

For a signed density $Z$, a kernel of this type has operator norm
$C:E\to E$ and price-supremum norm
$C\tau^{-1}:E\to L_x^\infty(L_l^2)$ for the relative output.
To verify the latter, first use the weighted factor Schur estimate at
fixed incoming and outgoing prices. The remaining price kernel is a
one-dimensional Gaussian at scale $\varepsilon$; its $L^2$ norm is
$C\varepsilon^{-1/2}=C\tau^{-1}$. Cauchy--Schwarz in price proves
the claim. Marginal integration in the factor has the same bound.
No factor-uniform supremum of $Z$ is used.

Consequently a term with one arbitrary density factor $Z$ and $j$
scalar directions obeys
\begin{equation}
 \|\mathrm{term}\|_{E}
   +\tau\|\mathrm{price\ marginal}/w\|_\infty
 \le C\|Z\|_{E}\prod_i\|v_i\|_\infty
 \le C\tau^{-j}\|Z\|_{E}\prod_i\|v_i\|_{K_h}.
 \label{app:skr:eq:rough-products}
\end{equation}
For the fixed known density $p$, use $|p|\le Cw\omega$ and put one
of the scalar directions in $L^2_{\rm uloc}$. Formula \eqref{app:skr:eq:mean-envelope}, Gaussian
translation estimates and the bounded relative density give instead
\begin{equation}
 \|\mathrm{term}\|_{E}
   +\tau\|\mathrm{price\ marginal}/w\|_\infty
 \le C\tau^{1-j}\prod_i\|v_i\|_{K_h},\qquad j\ge1.
 \label{app:skr:eq:smooth-products}
\end{equation}
For the supremum part of \eqref{app:skr:eq:smooth-products}, put all directions in $L^\infty$ and
use its prefactor $\tau$. For the $E$ part, the one $L^2$ direction
may occur at the terminal price, at an integrated price, or in a
mean-shift average; each case is bounded by the same Gaussian
translation estimate.

These estimates also cover the calendar kernel $J_h$. Its centered
payoff has the exact scaled form
\[
 B(v,r)=(r-v)^+-(-v)^+-r\mathbf1_{\{v<0\}},
 \quad 0\le B\le|r|,\quad B=0\text{ unless }|v|\le|r|.
\]
The payoff and incoming-price Jacobian cancel the normalization by
$h$. Parameter derivatives act on the scaled density, not on $B$;
only Gaussian moments are added. The same estimates cover
$h\partial_xJ_h$. One can differentiate the payoff once: the
conditional distribution function minus the incoming indicator is
Gaussian-localized. After one or more slope derivatives, the ordinary
chain rule gives fast reference-density derivatives times phase
increments of size $\varepsilon$; their powers cancel. The extra
factor $h$ ensures that this kernel has no larger operator bound than
the preceding ones. This argument differentiates the strike payoff,
not a rough slope.

\paragraph{The exact preconditioner and the rough-density supremum.}
Put $s_b=(\log m_b)'$. The distributional identity
$hJ_h''=\mathcal O_h^X-P^X$ gives exactly
\[
 (a_*I-h\partial_{xx})\frac{J_h}{m_b}
 =(a_*+h(s_b'-s_b^2))\frac{J_h}{m_b}
   +2h s_b\frac{\partial_xJ_h}{m_b}
   -\frac{\mathcal O_h^X-P^X}{m_b}.
\]
All parameter derivatives of $P^X$ vanish on the fixed-price affine
space. The other terms are covered by \eqref{app:skr:eq:rough-products}--\eqref{app:skr:eq:smooth-products}, including their
scaled output supremum. This cancellation is why the assertion does
not hold on an unrestricted $E$ density domain.

Finally decompose $P=p+(P-p)$. A pure $j$th slope derivative at $p$
has the bound \eqref{app:skr:eq:smooth-products}. The term with $P-p$ has the bound \eqref{app:skr:eq:rough-products}, and its
size $r\tau$ improves $\tau^{-j}$ to $r\tau^{1-j}$. A full $j$th
parameter derivative containing a density direction has only $j-1$
slope directions, because the maps are linear in the density, and
\eqref{app:skr:eq:rough-products} gives the same $\tau^{1-j}$ bound. This proves the estimates
for $1\le j\le2$.

The formulas are Fr\'echet differentiable: their kernels are smooth
functions of bounded slope and mean-shift parameters. One further
fixed fast derivative supplies a common polynomial Gaussian majorant
for differences of the differentiated kernels. Since convergence in
$K_h$ implies uniform convergence at fixed h, the same Schur bounds
give operator-norm continuity of the second derivatives.
No Gaussian derivative estimate for an arbitrary
rough-slope endpoint kernel is assumed; only fast derivatives of the
fixed smooth reference kernel were used.
\end{proof}

\subsection{Density variation in the tilted endpoint experiment}
\label{app:skr:sec:channel}

Fix a known smooth positive step with endpoint
$U=G_k(Y_{t+h}+\zeta(x,l))$ and conditional mean
$\mathbb E[U\mid x,l]=x$. Put
\[
 a(u)=1/k(u),\qquad \bar a(x,l)=\mathbb E[a(U)\mid x,l].
\]
Tilt the joint initial/terminal law by $a(U)$. Its raw marginal
 densities are
\begin{equation}
 r_-=p\bar a,\qquad r_+=p^+/k,\qquad
 p^+=\mathcal O_{t,h}(p,k).
 \label{app:skr:eq:tiltmarginals}
\end{equation}
Let $C:L^2(r_+)\to L^2(r_-)$ be conditional expectation given the
initial pair under the normalized tilted law. Its adjoint $C^*$
is the reverse conditional expectation. Put
$Q_+=I-C^*C$.

For an inverse-potential variation q with slope variation $v=q'$,
differentiating the conditional mean gives
\[
 \delta U=a(U)(Cq-q(U)),\qquad
 \delta_kp^+=\partial_u(r_+Q_+q).
\]
Define
\[
 D=r_+^{-1}\partial_u(r_+\,\cdot),\qquad
 K_x=DQ_+-Q_+\partial_u.
\]
Thus the full physical density derivative is exactly
\begin{equation}
 D\mathcal S_h[p]z
 =\mathcal O_{t,h}(z,k)+r_+(Q_+v+K_xq).
 \label{app:skr:eq:exactchannel}
\end{equation}
Here the first term is linear in z. The calibrated output has zero
price marginal. The initial shift label stays frozen throughout.

The commutator $K_x$ is self-adjoint in $L^2(r_+)$ and kills
constants, because $D^*=-\partial_u$. If $\mathcal J^G_h$ is the
symmetric physical conductance of $C^*C$, its physical kernel is
\[
 (\partial_xr_+)(y)\delta_y(dy')
   -(\partial_x+\partial_{x'})\mathcal J^G_h(y,y')\,dy'.
\]
The diagonal gives exact zero row mass. The known slow Gaussian
bounds below control its off-diagonal part and hence its symmetric
form by a product of two Gaussian jump seminorms.

\paragraph{Comparison of the known endpoint densities.}
We record the smooth-base estimates independently of the later
second-order consistency lemma:
\begin{equation}
 p^+-p=O(h)\mathfrak w,\qquad
 \bar a-1/k=O(h),\qquad r_+-r_-=O(h)\mathfrak w.
 \label{app:skr:eq:first-order-base}
\end{equation}
They hold with each required finite number of relative spatial
derivatives. The centered mean equation and the reference It\^o formula
give bounded $\partial_hb_{t,h}$ at a known smooth slope. Horizon
differentiation here holds the coefficient data, p and k fixed; it
imposes no regularity on their permitted dependence on $\delta$.
Apply the terminal forward equation once to the exact endpoint kernel.
Terminal derivatives transfer by integration by parts to the smooth
input and to simultaneous translation derivatives of the reference
kernel; the price-phase coefficient $k+b_x$ is bounded away from
zero. The Gaussian bounds give
$\partial_h\mathcal O_{t,h}(p,k)=O(\mathfrak w)$ with the required
fixed derivatives. Integrating from $h=0$ proves the first bound.
Taylor's formula for $a(U)$, using the exact conditional mean and
$\mathbb E[|U-x|^2\mid x,l]=O(h)$, proves the second. The same
differentiated mean and kernel formulas give its fixed derivatives.
The definitions of $r_\pm$ prove the third.

\subsection{An auxiliary diffusion and its jump energy}
\label{app:skr:sec:native-kernel}
\label{app:skr:sec:graph}

For the known smooth base put
\[
 r_\circ=p/k,\qquad
 \Gamma=\mathsf D_k A^R_t\mathsf D_k,\qquad
 \mathsf D_k=\operatorname{diag}(k^{-1},1).
\]
All coefficients in this subsection are frozen at the physical start
of the step. The auxiliary diffusion has generator $-H$, where
\begin{equation}
 H=-\frac1{2r_\circ}\operatorname{div}(r_\circ\Gamma\nabla),\qquad
 P_s=e^{-sH},\qquad Q_0=I-P_{2h}.
 \label{app:skr:eq:heat-comparison}
\end{equation}
Its drift is $\operatorname{div}(r_\circ\Gamma)/(2r_\circ)$.
It is smooth and bounded, and $\Gamma$ is uniformly positive.
The semigroup is reversible in $L^2(r_\circ)$; the finite invariant
measure $r_\circ$ need not have mass one. Write
$\|s\|_{Q_0}^2=\langle s,Q_0s\rangle_{r_\circ}$. For a centered
function F, its inverse-noise norm is
\[
 \|F\|_{Q_0^{-1}}
 =\sup_{\substack{\int s r_\circ=0\\ \|s\|_{Q_0}\le1}}
       |\langle F,s\rangle_{r_\circ}|.
\]

Let $G_s$ be the centered Gaussian density with covariance $sI$ on
$\mathbb R^2$. The known reference Gaussian estimates proved in
Subsection~\ref{app:skr:sec:pressure} apply also to this auxiliary
diffusion. Its jump form satisfies
\begin{equation}
 \|s\|_{Q_0}^2\asymp
 \iint r_\circ(z)G_h(z'-z)|s(z')-s(z)|^2\,dz\,dz'.
 \label{app:skr:eq:native-jump}
\end{equation}
Any fixed positive Gaussian covariance can be used on the right.
Indeed, split a larger displacement into a fixed number of equal
pieces and telescope. After translating the integration point, the
bounded logarithmic score of $r_\circ$ costs at most an exponential
of the displacement, absorbed by the Gaussian. Two-sided heat-kernel
bounds then compare the result with the actual jump form.

The product exponential weight has a Poincar\'e inequality: the
one-dimensional weight is comparable to $e^{-|x|}$, for which the
weighted Hardy inequality on each half-line gives
$\int |f-f(0)|^2e^{-|x|}\le4\int |f'|^2e^{-|x|}$;
tensorization gives the product inequality. Bounded comparison
transfers it to $r_\circ$, and uniform ellipticity gives
a positive spectral gap for H. The spectral theorem therefore gives
\begin{equation}
 \|s\|_{Q_0}^2\ge ch\|s\|_{L^2(r_\circ)}^2,
 \qquad \int s r_\circ=0.
 \label{app:skr:eq:heat-gap}
\end{equation}
In particular multiplication by a known bounded smooth function a
satisfies
\begin{equation}
 \|as\|_{Q_0}\le C\bigl(\|s\|_{Q_0}
                              +\sqrt h\|s\|_{L^2(r_\circ)}\bigr).
 \label{app:skr:eq:heat-product}
\end{equation}
This is just the product rule for a Gaussian jump and its second
moment. A price-only function is still a function on both coordinates;
the full factor and cross covariance remain in $P_s$ and $Q_0$.

The tilted operator C from the preceding subsection also has a
centered leading Gaussian of covariance $h\Gamma$, and $C^*C$ has
leading covariance $2h\Gamma$. To see this directly, its Lebesgue
kernel is
\[
 C((x,l),(u,l'))=
 \frac{q^R_{t,h}((x,l),(\xi(u)-\zeta(x,l),l'))}{\bar a(x,l)}.
\]
The terminal tilt cancels the transport Jacobian. Expand the known
smooth phase at x and use
\eqref{app:skr:eq:first-order-base}. The remainder is
$O(\sqrt h)$ times a broader Gaussian. The adjoint has the same
leading term because the density ratio and change of freezing point
cost $O(|z-z'|+h)$. Composition gives the asserted Gibbs kernel.
The required simultaneous spatial derivatives have the same bounds.
In particular its symmetric physical conductance has bounded slow
price derivatives. These are estimates only at the known base.

\paragraph{Centered short-jump kernels.}
Let $X=(x,l)$ and $Y=(y,\ell)$. Suppose a signed physical density
kernel $D(X,Y)$, allowing kernel measures and diagonal terms, has
zero column mass and the exponential Schur bounds
\begin{equation}
 \int e^{a|X-Y|^2/h}|D(X,Y)|\,dY\le V r_\circ(X),\qquad
 \int e^{a|X-Y|^2/h}|D(X,Y)|\,dX\le V r_\circ(Y)
 \label{app:skr:eq:density-kernel-schur}
\end{equation}
for a fixed $a>0$. Then
\begin{equation}
 \left\|r_\circ^{-1}\int D(\cdot,Y)f(Y)\,dY\right\|_{Q_0^{-1}}
       \le CV\|f\|_{L^2(r_\circ)}.
 \label{app:skr:eq:density-kernel-inverse}
\end{equation}
Indeed, the Schur bounds imply the jump estimate
\[
 \iint |D(X,Y)|\,|s(X)-s(Y)|^2\,dX\,dY
       \le CV\|s\|_{Q_0}^2.
\]
To prove it, insert a Gaussian midpoint $Z$ with density
$G_{bh}(Z-X)$, where $b$ is sufficiently large. The first leg uses
the row bound and native Gaussian jump comparison. For the second,
\[
 G_{bh}(Z-X)\le C G_{2bh}(Z-Y)
                      e^{|X-Y|^2/(2bh)},
\]
so the column bound applies when $1/(2b)\le a$. Column cancellation then writes the
pairing with a test as
\[
 \iint D(X,Y)[s(X)-s(Y)]f(Y)\,dX\,dY.
\]
Cauchy--Schwarz proves \eqref{app:skr:eq:density-kernel-inverse}.

\subsection{A localized scalar calibration inverse}
\label{app:skr:sec:pressure}

Let $\mathsf K_h$ be the reversible scalar Gibbs operator in the derivative of the price calibration equation, with invariant density $m(x)\asymp w(x)$. Its symmetric conductance is
$\mathcal J_h(x,y)=m(x)\mathsf K_h(x,\,\rmd y)/\,\rmd y$. For $(I-\mathsf K_h)\phi=f$, integration and symmetry give the exact flux identity
\begin{equation}
 -\int_{-\infty}^x m(y)f(y)\,\rmd y
 =\int_\R\mathcal T_h(x,z)\phi'(z)\,\rmd z,
 \label{app:skr:eq:crossingflux}
\end{equation}
where
\[
 \mathcal T_h(x,z)=
 \int_{u\le\min(x,z),\ v\ge\max(x,z)}\mathcal J_h(u,v)\,\rmd u\,\rmd v.
\]
The derivative in $x$ has a jump $-m(x)$ at $z=x$. Thus the second derivative contains exactly $-m(x)\delta_{x=z}$.
Define
\[
 B_hv(x)=\frac1{hm(x)}\int\mathcal T_h(x,z)v(z)\,\rmd z,
 \quad L_h=a_*I-h\partial_{xx},\quad D_h=L_hB_h.
\]
The diagonal jump gives $D_h=I+\mathcal R_h$ with a localized integral remainder. This exact singular coefficient is essential.

This is the derivative needed for calibration. Indeed
$m=r_+^X$ and $I-\mathsf K_h=\Pi_+Q_+\Pi_+$ on price-only functions,
where $\Pi_+$ denotes conditional expectation given price under $r_+$.
For a slope direction $v=q'$, the density identity gives
$\delta\nu=\partial_x[m(I-\mathsf K_h)q]$. Integrating with the call
boundary conditions and using \eqref{app:skr:eq:crossingflux} yields
$\delta J_h=-mB_hv$. Consequently
\[
 D_k\mathcal F_b=-D_h.
\]
This derivation precedes, and does not assume, a solution of the
scalar calibration equation.

\subsubsection{Exact conductance and scalar density remainder}
Let $\xi'=k$, let $\theta(x,l)=\xi(x)+b_k(x,l)$, and abbreviate the incoming density by $p$. The terminal-tilted initial/price-end measure has density
\[
 p(x,l)r_{t,h}^{x,l}(\xi(u)-\theta(x,l))\,\rmd x\,\rmd l\,\rmd u.
\]
Its initial marginal is $p\bar a$ and its scalar terminal marginal is $m=p_{\mathrm{out}}^X/k$. Thus the exact Gibbs conductance is
\begin{equation}
 \mathcal J_h(u,v)=\int_{\R^2}\frac{p(x,l)}{\bar a(x,l)}
 r_{t,h}^{x,l}(\xi(u)-\theta(x,l))
 r_{t,h}^{x,l}(\xi(v)-\theta(x,l))\,\rmd x\,\rmd l.
 \label{app:skr:eq:exactconductance}
\end{equation}
This formula involves no terminal factor and no change of variables through the mean shift. In particular its denominator has not been canceled by a translation-invariance argument.

Let $\gamma_v$ denote the centered scalar Gaussian of variance $v$. Uniformly in the initial pair,
\begin{equation}
 r_{t,h}^{x,l}(m)=h^{-1/2}\left[
 \gamma_{\sigma_R(t,x,l)^2}(m/\sqrt h)
 +\sqrt h\,e_{t,h}^{x,l}(m/\sqrt h)\right].
 \label{app:skr:eq:scalardensity}
\end{equation}
For $j=0,1$ and every fixed number of slow initial-pair derivatives,
\begin{equation}
 |\partial_z^j\partial_{x,l}^\alpha e_{t,h}^{x,l}(z)|
 \le C_\alpha(1+|z|)^{M_\alpha}e^{-c_\alpha z^2}.
 \label{app:skr:eq:scalarremainder}
\end{equation}
Here an initial-price derivative fixes the increment argument.

We give the kernel argument used both here and in Subsection~\ref{app:skr:sec:graph}.
Standard two-sided density and first-gradient bounds for this uniformly
elliptic class are also given by \cite[Theorem~1.2]{MePeZh21}; the simultaneous-translation
estimates needed here are explained explicitly below.
Freeze the full two-dimensional covariance at the terminal point in a Gaussian parametrix $Z_s$. Its residual, including time variation and drift, satisfies
\[
 |(\mathcal L-\mathcal L_{\mathrm{frozen}})Z_s|
 \le Cs^{-1/2}G_{cs}(z'-z).
\]
A coefficient difference is $O(|z-z'|+s)$, a second Gaussian derivative costs $s^{-1}$, and its spatial factor supplies $\sqrt s$. Gaussian convolution and the beta integral bound the convergent series $Z*\sum_{n\ge1}H^{*n}$ by $C\sqrt h\,G_{ch}$. For a first fast derivative place an initial-coordinate derivative on the outer Gaussian. The resulting time convolution has the integrable singularities $s^{-1/2}(h-s)^{-1/2}$ and is bounded by $CG_{ch}$.

Simultaneous translation of initial and terminal coordinates differentiates smooth coefficients but not the fast coordinate difference. Every differentiated coefficient difference retains first-order vanishing, so the same estimates hold after any fixed number of simultaneous derivatives. A terminal fast derivative is a simultaneous derivative minus an initial one; price translation invariance is neither available nor used. For the scalar marginal write the terminal factor as $l+v$ before integration. Initial-factor differentiation is then simultaneous translation in that coordinate too. Replacing terminal by initial freezing costs the same remainder. Gaussian integration in $v\in\R$ proves \eqref{app:skr:eq:scalardensity}--\eqref{app:skr:eq:scalarremainder}. Higher fixed fast derivative estimates follow by splitting the time interval and differentiating separate Gaussian factors, or recursively from the semigroup identity. All constants are global on both axes.

Insert this expansion in \eqref{app:skr:eq:exactconductance}. On its Gaussian scale,
\[
 \xi(u)-\theta(x,l)=k(u)(u-x)+O((u-x)^2+h),
 \qquad \bar a(x,l)=1/k(x)+O(h).
\]
Freezing all slow coefficients at $u$ and convolving the two leading Gaussians in the incoming price gives
\begin{equation}
 \frac{\mathcal J_h(u,u+\sqrt h\,z)}{m(u)}
 =h^{-1/2}[j_u(z)+\sqrt h\,E_h(u,z)].
 \label{app:skr:eq:conductanceexpansion}
\end{equation}
Here $j_u$ is the mixture of centered Gaussians with variances
$2\sigma_R(t,u,l)^2/k(u)^2$, and mixing law
$p(u,l)\,\rmd l/p^X(u)$. The remainder, its first fast derivative, and all required simultaneous slow derivatives have the bounds \eqref{app:skr:eq:scalarremainder}. The factor mixing law is comparable to $\omega$. Normalization gives $m=p^X/k+O(h)w$ in relative smooth norms. Gaussian upper and lower bounds give $m\asymp w$, and symmetry supplies the estimates with endpoints exchanged.

Only one fast derivative is needed. Off the diagonal,
\begin{equation}
 \partial_x^2\mathcal T_h(x,y)=
 \begin{cases}
 \displaystyle\int_{v\ge y}\partial_1\mathcal J_h(x,v)\,\rmd v,&x<y,\\[4pt]
 \displaystyle-\int_{u\le y}\partial_2\mathcal J_h(u,x)\,\rmd u,&x>y.
 \end{cases}
 \label{app:skr:eq:crossingsecond}
\end{equation}
The diagonal derivative jump is exactly $-m(x)$. If $s_m=m'/m$, the full kernel is
\begin{equation}
 D_h=a_*\frac{\mathcal T_h}{hm}
 -\frac{\partial_x^2\mathcal T_h}{m}
 +2s_m\frac{\partial_x\mathcal T_h}{m}
 -(s_m^2-s_m')\frac{\mathcal T_h}{m}.
 \label{app:skr:eq:Dfullkernel}
\end{equation}
Its singular part is exactly the identity. Integrating the Gaussian estimates in \eqref{app:skr:eq:crossingsecond} proves that the exact regular kernel differs from its frozen regular kernel by $E_h^D$ with
\begin{align}
 &\sup_x\int e^{\eta|x-y|/\sqrt h}|E_h^D(x,y)|\,\rmd y
 +\sup_y\int e^{\eta|x-y|/\sqrt h}|E_h^D(x,y)|\,\rmd x
 \le C\sqrt h
 \label{app:skr:eq:Schurerror}
\end{align}
for some $\eta>0$, and the same bound after every required simultaneous derivative $\partial_x+\partial_y$. In \eqref{app:skr:eq:Dfullkernel}, the last two terms have norms $O(\sqrt h)$ and $O(h)$ respectively. In the first two, coefficient displacement under the crossing integrals has Gaussian size $\sqrt h$ and gives the remaining error. This proves the differentiated comparison directly from a first, not a second, fast derivative of the heat-kernel remainder.

\begin{lemma}[Scalar localized pressure inverse]
\label{app:skr:lem:pressureinverse}
At a known smooth base, $D_h^{-1}$ is a bounded multiplication
operator plus an integral kernel with uniformly bounded exponentially
weighted Schur norm on scale $\sqrt h$. In particular, it acts
uniformly on $L^\infty$ and $L^2_{\mathrm{uloc}}$ for sufficiently
small $h$.
\end{lemma}
\begin{proof}
We construct the inverse using a Gaussian-mixture walk with a uniform
probability of killing.  The small exact-kernel remainder will be
handled in the same real-variable kernel norm.

Use the centered Gaussian mixture $j_x$ from
\eqref{app:skr:eq:conductanceexpansion}.  Its component variances are
$2v^2$, with $0<\underline v\le v\le\overline v<\infty$.
Its crossing profile is
\[
 t_x(s)=\int_{|s|}^{\infty}(u-|s|)j_x(u)\,\rmd u.
\]
For a centered Gaussian density $\gamma_{2v^2}$ and $s\ge0$,
Gaussian integration by parts gives
\[
 \int_s^\infty(u-s)\gamma_{2v^2}(u)\,\rmd u
 =2v^2\gamma_{2v^2}(s)
   -s\int_s^\infty\gamma_{2v^2}(u)\,\rmd u.
\]
Taking mixtures and using Fubini consequently gives
\begin{equation}
 0\le t_x(s)\le2\overline v^{\,2}j_x(s),\qquad
 \underline v^{\,2}\le\int_\R t_x(s)\,\rmd s
 =\frac12\int_\R u^2j_x(u)\,\rmd u
 \le\overline v^{\,2}.
 \label{app:skr:eq:killed-crossing-profile}
\end{equation}
Thus the nonnegative kernel
\[
 \mathcal K_h(x,y)
 =h^{-1/2}\left[
 j_x\!\left(\frac{y-x}{\sqrt h}\right)
 -a_*t_x\!\left(\frac{y-x}{\sqrt h}\right)\right]
\]
has row mass at most $1-\epsilon$, where
$\epsilon=a_*\underline v^{\,2}>0$.
It is the transition kernel of a walk that is killed with probability
at least $\epsilon$ at every step.

For an integral kernel $K$ write
\[
 \|K\|_{\eta,h}
 =\max\left\{
 \sup_x\int e^{\eta|x-y|/\sqrt h}|K(x,y)|\,\rmd y,
 \sup_y\int e^{\eta|x-y|/\sqrt h}|K(x,y)|\,\rmd x\right\}.
\]
These norms are submultiplicative under kernel composition.  The
Gaussian bounds give an $O(\eta)$ change from the unweighted row
mass.  For a column integral, replacing the profile frozen at
$x=y-\sqrt h\,s$ by the one frozen at $y$ costs $O(\sqrt h)$,
by its bounded slow derivatives and Gaussian tails.  Hence
\begin{equation}
 \|\mathcal K_h\|_{\eta,h}
 \le1-\epsilon+C\eta+C\sqrt h.
 \label{app:skr:eq:killed-schur-contraction}
\end{equation}

The frozen crossing operator has kernel
$h^{-1/2}t_x((y-x)/\sqrt h)$. The diagonal jump and
\eqref{app:skr:eq:Schurerror} therefore give directly
\begin{equation}
 D_h=I-\mathcal K_h+\mathcal E_h,
 \qquad \|\mathcal E_h\|_{\eta,h}\le C\sqrt h.
 \label{app:skr:eq:killed-exact-remainder}
\end{equation}
Choose a fixed $\eta>0$ small enough and then $h$ small enough.
Equations \eqref{app:skr:eq:killed-schur-contraction} and
\eqref{app:skr:eq:killed-exact-remainder} give
$\|\mathcal K_h-\mathcal E_h\|_{\eta,h}
\le1-\epsilon/2$.  Thus
\begin{equation}
 D_h^{-1}
 =\sum_{n=0}^{\infty}
       (\mathcal K_h-\mathcal E_h)^n
 \label{app:skr:eq:killed-neumann-inverse}
\end{equation}
is a two-sided inverse.  Without the remainder, this is the expected
occupation operator of the killed walk.  The exact inverse is dominated
in absolute value by the occupation series with nonnegative kernel
$\mathcal K_h+|\mathcal E_h|$, which is still substochastic
for the same choice of parameters.  The walk has a uniformly bounded
expected lifetime, and the exponentially weighted norm controls its
displacement before killing.  Thus the integral part of the inverse
has the required uniform localization. Exponential density weights
are absorbed by these kernels and their bounded logarithmic scores.

Preconditioning adds no solutions: the homogeneous equation
$(a_*I-h\partial_{xx})u=0$ has only combinations of
$e^{\sqrt{a_*/h}x}$ and $e^{-\sqrt{a_*/h}x}$. Their growing tails
are excluded by the stated spaces for small $h$.
\end{proof}

\subsection{A local calibrated branch on a shrinking neighborhood}
\label{app:skr:sec:localscalar}

\begin{lemma}[Scaled local scalar calibration]
\label{app:skr:lem:localscalar}
At a known smooth positive base with its own terminal price law,
there are fixed $r,C>0$ such that, for
\[
 P\in p+\mathcal Z,\qquad
 \|P-p\|_{E}+\|R\|_{K_h}<r h^{1/4},
\]
the equation $\mathcal F_b(P,k)+R=0$ has a unique local positive
slope $k=\kappa_h(P,R)$ near the base slope. The slope and endpoint
maps are $C^2$ with
\[
 \|D\kappa_h\|+\|D\mathcal S_h\|\le C,\qquad
 \|D^2\kappa_h\|+\|D^2\mathcal S_h\|\le Ch^{-1/4}.
\]
The displacement of the slope in $K_h$ is at most
$C(\|P-p\|_{E}+\|R\|_{K_h})$. Signed densities are allowed.
\end{lemma}
\begin{proof}
At the known base $D_k\mathcal F_b=-D_h$.
The already constructed inverse is bounded in $L^2_{\rm uloc}$ and
$L^\infty$ by its weighted Schur kernel, hence uniformly in $K_h$.
The preceding second derivative bound times the radius $rh^{1/4}$
is $Cr$; choose $r$ so that the slope derivative remains invertible
and the slope stays positive. The quantitative implicit-function
theorem gives the branch and the uniform first derivative. Its
second implicit derivative identity gives $Ch^{-1/4}$. Composition with the endpoint map
gives the same estimates. The residual is an additive Banach
parameter, so no regularity of it beyond $K_h$ is assumed.
\end{proof}

\subsection{Stability in an explicit weighted density norm}
\label{app:skr:sec:strongchannel}

The auxiliary reversible diffusion retains the full joint covariance;
its difference from the exact density channel is a centered short-jump
error. Its heat interpolation is an auxiliary analytic construction,
not a Markov representation of the physical Bass bridge.

\begin{lemma}[A weighted heat interpolation]
\label{app:skr:lem:heat-semiconvexity}
Let $r_\circ\asymp w\omega$ be a known smooth density with bounded
relative derivatives, let $\Gamma$ be a known smooth covariance
matrix with $cI\le\Gamma\le CI$ and bounded derivatives of the
required orders, and put
\[
 H=-\frac1{2r_\circ}\operatorname{div}(r_\circ\Gamma\nabla),\qquad
 P_s=e^{-sH},\qquad Q_h=I-P_{2h}.
\]
Thus $P_s$ is the reversible Markov semigroup in $L^2(r_\circ)$.
Let $\Pi$ denote conditional expectation given price. Suppose $W$
is positive, bounded above and below with the required bounded
derivatives, and
\[
 W(x,l)\Gamma_{11}(x,l)=\kappa(x)
\]
is price-only. For $f,v\in L^2(r_\circ)$ with $v$ price-only, set
\[
 g=P_hf+Q_hv,\qquad \Pi g=0.
\]
There are uniform $c,C,h_0>0$ such that, for $h\le h_0$,
\begin{equation}
 \int Wg^2r_\circ+c\|g\|_{Q_h}^2
 \le\int Wf^2r_\circ+Ch\bigl(\|f\|_{L^2(r_\circ)}^2+
                                  \|v\|_{L^2(r_\circ)}^2\bigr).
 \label{app:skr:eq:heat-weighted-stability}
\end{equation}
The estimate retains the full two-dimensional heat transport.
\end{lemma}
\begin{proof}
All unlabelled $L^2$ norms use $r_\circ$. Write
\[
 a_s=P_s(f-P_hv),\qquad b_s=P_{h-s}v,\qquad
 u_s=a_s+b_s,\qquad 0\le s\le h.
\]
Then $u_0=f$, $u_h=g$, $u_s'=-Ha_s+Hb_s$, and
$u_s''=H^2u_s$. Put
$E(s)=\langle Wu_s,u_s\rangle_{r_\circ}$ and $K=[H,W]$.
Self-adjointness of $H$ gives the exact second-variation identity
\begin{equation}
 E''(s)=4\|Ha_s\|_{L^2(Wr_\circ)}^2
       +4\|Hb_s\|_{L^2(Wr_\circ)}^2
       +2\langle Ku_s,Hu_s\rangle_{r_\circ}.
 \label{app:skr:eq:heat-second-variation}
\end{equation}
Indeed,
$E''=2\|u_s'\|_{L^2(Wr_\circ)}^2
 +2\langle Wu_s,H^2u_s\rangle_{r_\circ}$.
Move one H past W in the second term and use
$Hu_s=Ha_s+Hb_s$; the two squared norms cancel their cross terms.

The commutator is the known first-order operator
\[
 Ku=-\Gamma\nabla W\cdot\nabla u
       -\frac{\operatorname{div}(r_\circ\Gamma\nabla W)}{2r_\circ}u.
\]
Its coefficients are bounded. Moreover
\[
 \|\nabla u\|_2^2\le C\langle u,Hu\rangle_{r_\circ}
       \le C\|u\|_2\|Hu\|_2.
\]
Use this inequality and Young's inequality to absorb the commutator
into the first two terms of
\eqref{app:skr:eq:heat-second-variation}. Since
$Hu_s=Ha_s+Hb_s$, this proves
\begin{equation}
 E''(s)\ge-CE(s).
 \label{app:skr:eq:heat-semiconvexity}
\end{equation}
Only the $L^2$ size of the scalar pressure enters the resulting bound.
The ordinary $L^2(r_\circ)$ contraction of $P_s$ also gives
\[
 \sup_{0\le s\le h}E(s)
       \le C(\|f\|_2^2+\|v\|_2^2).
\]

At the terminal time,
\[
 E'(h)=-2\langle Wg,Hg\rangle_{r_\circ}
                         +4\langle Wg,Hv\rangle_{r_\circ}.
\]
The first term is at most
$-c\|\nabla g\|_2^2+C\|g\|_2^2$, by one integration by parts.
For a price-only v, put
\[
 d=\frac{\partial_x(r_\circ\Gamma_{11})+
                  \partial_l(r_\circ\Gamma_{21})}{r_\circ}.
\]
Then
\[
 WHv=-\tfrac12\kappa(x)v_{xx}-\tfrac12Wd\,v_x.
\]
The first term pairs to zero with g because $\Pi g=0$.
Integrating the second once in price moves its derivative onto g
and known coefficients. Cauchy--Schwarz and Young therefore give
\begin{equation}
 E'(h)\le-c\|\nabla g\|_2^2
                       +C\bigl(E(h)+\|v\|_2^2\bigr).
 \label{app:skr:eq:heat-terminal-derivative}
\end{equation}

Integrating \eqref{app:skr:eq:heat-semiconvexity} backward from h
gives
\[
 E(h)\le E(0)+hE'(h)+C\int_0^h sE(s)\,ds.
\]
Use the interior bound and
\eqref{app:skr:eq:heat-terminal-derivative}, and absorb $ChE(h)$
for small h. This yields
\[
 E(h)+ch\|\nabla g\|_2^2
       \le E(0)+Ch(\|f\|_2^2+\|v\|_2^2).
\]
Finally,
\[
 \|g\|_{Q_h}^2
   =2\int_0^h\langle P_sg,HP_sg\rangle_{r_\circ}\,ds
   \le2h\langle g,Hg\rangle_{r_\circ}
   \le Ch\|\nabla g\|_2^2,
\]
which proves \eqref{app:skr:eq:heat-weighted-stability}.

These computations extend to the stated $L^2$ inputs. For fixed
$h>0$, the calibration identity reads
\[
 v=\Pi P_{2h}v-\Pi P_hf.
\]
The known heat-kernel derivative bounds make its right-hand side
smooth in price, with each required weighted derivative finite at
that fixed h. Thus $v$ and g have the regularity used at s=h;
no uniform higher derivative estimate is asserted. On $0<s<h$
both heat factors are smooth. Apply the integrated inequality from
$s=\epsilon$ to h and let $\epsilon\downarrow0$, using strong
semigroup continuity and the uniform interior bound.
\end{proof}

For the known base use the auxiliary density $r_\circ=p/k$ and
covariance $\Gamma$ from \eqref{app:skr:eq:heat-comparison}, and set
\[
 W=\frac1{k\sigma_R^2}.
\]
Then $W\Gamma_{11}=k^{-3}$ is price-only and
\[
 \int W(z/r_\circ)^2r_\circ
       =\int\frac{z^2}{\sigma_R^2p}.
\]
This is the explicit physical weight in the standalone theorem.

\paragraph{Comparison of the exact kernels.}
Retain $r_+=p^+/k$, the known tilted conditional operator C and
$Q_+=I-C^*C$ solely to write the exact slope derivative. Define
\[
 \mathcal A f=r_\circ^{-1}\mathcal O_h(r_\circ f,k),\qquad
 L=r_+/r_\circ,\qquad
 \mathcal B=LQ_+-Q_0,\qquad
 \mathcal J=LK_x.
\]
The exact derivative, with $f=z/r_\circ$, $v=D\kappa_h[p]z$ and $q'=v$,
is
\begin{equation}
 g:=\frac{D\mathcal S_h[p]z}{r_\circ}
   =P_hf+Q_0v+(\mathcal A-P_h)f+\mathcal Bv+\mathcal Jq.
 \label{app:skr:eq:heat-density-identity}
\end{equation}
Here $\Pi f=\Pi g=0$, where $\Pi$ is conditional averaging on
price under $r_\circ$, and the scalar branch bound gives
$\|v\|_{L^2_{\rm uloc}}\le C\|z\|_E$.

Both $\mathcal A$ and $P_h$ preserve the integral against $r_\circ$.
Their kernels have the same leading Gaussian of covariance $h\Gamma$.
Thus $\mathcal A-P_h$ has exactly zero column mass and an
$O(\sqrt h)$ Gaussian majorant. This comparison keeps the full
factor and cross covariance; it does not replace the action on a
rough density by a frozen-factor convolution.

Likewise $C^*C$ and $P_{2h}$ have the same leading Gaussian of
covariance $2h\Gamma$, and $L=1+O(h)$ with known smooth bounds.
Consequently $\mathcal B$ has an $O(\sqrt h)$ Gaussian regular
kernel and the diagonal term $(L-1)I$, of size $O(h)$. Its column
mass is exactly zero, because
\[
 \int r_\circ LQ_+v=\int r_+Q_+v=0=\int r_\circ Q_0v
\]
for every test function v. The diagonal term must be retained in
this cancellation.

Finally $\mathcal J$ is self-adjoint in $L^2(r_\circ)$ and kills constants:
\[
 \langle s,\mathcal Jq\rangle_{r_\circ}
       =\langle s,K_xq\rangle_{r_+}.
\]
Its symmetric off-diagonal conductance is bounded by a known full
Gaussian conductance. Therefore its form is bounded by the product
of the two Gaussian jump seminorms. Since q is price-only,
\[
 q(x+d)-q(x)=d\int_0^1v(x+sd)\,ds,
\]
and its jump energy on any translated price window is at most
$Chw(n)\|v\|_{L^2_{\rm uloc}}^2$. Only these increments occur;
no choice or estimate of a local pressure constant is required.

\paragraph{The complete localized error.}
Put $e=\|z\|_E$, $\chi=\chi_n$ and $w_n=w(n)$. From
\eqref{app:skr:eq:heat-density-identity},
\begin{equation}
 \chi g=P_h(\chi f)+Q_0(\chi v)+E_n,
 \label{app:skr:eq:heat-cutoff-identity}
\end{equation}
where
\[
 E_n=[\chi,P_h]f+[\chi,Q_0]v
       +\chi(\mathcal A-P_h)f+\chi\mathcal Bv+\chi\mathcal Jq.
\]
The complete error is centered against $r_\circ$: both $\chi f$ and
$\chi g$ are price-centered, and $P_h$ preserves $r_\circ$ while $Q_0$
has zero integral. Moreover,
\begin{equation}
 \|E_n\|_{Q_0^{-1}}\le C\sqrt h\sqrt{w_n}\,e.
 \label{app:skr:eq:heat-cutoff-error}
\end{equation}
Here and below the inverse norm is on centered functions.

We give the localization details. The two Markov commutators have
kernels $O(\sqrt h)$ times a Gaussian. Their row and column masses
are $O(h)$, since $P_h\chi-\chi=O(h)$, and similarly at time $2h$.
The bounds are supported near the price window, with Gaussian tails.
If D is the physical density kernel (including its output weight
$r_\circ$) of either centered operator $\mathcal A-P_h$ or
$\mathcal B$, then the column mass created by multiplication by
$\chi$ is
\[
 \int D(X,Y)[\chi(X)-\chi(Y)]\,dX.
\]
Its size is $O(h)$: D already costs $\sqrt h$, and the cutoff
displacement supplies another $\sqrt h$. A diagonal term contributes
zero to this expression. Subtract this column mass on the diagonal;
the remaining centered kernel has $O(\sqrt h)$ exponential Schur
norm. The midpoint bound controls it, while the $O(h)$ multiplication
term is bounded using \eqref{app:skr:eq:heat-gap}. Weighted Schur
estimates and the global uniformly local input bounds supply exactly
$\sqrt{w_n}e$. Far windows are summable by Gaussian tails and the
bounded logarithmic weight derivatives.

For the last term, pair $\chi\mathcal Jq$ with a centered test s
and use the symmetric form on $\chi s$ and q. Only pairs touching
the price window remain. The local q-jump bound above supplies
$C\sqrt h\sqrt{w_n}e$. The product rule bounds the jump of
$\chi s$ by $C(\|s\|_{Q_0}+\sqrt h\|s\|)$, and the gap
absorbs the second term. This proves
\eqref{app:skr:eq:heat-cutoff-error} without a factor derivative
of an input density or a localized projection assertion.

\paragraph{A scalar correction for the auxiliary heat experiment.}
Solve, with zero $r_\circ$-mean, the global scalar equation
\begin{equation}
 \Pi Q_0\Pi w_n^\sharp=\Pi E_n.
 \label{app:skr:eq:heat-pressure-correction}
\end{equation}
The gap gives existence by the elementary coercive Hilbert-space
argument. Testing with $w_n^\sharp$ yields
\[
 \|w_n^\sharp\|_{Q_0}\le\|E_n\|_{Q_0^{-1}},\qquad
 \|w_n^\sharp\|_{L^2(r_\circ)}\le C\sqrt{w_n}e.
\]
Set
\[
 v_n=\chi_n v+w_n^\sharp,\qquad
 u_n=P_h(\chi_n f)+Q_0v_n,\qquad
 R_n=E_n-Q_0w_n^\sharp.
\]
Then $\Pi u_n=\Pi R_n=0$, $\chi_n g=u_n+R_n$,
and $\|v_n\|_{L^2(r_\circ)}\le C\sqrt{w_n}e$. Orthogonality gives
\begin{equation}
 \|R_n\|_{Q_0^{-1}}\le\|E_n\|_{Q_0^{-1}}.
 \label{app:skr:eq:heat-projected-error}
\end{equation}
This is a single global projection applied separately to each
complete centered window error. No sharp pressure bound or localized
projection theorem is needed.

\paragraph{The explicit density energy.}
Apply Lemma~\ref{app:skr:lem:heat-semiconvexity} to the input $\chi_n f$ and
the scalar $v_n$. With
\[
 d_n=\|u_n\|_{Q_0}^2,
\]
it gives
\begin{equation}
 \|u_n\|_{L^2(Wr_\circ)}^2+c d_n
 \le\|\chi_n f\|_{L^2(Wr_\circ)}^2+Chw_ne^2.
 \label{app:skr:eq:heat-local-dissipation}
\end{equation}
Multiplication by the known
smooth W obeys
\[
 \|Wu_n\|_{Q_0}
 \le C\bigl(\|u_n\|_{Q_0}+\sqrt h\|u_n\|_{L^2(r_\circ)}\bigr).
\]
Pair $R_n$ with $Wu_n$, use
\eqref{app:skr:eq:heat-projected-error}, and absorb a fraction of
$d_n$ by Young's inequality. Since $Q_0\le I$, the squared
$L^2$ norm of $R_n$ is at most its squared inverse norm. Therefore
\begin{equation}
 \|\chi_n g\|_{L^2(Wr_\circ)}^2+c d_n
 \le\|\chi_n f\|_{L^2(Wr_\circ)}^2+Chw_ne^2,
 \qquad
 \|\chi_n g\|_{Q_0}
 \le C\bigl(d_n^{1/2}+\sqrt h\sqrt{w_n}e\bigr).
 \label{app:skr:eq:explicit-local-linear}
\end{equation}
Because $W/r_\circ=1/(\sigma_R^2p)$, division by $w_n$ and the supremum
give the linear $1+Ch$ estimate in
\eqref{app:skr:eq:explicit-density-norm}.

Thus
\begin{equation}
 \mathcal N_b(D\mathcal S_h[p]z)\le(1+Ch)\mathcal N_b(z).
 \label{app:skr:eq:strongchannel}
\end{equation}
The retained output dissipation in
\eqref{app:skr:eq:explicit-local-linear} will absorb the nonlinear
remainder. Along a known family, metric comparison follows directly
from the pointwise ratio of the two weights $(p\sigma_R^2)^{-1}$.
No regularity of a slope family is needed for that comparison.

\subsection{The nonlinear remainder from centered density kernels}
\label{app:skr:sec:quadratic}

All noise norms in this subsection use the original known smooth
base's auxiliary density $r_\circ=p/k$ and heat noise $Q=Q_0$.
Write $\tau=h^{1/4}$. Evaluated slopes need only be measurable and
bounded above and below by fixed positive constants. The input
coefficient of a kernel will either be the known density
$p\le Cw\omega$, or an arbitrary signed density in $E$.

For scalar price inputs use a kernel $D_l(x,y)$ with
\begin{equation}
 \int D_l(x,y)\,dx=0,\qquad
 \begin{aligned}
 \int e^{a|x-y|^2/h}|D_l(x,y)|\,dy&\le Vw(x)\omega(l),\\
 \int e^{a|x-y|^2/h}|D_l(x,y)|\,dx&\le Vw(y)\omega(l).
 \end{aligned}
 \label{app:skr:eq:scalar-density-kernel-schur}
\end{equation}
Embed it as a joint kernel supported on equal factor coordinates
and apply the preceding argument to $f(y,\ell)=v(y)$. Thus its
inverse-noise bound is $CV\|v\|_{L^2(w)}$. The column bound and
cancellation in \eqref{app:skr:eq:scalar-density-kernel-schur} hold
separately for every output factor; averaging them over the factor
would not suffice.

These estimates localize to any fixed translated price cutoff
$\chi_n$. For an input cut off on a fixed larger interval, test the
uncut centered operator against $\chi_n s$. For centered $s$,
\[
 \|\chi_ns\|_Q
 \le C(\|s\|_Q+\sqrt h\|s\|_{L^2(r_\circ)})
 \le C\|s\|_Q
\]
by the native gap. The near-input norm is at most
$C\sqrt{w(n)}$ times its uniformly local norm. Other input windows
have an extra factor $e^{-c\,{\rm dist}^2/h}$ in Schur norm. An
ordinary $L^2$ bound followed by the centered gap estimates these
terms; the Gaussian tail absorbs its $h^{-1/2}$ factor and all
price-weight ratios. Consequently each local centered-test norm is
bounded by $CV\sqrt{w(n)}$ times the corresponding global uniformly
local input norm. No separate localized raw source is declared
centered. This argument uses the single cutoffs defining the density
energy, without derivative components or an additional window sum.

\paragraph{Exact phase derivatives and their localization.}
Write the conditional endpoint density as
\[
 K_k(x,l;u,l')=k(u)q_{t,h}^{x,l}
       \left(\int_x^u k(y)\,dy-b_k(x,l),l'\right).
\]
For slope directions $v_1,v_2$ put
$B_i=\int_x^u v_i(y)\,dy-D_kb_k[v_i]$ and
$b_{12}=D_k^2b_k[v_1,v_2]$. Fast derivatives of the fixed reference
kernel give the exact formulas
\begin{align}
 D_kK_k[v_1]
   &=v_1(u)q+k(u)q_mB_1,\nonumber\\
 D_k^2K_k[v_1,v_2]
   &=v_1(u)q_mB_2+v_2(u)q_mB_1
       +k(u)q_{mm}B_1B_2-k(u)q_m b_{12}.
 \label{app:skr:eq:density-phase-derivatives}
\end{align}
No spatial derivative of a slope occurs.

The mean-root derivatives have actual Gaussian kernels in the slope
directions. To see this, in rescaled variables put
\[
 A_i(r)=\int_0^r v_i(x+\sqrt h s)\,ds,
 \qquad \mathfrak q(r)=\bar Q_h(x,l;A_k(r)-\beta_k),
 \qquad \beta_k=b_k/\sqrt h.
\]
Integration by parts in the first variation of the mean equation
cancels the terminal slope term. Its first two derivatives are
exactly
\begin{equation}
 \beta_i=\frac{\int\mathfrak q A_i}{\int\mathfrak q},\qquad
 \beta_{12}=\frac{\int\mathfrak q_m
                  (A_1-\beta_1)(A_2-\beta_2)}{\int\mathfrak q}.
 \label{app:skr:eq:density-mean-derivatives}
\end{equation}
The denominator is uniformly positive. Gaussian-tail Fubini in the
first formula, and the bound
$|A_1-\beta_1|\le C\|v_1\|_\infty(1+|r|)$ in the second, show
\begin{equation}
 |D_kb_k[v]|\le C\sqrt h\,\mathcal G_h|v|,\qquad
 |D_k^2b_k[v_1,v_2]|
       \le C\sqrt h\|v_1\|_\infty\mathcal G_h|v_2|,
 \label{app:skr:eq:density-mean-kernels}
\end{equation}
where $\mathcal G_h$ is a fixed Gaussian price average about the
initial price. These inequalities follow from the explicit signed
kernels, not from an inference that an operator bound implies
localization.

When all scalar directions are bounded,
$|B_i|\le C\|v_i\|_\infty(|u-x|+\sqrt h)$.
The fast Gaussian bounds in
\eqref{app:skr:eq:density-phase-derivatives} therefore give joint
Gaussian majorants with constants $C\|v_1\|_\infty$ and
$C\|v_1\|_\infty\|v_2\|_\infty$, respectively.
Every column has zero mass: changing a slope preserves the
conditional terminal factor law, separately for each initial state.
After weighting a signed input $Z$ by $Z/r_\circ$, the joint kernel
estimate thus gives, in each normalized local centered-test norm,
\begin{align}
 \|D_k\mathcal O_h[Z,k]v\|_{Q^{-1}}
      &\le C\|v\|_\infty\|Z\|_E,\nonumber\\
 \|D_k^2\mathcal O_h[Z,k](v_1,v_2)\|_{Q^{-1}}
      &\le C\|v_1\|_\infty\|v_2\|_\infty\|Z\|_E.
 \label{app:skr:eq:rough-density-noise}
\end{align}
Here and below a density source is divided by the fixed $r_\circ$ in
an inverse-noise norm.

For the known input $p\le Cw\omega$, integrate
\eqref{app:skr:eq:density-phase-derivatives} against $p(x,l)$ and
keep one direction in local $L^2$. Its possible placements are the
terminal value $v(u)$, a segment average
$v(x+\theta(u-x))$, or the Gaussian mean-root average in
\eqref{app:skr:eq:density-mean-kernels}. The first gives a diagonal
atom. In the second, change variables
$y=(1-\theta)x+\theta u$ after the Gaussian input integration;
the resulting price kernel has variance $C(1-\theta)^2h$ about
$u$. The third convolves two fixed-scale Gaussian price kernels.
Polynomial displacement factors are absorbed into broader Gaussians
before this change of variables. All positive mixing measures have
uniformly bounded mass, including the endpoint $\theta=1$.
Gaussian integration of the known factor weight leaves
$C\omega(l')$ uniformly on the whole factor axis.

These are exactly the row and column exponential Schur bounds in
\eqref{app:skr:eq:scalar-density-kernel-schur}: Gaussian mixtures of
scales at most $Ch$ and diagonal atoms have these bounds uniformly.
For the second derivative fix $v_1$ in $L^\infty$; its kernel
constant is multiplied by at most $C\|v_1\|_\infty$.
Factor conservation for every scalar direction gives zero price
mass of every scalar column separately at each output factor,
including its diagonal part. Consequently
\begin{align}
 \|D_k\mathcal O_h[p,k]v\|_{Q^{-1}}
       &\le C\|v\|_{L^2_{\rm uloc}},\nonumber\\
 \|D_k^2\mathcal O_h[p,k](v_1,v_2)\|_{Q^{-1}}
       &\le C\|v_1\|_\infty\|v_2\|_{L^2_{\rm uloc}}.
 \label{app:skr:eq:known-density-noise}
\end{align}
The evaluated k may be rough; only its uniform positive bounds have
been used. All target norms remain those of the original known base.

\begin{lemma}[Second-order density remainder in inverse noise]
\label{app:skr:lem:quadratic-noise}
Let $\mathcal S_h$ be the calibrated branch from the $C^2$ scalar
implicit theorem. If $z\in\mathcal Z$ and $e=\|z\|_E\le c\tau$,
then
\begin{equation}
 \mathcal S_h(p+z)-\mathcal S_h(p)=D\mathcal S_h[p]z+R_2,\qquad
 \left\|\frac{\chi_nR_2}{r_\circ}\right\|_{Q^{-1}}
      \le C\sqrt{w(n)}\,\tau^{-1}e^2.
 \label{app:skr:eq:quadratic-noise}
\end{equation}
The complete localized source is centered. The constants and the
noise form belong to the original known base.
\end{lemma}
\begin{proof}
Put $P_s=p+sz$, $k_s=\kappa_h(P_s,0)$,
$v_s=\partial_sk_s$ and $a_s=\partial_s^2k_s$.
The uniform first and second implicit derivative bounds imply
\[
 \|v_s\|_{L^2_{\rm uloc}}\le Ce,\quad
 \|v_s\|_\infty\le C\tau^{-1}e,\quad
 \|a_s\|_{L^2_{\rm uloc}}\le C\tau^{-1}e^2,\quad
 \|a_s\|_\infty\le C\tau^{-2}e^2.
\]
The exact second derivative is
\[
 \partial_s^2\mathcal S_h(P_s)
 =2D_{pk}\mathcal O_h[z,v_s]
     +D_{kk}\mathcal O_h[P_s](v_s,v_s)
     +D_k\mathcal O_h[P_s]a_s.
\]
Decompose $P_s=p+sz$ in the last two terms. The joint and scalar
kernel estimates bound the resulting five terms by
\[
 C\sqrt{w(n)}\{\tau^{-1}e^2+\tau^{-2}e^3\}
       \le C\sqrt{w(n)}\tau^{-1}e^2.
\]
This split uses the known bounded p only in the scalar estimate;
the other input sz is allowed to be signed and need not be pointwise
close to p. Integrating the second derivative with weight $1-s$
proves the $C^2$ remainder identity and its inverse-noise estimate.
Exact calibration makes each full second derivative have zero price
marginal, so multiplying it by $\chi_n$ preserves its total mass
zero. This justifies the actual inverse-noise norm of the sum;
no nonlinear pressure projection is needed.
\end{proof}

\paragraph{Absorption by the output dissipation.}
Normalize $d_n$ from \eqref{app:skr:eq:explicit-local-linear} by
$w(n)$. The localized linear estimate retains this nonnegative
output dissipation, and multiplication by W gives output Q norm
at most $C(d_n^{1/2}+\sqrt h e)$.
Pair \eqref{app:skr:eq:quadratic-noise} with the linear output in
the explicit physical energy. Young's inequality gives
\[
 C\tau^{-1}e^2(d_n^{1/2}+\sqrt h e)
 \le\epsilon d_n+Ch e^2+C\tau^{-2}e^4.
\]
The squared remainder is at most $C\tau^{-2}e^4$, because W is
bounded and $Q\le I$. Absorb the chosen fraction of the same-window
output dissipation, divide by the price weight and take the supremum.
Norm equivalence yields
\[
 \mathcal N_b(\mathcal S_h(p+z)-\mathcal S_h(p))^2
 \le(1+Ch)\mathcal N_b(z)^2+Ch^{-1/2}\mathcal N_b(z)^4.
\]
Taking square roots gives
\begin{equation}
 \mathcal N_b(\mathcal S_h(p+z)-\mathcal S_h(p))
 \le(1+Ch)\mathcal N_b(z)+Ch^{-1/2}\mathcal N_b(z)^3.
 \label{app:skr:eq:cubic-stability}
\end{equation}
Only the output jump, $C^2$ endpoint calculus and column conservation
are needed; no derivative of an incoming density is taken.

\begin{proof}[Completion of the proof of
Theorem~\ref{app:skr:thm:wmot-calibration}]
Use $h=\delta$ and $t=0$. Lemmas~\ref{app:skr:lem:finite}
and~\ref{app:skr:lem:localscalar} construct the signed calibrated
branch on a ball of radius $r\delta^{1/4}$ in $p+\mathcal Z$.
Exact price calibration puts all fixed-target endpoint differences
in $\mathcal Z$. Estimate~\eqref{app:skr:eq:explicit-local-linear} and
Lemma~\ref{app:skr:lem:quadratic-noise} give
\eqref{app:skr:eq:cubic-stability}.

Use $\mathcal N_b$ at both endpoints. For small $\delta$, norm
equivalence puts the ball $\mathcal N_b(z)\le\delta^{3/4}$ inside
the signed branch's $r\delta^{1/4}$ neighborhood. On this ball,
\[
 \delta^{-1/2}\mathcal N_b(z)^3
       \le\delta\mathcal N_b(z).
\]
Thus \eqref{app:skr:eq:cubic-stability} proves
\eqref{app:skr:eq:unit-stability}; pointwise comparison of the
explicit weights proves \eqref{app:skr:eq:unit-metrics}.

It remains to identify the calibrated endpoint for a nonnegative
incoming density. The kernel is positive, its conditional price mean
is the incoming price, and its terminal price law is the prescribed
one. Choose $\psi'=\xi_\kappa$ for its calibrated slope $\kappa$.
The identity $\xi_\kappa(U)=Y_1+\zeta_\kappa(x,l)$ gives, for any
competing martingale endpoint $Z$ with the same terminal price law,
coupled with the reference endpoint conditional on $(x,l)$,
\[
 \E[UY_1]-\E[ZY_1]
 =\E\bigl[\psi(Z)-\psi(U)-\psi'(U)(Z-U)\bigr]
 \ge\frac c2\E|Z-U|^2.
\]
Indeed, the common terminal marginal cancels the potential terms,
and the conditional martingale means cancel the frozen shift.
All terms are integrable because the slope is bounded and the laws
have second moments. Maximizing the reference covariance is the WMOT
problem of the main text. The displayed identity proves optimality
and uniqueness, without an additional dual-attainment assertion.
The same argument applies whenever the local scalar construction succeeds.

\end{proof}

\subsection{Smooth consistency at a known input}
\label{app:skr:sec:wmot-consistency}

The consistency statement concerns only known smooth input data. Write $h=\delta$ and allow a translated physical start $t$.
\begin{lemma}[First-order smooth consistency]
\label{app:skr:lem:wmot-consistency}
For known smooth $p,k$ with the stated positive bounds,
\begin{equation}
 \mathcal O_{t,h}(p,k)=p+h\mathcal G_t(k)^*p+O(h^2),
 \label{app:skr:eq:expansion}
\end{equation}
where
\begin{equation}
 \mathcal G_t(k)=\tfrac12(\mathsf D_k A^R_t\mathsf D_k):D^2
                   +b^R_t\partial_l,
 \qquad\mathsf D_k=\operatorname{diag}(k^{-1},1).
 \label{app:skr:eq:leadingG}
\end{equation}
The endpoint remainder is $O(h^2)$ in $E$ and in every fixed finite
relative smooth norm needed below. The normalized calendar satisfies
\begin{equation}
 J_h(p,k)(x)=\frac1{2k(x)^2}
                 \int\sigma_R(t,x,l)^2p(x,l)\,dl+O(h),
 \label{app:skr:eq:smooth-calendar}
\end{equation}
with the remainder divided by $w$ in each required fixed smooth norm.
For each required output order $s$, a sufficiently large fixed $r$
also gives the directional expansion
\begin{equation}
 D_kJ_h(p,k)[v]
 =-\frac1{k^3}
       \left(\int\sigma_R(t,\cdot,l)^2p(\cdot,l)\,dl\right)v
       +O(h)w\|v\|_{C_b^r},
 \label{app:skr:eq:directional-calendar-consistency}
\end{equation}
where the relative remainder is measured in $C_b^s$.
The constants are uniform on known bounded positive smooth sets.
Only finitely many known spatial and physical-time derivatives are used.
\end{lemma}
\begin{proof}
The ordinary reference semigroup and the scalar mean equation give a
smooth centering shift at $h=0$, with bounded first two horizon
derivatives in every required finite spatial norm, and
\[
 b_{t,h}(x,l)=\frac h2\sigma_R(t,x,l)^2
                   \frac{k'(x)}{k(x)^2}+O(h^2).
\]
To obtain a strong remainder, write the reference transition density as
$q_h^R((x,l),(y,l'))$ and substitute
$y=x+\xi_k(z)-\xi_k(x)-b_{t,h}(x,l)$.
For the substituted kernel $\widetilde q$, put
$\mathcal D_x=\partial_{\mathrm{initial}\,x}+\partial_y$ and
$\mathcal D_l=\partial_{\mathrm{initial}\,l}+\partial_{l'}$.
The exact identities
\[
 \partial_x\widetilde q=\mathcal D_xq-(k+b_x)\partial_yq,
 \qquad
 \partial_l\widetilde q=\mathcal D_lq-\partial_{l'}q-b_l\partial_yq
\]
transfer terminal derivatives to the known incoming density or to
simultaneous kernel derivatives. Here $k+b_x$ stays positive, and the
simultaneous derivatives obey the previously proved Gaussian bounds
without a fast short-time loss. Apply the reference forward equation
for time derivatives of the kernel and repeat these two identities.
Gaussian Schur bounds and the exponential weights then bound the first
two $h$ derivatives of the endpoint density, with any fixed required
output derivatives, uniformly down to zero. The first $h$ derivative
therefore has a strong right limit at zero. It\^o's formula identifies
that limit as $\mathcal G_t(k)^*p$, and Taylor's formula gives
\eqref{app:skr:eq:expansion}.

Apply the call primitive to the
price marginal of \eqref{app:skr:eq:expansion}. Mass and first-moment
cancellation permit the two-tail weighted bound. Since
$hJ_h''=\mathcal O_{t,h}^X-p^X$, this gives
\eqref{app:skr:eq:smooth-calendar}, including its fixed scalar
derivative bounds. Differentiate this calculation once in a smooth
slope direction. The derivative-transfer identities still apply;
their coefficients and slope derivatives have bounded finite smooth
norms. Thus the same two horizon-derivative bounds hold after this
parameter differentiation. Differentiating the leading calendar
gives \eqref{app:skr:eq:directional-calendar-consistency}, including
its $O(h)$ remainder. No expansion at a rough input density is used.
\end{proof}

\subsection{The smooth target density by a weighted Feynman--Kac formula}
\label{app:skr:sec:density-flow}

Write the target generator in the form
\[
 L_t\phi=B_t\cdot\nabla\phi+\tfrac12 A_t:D^2\phi,
 \qquad A_t=\Sigma_t\Sigma_t^\top,
\]
where, for the present model,
\[
 B=\binom{0}{b},\qquad
 \Sigma=\begin{pmatrix}
 f&0\\ \rho\beta&\sqrt{1-\rho^2}\,\beta
 \end{pmatrix}.
\]
Let $W(x,l)=w(x)\omega(l)=\mathfrak w(x,l)$ and write
$p_0=W\mathsf r_0$, as in the regular-start assumption.
All positive-order derivatives of $\log W$ are bounded. Expanding
the forward operator gives exactly
\begin{equation}
 W^{-1}L_t^*(Wr)=\widetilde L_t r+c_t r,
 \qquad
 \widetilde L_t=\tfrac12 A_t:D^2+\widetilde B_t\cdot\nabla,
 \label{app:skr:eq:weighted-forward-generator}
\end{equation}
where
\begin{equation}
 \widetilde B_t=\operatorname{div}A_t+A_t\nabla\log W-B_t,
 \qquad c_t=W^{-1}L_t^*W.
 \label{app:skr:eq:weighted-forward-coefficients}
\end{equation}
Here $(\operatorname{div}A)_j=\sum_i\partial_i A_{ij}$.
The functions
$\widetilde B,c,\Sigma$ and all their required derivatives are
bounded.

\begin{lemma}[Propagation of the relative density bounds]
The target density is $p_t=W \mathsf r_t$, where
\[
 0<c_T\le \mathsf r_t(z)\le C_T,
 \qquad
 \sup_{0\le t\le T,z}|\partial_t^j\partial_z^\alpha \mathsf r_t(z)|
       <\infty
\]
for every fixed j and spatial multi-index $\alpha$.
\end{lemma}
\begin{proof}
For fixed t and starting point z, run the ordinary auxiliary SDE
\[
 dZ_s=\widetilde B_{t-s}(Z_s)\,ds+\Sigma_{t-s}(Z_s)\,d\widetilde W_s,
 \qquad Z_0=z,\quad 0\le s\le t.
\]
The coefficient time is reversed; the spatial process is an ordinary
forward SDE. Define
\begin{equation}
 \mathsf r_t(z)=\mathbb E_z\!\left[
       \mathsf r_0(Z_t)\exp\left(\int_0^t c_{t-s}(Z_s)\,ds\right)\right].
 \label{app:skr:eq:weighted-density-fk}
\end{equation}
Since $c$ is bounded and $\mathsf r_0$ is bounded above and below by
positive constants, the same is true of $\mathsf r_t$, uniformly on
$[0,T]$.

Only ordinary SDE variations are needed for its derivatives. The
first derivative of $Z_s$ with respect to z satisfies the linear
equation obtained by differentiating the bounded smooth coefficients.
Every higher derivative satisfies the same linear equation with a
forcing made of lower derivatives. Burkholder--Davis--Gundy and
Gronwall therefore give every required moment of these derivatives,
uniformly over $t\le T$ and z. Differentiating
\eqref{app:skr:eq:weighted-density-fk}, including its time-integrated
bounded potential, yields the claimed uniform spatial bounds.

For clarity, the reversed coefficient time in
\eqref{app:skr:eq:weighted-density-fk} produces the forward equation
\[
 \partial_t \mathsf r_t=(\widetilde L_t+c_t)\mathsf r_t,
 \qquad \mathsf r_{t=0}=\mathsf r_0.
\]
One can verify it directly. The Markov property splits the first
short interval of the auxiliary experiment for $\mathsf r_{t+h}$ and leaves
$\mathsf r_t$ as its terminal function. It\^o's formula on this first
interval gives
\[
 \mathsf r_{t+h}-\mathsf r_t=h(\widetilde L_t+c_t)\mathsf r_t+o(h).
\]
SDE stability supplies continuity in t, also after each fixed spatial
differentiation. The established bounds make the remainder and right
derivative locally uniform.
The right derivative is continuous, hence is the time derivative.
Repeated use of this equation gives all fixed time-derivative bounds,
including right derivatives at zero.

Equation \eqref{app:skr:eq:weighted-forward-generator} now shows
that $p_t=W \mathsf r_t$ solves the target forward equation with initial
density $p_0$. It is the density of the target SDE: test this forward
equation against the backward expectations of a smooth compactly
supported terminal function. Their backward It\^o equation makes
the integral constant in time, identifying it with the target
expectation. Integration by parts is justified by the exponential
tails of W and the established derivative bounds. This also proves
mass one without a separate assumption on the auxiliary process.
\end{proof}

Uniform ellipticity is not needed for this density estimate. The
positive smooth initial ratio is propagated by an ordinary diffusion
with a bounded multiplicative weight.

\subsection{Composition along the smooth SLV target}
\label{app:skr:sec:iteration}

\begin{lemma}[An explicit smooth comparison slope]
\label{app:skr:lem:smooth-predictor}
Under the coefficient and regular-start assumptions of the convergence
theorem, put $p_t=p(t,\cdot)$, $k_0(t,x)=1/\lambda(t,x)$, and
\[
 \begin{aligned}
 M_t(x)&=\int\sigma_R(t,x,l)^2p_t(x,l)\,dl,
 &J^\star_{t,h}&=\frac{\operatorname{Call}(\mu_{t+h})
                  -\operatorname{Call}(\mu_t)}h,\\
 E^0_{t,h}&=J_h(p_t,k_0(t))-J^\star_{t,h},
 &D^0_t&=\frac{M_t}{k_0(t,\cdot)^3}.
 \end{aligned}
\]
Define the known slope pointwise by
\begin{equation}
 k^{\rm app}_{t,h}=k_0(t)+E^0_{t,h}/D^0_t.
 \label{app:skr:eq:explicit-slope-predictor}
\end{equation}
It is positive for small $h$. For each fixed required finite $r$,
\begin{equation}
 \|k^{\rm app}_{t,h}-k_0(t)\|_{C_b^r}
 \le C_{r,T}h,\qquad 0\le t\le T-h.
 \label{app:skr:eq:predictor-smooth-bounds}
\end{equation}
Write
\[
 q^{\rm app}_{t,h}=\mathcal O_{t,h}(p_t,k^{\rm app}_{t,h}),\qquad
 m_{\rm app}=(q^{\rm app}_{t,h})^X/k^{\rm app}_{t,h}.
\]
Then
\begin{equation}
 \|q^{\rm app}_{t,h}-p_{t+h}\|_E\le C_Th^2,\qquad
 \left\|L_h\frac{J_h(p_t,k^{\rm app}_{t,h})-J^\star_{t,h}}
                         {m_{\rm app}}\right\|_{K_h}
 \le C_Th^2.
 \label{app:skr:eq:predictor-defects}
\end{equation}
The second expression is a prescribed target residual at the known
base $(p_t,k^{\rm app}_{t,h})$.
\end{lemma}
\begin{proof}
The target density and coefficient bounds give $D^0_t\asymp w$;
$D^0_t/w$ and its reciprocal have bounded derivatives.
Known smooth calendar consistency and the target forward equation
give $E^0_{t,h}=O(h)w$ in every required finite relative spatial
norm.
The leading terms agree because $k_0=1/\lambda$.
Thus \eqref{app:skr:eq:predictor-smooth-bounds} follows by division,
and positivity follows from the positive lower bound for $k_0$.

Set $d=E^0_{t,h}/D^0_t=O(h)$ in the required smooth norms.
The fundamental theorem of calculus and
\eqref{app:skr:eq:directional-calendar-consistency} give
\[
 \begin{aligned}
 J_h(p_t,k_0+d)-J^\star_{t,h}
 &=E^0_{t,h}+\int_0^1D_kJ_h(p_t,k_0+sd)[d]\,ds\\
 &=E^0_{t,h}-D^0_td+O(h)\,w\,\|d\|_{C_b^r}
                       +O(1)\,w\,\|d\|_{C_b^r}^2
 =O(h^2)w.
 \end{aligned}
\]
Here the last quadratic term comes simply from
$M_t/(k_0+sd)^3-D_t^0=O(h)w$ in the same smooth norms;
no second derivative of the calendar map is needed.
Take the output order at least two higher than needed for the
preconditioner. The known-base endpoint estimates give
$m_{\rm app}\asymp w$ with the same required smooth relative bounds.
Division and $L_h=a_*I-h\partial_{xx}$ therefore give an $O(h^2)$
residual in both $L^2_{\rm uloc}$ and $L^\infty$, hence in $K_h$.

The first-order strong endpoint expansion applied directly at the
known smooth slope $k^{\rm app}_{t,h}$ gives
\[
 q^{\rm app}_{t,h}
 =p_t+h\mathcal G_t(k^{\rm app}_{t,h})^*p_t+O(h^2).
\]
Since $k^{\rm app}_{t,h}-k_0=O(h)$ in the required smooth norms,
$\mathcal G_t(k^{\rm app}_{t,h})^*p_t
 -\mathcal G_t(k_0)^*p_t=O(h)$ in $E$.
The target forward equation has generator $\mathcal G_t(k_0)$,
so this proves the first bound in
\eqref{app:skr:eq:predictor-defects}.
No implicit regularity theorem for an optimizing slope is used.
\end{proof}

\begin{corollary}[Quantitative calibration]
\label{app:skr:prop:quantitative-calibration}
Under the assumptions of Theorem~\ref{app:skr:thm:main}, let
$h=T/N$, $t_j=jh$, and let $p_j$ be the joint density at the $j$th
SKR endpoint. For sufficiently small $h$,
\begin{equation}
 \max_{0\le j\le N}\|p_j-p(t_j)\|_E\le C_T h,
 \qquad
 \max_{0\le j<N}\|k_j-1/\lambda(t_j,\cdot)\|_{K_h}\le C_T h.
 \label{app:skr:eq:kclose}
\end{equation}
Here $k_j$ is the positive inverse slope of the exact Bass update.
In particular its $L^2_{\rm uloc}$ error is $O(h)$ and its uniform
error is $O(h^{3/4})$.
\end{corollary}
\begin{proof}
At step j use the known smooth base
$b_j=(p(t_j),k^{\rm app}_{t_j,h})$ from
Lemma~\ref{app:skr:lem:smooth-predictor}. Write
$q_j^{\rm app}=\mathcal S_0(p(t_j))$ for its own endpoint and set
\[
 R_j=L_h\frac{J_h(p(t_j),k^{\rm app}_{t_j,h})-J^\star_{t_j,h}}
                         {m_{{\rm app},j}}.
\]
The predictor gives $\|R_j\|_{K_h}\le Ch^2$ and
$\|q_j^{\rm app}-p(t_{j+1})\|_E\le Ch^2$.
The joint local branch of Lemma~\ref{app:skr:lem:localscalar}
gives $\mathcal S_{R_j}$, uniformly for incoming
$P=p(t_j)+z$ in its shrinking neighborhood. Its bounded first
residual derivative implies
\begin{equation}
 \|\mathcal S_{R_j}(P)-\mathcal S_0(P)\|_E\le Ch^2.
 \label{app:skr:eq:residual-base-displacement}
\end{equation}
The equation $\mathcal F_{b_j}(P,k)+R_j=0$ says exactly
$L_h((J_h(P,k)-J^\star_{t_j,h})/m_{{\rm app},j})=0$.
The exponential calendar tails exclude its growing homogeneous
modes, so this branch has price marginal $\mu_{t_{j+1}}$.
For nonnegative P the positive endpoint kernel and the WMOT identity
in the standalone theorem identify it as the unique optimizer.

The explicit norm is defined on all density differences. Its triangle
inequality, the standalone theorem for the known base's own target,
and the two $O(h^2)$ defects give
\begin{equation}
 \begin{aligned}
 \mathcal N_{b_j}(\mathcal S_{R_j}(P)-p(t_{j+1}))
 &\le\mathcal N_{b_j}(\mathcal S_{R_j}(P)-\mathcal S_0(P))\\
 &\quad+\mathcal N_{b_j}(\mathcal S_0(P)-\mathcal S_0(p(t_j)))
       +\mathcal N_{b_j}(q_j^{\rm app}-p(t_{j+1}))\\
 &\le(1+Ch)\mathcal N_{b_j}(P-p(t_j))+Ch^2.
 \end{aligned}
 \label{app:skr:eq:predictor-step-decomposition}
\end{equation}

Starting from $p_0=p(0)$, iterate these positive calibrated maps.
Put $e_j=\mathcal N_{b_j}(p_j-p(t_j))$ for $j<N$, and use the last
step's physical norm at $j=N$. Thus no reference interval beyond T
is required. All these differences have both marginals zero.
The bounded time derivative of $\log(p_t\sigma_R(t)^2)$ gives
metric comparison at cost $1+Ch$ between consecutive steps. This
compares the same physical density difference and uses no time
regularity of the predictor slope. As long as $e_j\le h^{3/4}$,
Theorem~\ref{app:skr:thm:wmot-calibration} and
\eqref{app:skr:eq:predictor-step-decomposition} give
\begin{equation}
 e_{j+1}\le(1+Ch)e_j+Ch^2,\qquad e_0=0.
 \label{app:skr:eq:density-recursion}
\end{equation}
Discrete Gronwall gives $e_j\le A_T h$. For sufficiently small h,
$A_T h<\tfrac12h^{3/4}$, so induction at the first possible exit
closes on every fixed finite horizon.

The local optimizer identity identifies every constructed iterate
with the actual SKR update. The scalar branch bound gives
\[
 \|k_j-k_0(t_j)\|_{K_h}
 \le C\bigl(\|p_j-p(t_j)\|_E+\|R_j\|_{K_h}\bigr)
       +\|k^{\rm app}_{t_j,h}-k_0(t_j)\|_{K_h}
 \le C_T h.
\]
Norm equivalence proves \eqref{app:skr:eq:kclose}.
\end{proof}

\begin{proof}[Proof of \Cref{app:skr:thm:main}]
Quantitative calibration gives
$\|k_j-1/\lambda(t_j,\cdot)\|_\infty\le C_T h^{3/4}$.
The reference conditional second moment is $Ch$, so
Lemma~\ref{app:skr:lem:centered-slope-stability} yields
$\mathbb E|U_j-V_j|^2\le C_T h^{5/2}$.
Proposition~\ref{app:skr:prop:endpoint-criterion} now gives
$\mathbb E\sup_{t\le T}|X_t^h-X_t|^2\le C_T h$.
\end{proof}

The convergence statement retains the regular-start and coefficient
assumptions of Theorem~\ref{app:skr:thm:main}. The autonomous-price
corollary permits any square-integrable initial pair, including a
deterministic initial price. A general Dirac-start convergence result
and existence of an SLV limit for arbitrary prescribed marginal curves
are not asserted. Composition uses the standalone WMOT
stability theorem, an explicit smooth comparison slope and first-order
consistency; no regularity theorem for an exactly calibrated
comparison optimizer is required.
\endgroup

\bibliography{references}
\bibliographystyle{plainnat}
	
\end{document}